\documentclass[11pt]{article}
\usepackage[margin=1in]{geometry}
\usepackage[utf8]{inputenc}
\usepackage{amsmath, amssymb, amsthm, braket, pxfonts, tikz, bbm, algpseudocode, mathrsfs, mathtools}

\DeclarePairedDelimiter\floor{\lfloor}{\rfloor}

\newcommand{\noise}{\gamma}

\newcommand{\ER}{Erd{\"o}s-R{\'e}nyi }

\newcommand{\Base}{B} 

\newcommand{\Inter}{G} 

\newcommand{\cH}{\mathcal{H}}
\newcommand{\cS}{\mathcal{S}}
\newcommand{\cG}{\mathcal{G}}
\newcommand{\cD}{\mathcal{D}}
\newcommand{\cE}{\mathcal{E}}

\newcommand{\cF}{\mathcal{F}}
\newcommand{\num}[2]{X_{#1}(#2)}

\newcommand{\dnull}{\cD_{\text{\tiny null}}}
\newcommand{\dstruct}{\cD_{\text{\tiny struct}}}

\newcommand{\eps}{\epsilon}
\newcommand{\aut}{\operatorname{aut}}
\renewcommand{\emptyset}{\varnothing}

\newcommand{\isH}[2]{\mathbf{1}_{#1 \supseteq #2}}

\newcommand{\density}{{(1-\delta)^{-1}}}

\usepackage{paralist}
\usepackage[hidelinks,colorlinks,linkcolor=blue]{hyperref}
\usepackage{amsthm}
\usepackage{thmtools,thm-restate}

\newtheorem{theorem}{Theorem}[section]
\newtheorem*{theorem*}{Theorem}

\newtheorem{proposition}[theorem]{Proposition}
\newtheorem*{proposition*}{Proposition}
\newtheorem{lemma}[theorem]{Lemma}
\newtheorem*{lemma*}{Lemma}

\newtheorem*{conjecture*}{Conjecture}
\newtheorem{fact}[theorem]{Fact}
\newtheorem*{fact*}{Fact}

\newtheorem*{hypothesis*}{Hypothesis}
\newtheorem{conjecture}[theorem]{Conjecture}

\theoremstyle{definition}
\newtheorem{definition}[theorem]{Definition}

\newtheorem{algorithm}[theorem]{Algorithm}
\newtheorem{problem}[theorem]{Problem}

\theoremstyle{remark}
\newtheorem{claim}[theorem]{Claim}
\newtheorem*{claim*}{Claim}

\newtheorem*{remark*}{Remark}
\newtheorem{observation}[theorem]{Observation}
\newtheorem*{observation*}{Observation}

\usepackage[capitalise,nameinlink]{cleveref}
\crefname{lemma}{Lemma}{Lemmas}
\crefname{fact}{Fact}{Facts}
\crefname{theorem}{Theorem}{Theorems}
\crefname{corollary}{Corollary}{Corollaries}
\crefname{conjecture}{Conjecture}{Conjectures}
\crefname{claim}{Claim}{Claims}
\crefname{example}{Example}{Examples}
\crefname{algorithm}{Algorithm}{Algorithms}
\crefname{problem}{Problem}{Problems}
\crefname{definition}{Definition}{Definitions}
\usepackage{paralist}

\newcommand{\ff}[2]{#1^{\underline{#2}}}

\newcommand{\cc}[1]{\operatorname{\mathcal{#1}}}

\newcommand{\bb}[1]{\operatorname{\mathbb{#1}}}

\DeclareMathOperator*{\E}{\mathbb{E}}
\newcommand{\Var}{\mathbb{V}}

\newcommand{\Ind}{\bb{I}}

\title{(Nearly) Efficient Algorithms for the Graph Matching Problem on Correlated Random Graphs}

\author{Boaz Barak\thanks{Harvard University. Supported by NSF awards CCF 1565264 and CNS 1618026, and the Simons Foundation. Emails: \texttt{b@boazbarak.org}, ~\texttt{chiningchou@g.harvard.edu}, ~\texttt{leizhixian.research@gmail.com}, ~\texttt{tselil@seas.harvard.edu}, ~\texttt{ysheng@g.harvard.edu}.} \and Chi-Ning Chou\footnotemark[1] \and Zhixian Lei\footnotemark[1]  \and Tselil Schramm\footnotemark[1] \and Yueqi Sheng\footnotemark[1]}

\begin{document}

\maketitle

\begin{abstract}
    We give a quasipolynomial time algorithm for the \textit{graph matching problem} (also known as \textit{noisy} or \textit{robust} graph isomorphism)  on correlated random graphs. 
Specifically, for every $\gamma>0$, we give a $n^{O(\log n)}$ time algorithm that given a pair of $\gamma$-correlated \ER graphs $G_0,G_1$ with average degree between $n^{o(1)}$ and $n^{1/153}$, recovers the "ground truth" permutation $\pi\in S_n$ that matches the vertices of $G_0$ to the vertices of $G_n$ in the way that minimizes the number of mismatched edges.
We also give a recovery algorithm for a denser regime, and a polynomial-time algorithm for distinguishing between correlated and uncorrelated graphs.

Prior work showed that recovery is information-theoretically possible in this model as long the average degree was at least $\log n$, but truly subexponential-time algorithms were only known for graphs with average degree $\Omega(n^2)$.
    In contrast our algorithms succeed for average degree as low as $n^{o(1)}$.
\end{abstract}
\thispagestyle{empty}
\clearpage
\setcounter{page}{1}

\section{Introduction}

The \textit{graph matching problem} is a well-studied computational problem in a great many areas of computer science.
Some examples include machine learning~\cite{cour2007balanced}, computer vision~\cite{cho2012progressive},  pattern recognition~\cite{berg2005shape}, computational biology~\cite{singh2008global,vogelstein2011large}, social network analysis~\cite{korula2014efficient}, and de-anonimzation~\cite{narayanan2009anonymizing}.\footnote{See the surveys~\cite{livi2013graph,conte2004thirty}, the latter of which is titled ``Thirty Years of Graph Matching in Pattern Recognition''.}
The graph matching problem is the task of computing, given a pair $(G_0,G_1)$ of $n$ vertex graphs, the permutation
\begin{equation}
\pi^* = \arg\min_{\pi \in \cS_n} \|G_0 - \pi(G_1)\|  \label{eq:matchingperm}
\end{equation}
where we identify the graphs with their adjacency matrices, and write $\pi(G_1)$ for the matrix obtained by permuting the rows and columns according to $\pi$ (i.e., the matrix $P^\top G_1 P$ where $P$ is the permutation matrix corresponding to $\pi$).

\subsection{The Correlated \ER model}

The graph matching problem  can be thought of as a noisy (and hence harder) variant of the  \textit{graph isomorphism} problem.
In fact, the graph matching problem is NP hard in the worst case.\footnote{If we allow weights and self-loops it is equivalent to the \emph{quadratic assignment problem}~\cite{lawler1963quadratic,Burkard1999}.} 
O'Donnell et al. also show that graph matching is hard to approximate assuming Feige's Random 3SAT hypothesis~\cite{o2014hardness}. 
Hence, much of the existing work is focused on practical heuristics or specific generative models.
In a 2011 paper, Pedarsani et al. \cite{pedarsani2011privacy} introduced the \emph{correlated \ER model} as a case study for a de-anonymization task. 
This is the $\bb{G}(n,p;\noise)$ model in which the pair $(G_0,G_1)$ is generated as follows:\footnote{Some works also studied a more general variant where $G_0$ and $G_1$ use different subsampling parameters $\noise_0,\noise_1$.}

\begin{compactitem}
\item We sample a ``base graph'' $\Base$ from the \ER distribution $\bb{G}(n,p)$.
\item We sample $\pi$ at random in $\cS_n$ (the set of permutations on the elements $[n]$).
\item We let $G_0$ be a randomly subsampled subgraph of $\Base$ obtained by including every edge of $\Base$ in $G_0$ with probability $\noise$ independently.
\item We let $G_1$ be an independently subsampled subgraph of $\pi(\Base)$ obtained by including every edge of $\pi(\Base)$ in $G_1$ with probability $\noise$ independently.
\end{compactitem}
Given $(G_0,G_1)$, our goal is to recover $\pi$.
Though initially introduced as a toy model for a specific application, the problem of recovering $\pi$ in $\bb{G}(n,p;\gamma)$ is a natural and well-motivated statistical inference problem, and it has since received a lot of attention in the information theory and statistics communities (c.f.,~\cite{yartseva2013performance, korula2014efficient,lyzinski2014seeded,kazemi2015growing,cullina2016improved,cullina2017exact,mossel2018seeded}).

Below we will use $\dstruct(n,p;\noise)$ (or $\dstruct$ for short, when the parameters are clear from the context) to denote the ``structured'' joint distribution above on triples $(G_0,G_1,\pi)$ of pairs of graphs and a permutation $\pi$ such that $G_1$ is a noisy version of $\pi(G_0)$.
One can see that the graphs $G_0$ and $G_1$ are individually distributed according to the \ER distribution $\bb{G}(n,p\noise)$, but there is significant correlation between $G_0$ and $G_1$.
It can be shown that as long as $p\noise^2 \gg \log n/n$, the permutation $\pi$ will be the one that minimizes the right-hand side of (\ref{eq:matchingperm}), and hence it is possible to recover $\pi$ information theoretically.
Indeed, Cullina and Kivayash~\cite{cullina2016improved,cullina2017exact}  precisely characterized the parameters $p,\noise$ for which information theoretic recovery is possible.
Specifically, they showed recovery is possible if $p\noise^2 > \tfrac{\log n + \omega(1)}{n}$ and impossible when $p\noise^2 < \tfrac{\log n - \omega(1)}{n}$.

However, none of these works have given \textit{efficient} algorithms.
Yartseva and Grossglauser~\cite{yartseva2013performance} analyzed a simple algorithm known as \emph{Percolation Graph Matching (PGM)}, which was used successfully by Narayanan  and Shmatikov~\cite{narayanan2009anonymizing} to de-anonymize many real-world networks.
(Similar algorithms were  also analyzed by \cite{korula2014efficient,kazemi2015growing,lyzinski2014seeded}.)
This algorithm starts with a "seed set" $S$ of vertices in $G_0$ that are mapped by $\pi$ to $G_1$, and for which the mapping $\pi|_S$ is given.
It propagates this information according to a simple percolation, until it recovers the original permutation.
Yartseva and Grossglauser gave precise characterization of the size of the seed set required as a function of $p$ and $\noise$~\cite{yartseva2013performance}.
Specifically, in the case that $\noise = \Omega(1)$ and $p=n^{-1+\delta}$ (where the expected degree of $G_0$ and $G_1$ is $\Theta(n^\delta)$), the size of the seed set required is $|S|=n^{1-\delta - \Theta(\delta^2)}$.
In the general setting when one is  \textit{not} given such a seed set, we would require about $n^{|S|}$ steps to obtain it by brute force, which yields an  $\exp(n^{\Omega(1)})$ time algorithm in this regime.
Lyzinski et al.~\cite{Lyzinskiconvexrisk} also gave negative results for popular convex relaxations for graph matching on random correlated graphs.

\paragraph{Subsequent work.} Subsequent to our work, Mossel and Xu \cite{mossel2018seeded} obtained new algorithms for the seeded setting based on a delicate analysis of local neighborhoods.
Notably, they achieve recovery at the information-theoretic threshold.
Though the exact dependence of the seed set size on the average degree is complicated to state, roughly speaking whenever $pn = n^{\delta}$ for $\frac{1}{\delta} \in \mathbb{N}$, their seed set has size $O(\log n)$, giving quasi-polynomial time algorithms.
However, when $\frac{1}{\delta} \not\in\mathbb{N}$, the seed set size is $n^{\Omega(1)}$ with the constant in the exponent depending on the difference of $\frac{1}{\delta}$ and $\lfloor \frac{1}{\delta}\rfloor$.

\begin{figure}
{\small
\center
\begin{tabular}{|c|c|c|c|}
\hline
{\bf Paper } & {\bf Algorithm} & {\bf Requirements} & {\bf Runtime, } $pn = n^{\delta}$\\
\hline
Cullina \& Kivayash&
info-theoretic & 
$p\gamma^2 n \ge \log n + \omega(1)$ &
$\exp(O(n))$\\
\cite{cullina2016improved,cullina2017exact} & &&\\
\hline
Yartseva \& Grossglauser&
percolation & 
$p\gamma^2 n = \Omega(\log n)$ &
$\exp(n^{1-\delta - \Theta(\delta^2)})$\\
\cite{yartseva2013performance} & 
&&\\
\hline
{\bf This paper} &
subgraph & 
$pn \in [n^{o(1)}, n^{1/153}] \cup [n^{2/3}, n^{1-\epsilon}]$ &
 \\
&matching &$\gamma \ge (1/\log n)^{o(1)}$&$n^{O(\log n)}$\\
\hline
Mossel \& Xu&
seeded local & 
$pn \ge \log n + \omega(1)$&
$n^{O(\log n)}$ if $\frac{1}{\delta} \in \mathbb{N}$, \\
\cite{mossel2018seeded} &statistics &$\gamma = \Theta(1)$ &else $\exp(n^{\Omega(1)})$ \\
\hline
\end{tabular}
\caption{A comparison of algorithms for recovery of the permutation in the correlated \ER model, when $(G_0,G_1,\pi) \sim \dstruct(n,p;\gamma)$.
}
}
\end{figure}

\subsection{Our results}

In this work we give quasipolynomial time algorithms for recovering the hidden permutation $\pi$ in the $\bb{G}(n,p;\noise)$ model for every constant (and even slightly sub-constant) $\noise$ and  a wide range of $p$.

\begin{theorem}[Recovery] \label{thm:recovery:intro}
    For every $\eps>0$ and $\gamma > 0$, if $pn \in [n^{o(1)},n^{1/153}]$ or $pn \in [n^{2/3},n^{1-\epsilon}]$, then there is a quasipolynomial-time randomized algorithm $A$ such that with high probability over $(G_0,G_1,\pi) \sim \dstruct(n,p;\noise)$ and over the choices of $A$, $A(G_0,G_1)=\pi$.
\end{theorem}

One can see that we obtain (nearly) efficient recovery even for sub-polynomial degrees.
As discussed in \cref{sec:recovery}, our results are more general and handle (slightly) sub-constant noise $\noise$.
See \cref{thm:recovery} for a precise statement of the parameters (including the $o(1)$ in the minimum sparsity $np = n^{o(1)}$).
To the best of our knowledge, the best previously known algorithms for any $pn<n^{1-\epsilon}$ required subexponential (i.e., $\exp(n^{\Omega(1)})$) time.

\medskip

At first, the requirement that the average degree $pn$ be in a union of two disjoint intervals may seem strange. 
Indeed, modulo a combinatorial conjecture, our algorithm works for all values of $np \in [\Omega(\log n),1)$.
In order to give this conjecture, we need the following definition; for the sake of exposition, we have pared it down. 
For the full requirements, see \cref{thm:graph family}.
\begin{definition}[simplified version]
Let $v,e$ be positive integers.
We say that $\cH$ is a {\em $(v,e)$-test family} if $\cH$ is a set of $v$-vertex $e$-edge graphs, such that each $H \in \cH$ has no non-trivial automorphisms, every strict subgraph of $H$ has edge density $< \frac{e}{v}$, and further for pairs of distinct $H,H' \in \cH$, no shared subgraph of $H,H'$ has density larger than $.99\frac{e}{v}$. Finally, we also require $|\cH|\ge v^{\Omega(e)}$.\footnote{
Notice that there are only $\binom{v^2/2}{e}$ graphs on $v$ vertices and $e$ edges, so the size requirement on $\cH$ is quite stringent.}
\end{definition}

\begin{conjecture}\label{conjecture:combinatorial}
For all sufficiently large integers $v > v_0$, for every integer $e$ such that $v + \log v < e \ll v^2$, there exists a $(v,e)$-test family.
\end{conjecture}

A proof of this conjecture would immediately extend \cref{thm:recovery:intro} to every $pn \in [\Omega(\log n), 1)$.
In fact, our proof of \cref{thm:recovery:intro} proceeds by establishing this conjecture for $v = \Theta(\log n)$ and $e \in [v + o(v), (1+\frac{1}{152})v] \cup [3v, O(v)]$.
We find it difficult to believe that the existence of a $(v,e)$-test family would be discontinuous in $e$ as a function of $v$; however our techniques for the two regimes are different, and while we did not make a special effort to optimize the constants $\frac{1}{152}$ or $3$, it seems that completely filling in the gap requires some delicate and technical combinatorial arguments.

\paragraph{Hypothesis testing.} We also consider the potentially easier ``hypothesis testing'' task of \textit{distinguishing} a pair of graphs $(G_0,G_1)$ sampled from $\dstruct(n,p;\noise)$  from a pair $(G_0,G_1)$ that is drawn from the ``null distribution'' $\dnull(n,p\noise)$ of two independent samples from $\bb{G}(n,p\noise)$.
For this problem we give a \textit{polynomial time} algorithm for a range of values of $p$.

\begin{theorem}[Distinguishing]\label{thm:distinguish:intro}
    For arbitrarily small $\epsilon>0$ and for every $\noise>0$, if $pn \in [n^\eps,n^{1/153}]$ or
	$pn \in [n^{2/3},n^{1-\epsilon}]$, then there is a \textit{pseudo-polynomial} time\footnote{ 
The algorithm is pseudo-polynomial because it depends on the bit complexity of $\frac{\log p}{\log n}$.}
deterministic algorithm $A$ that distinguishes with probability at least\footnote{We can amplify this to probability $1-\delta$, but this incurs a dependence on $\delta$ in the exponent of the runtime.} $0.9$ between the case that $(G_0,G_1)$ are sampled from $\dstruct(n,p;\noise)$ and the case that they are sampled from $\dnull(n,p\noise)$.
\end{theorem}

See \cref{thm:distinguish} for the full settings of parameters that we achieve for distinguishing.

\subsection{Approach and techniques}

In this section we illustrate our approach and techniques.
For the sake of simplicity and concreteness we first focus on the following task.
Given a pair of graphs $(G_0,G_1)$, distinguish between the following two cases for $p=n^{-1+\delta}$ (i.e., graphs of average degree $\sim n^{\delta}$):

\begin{compactdesc}

\item[Null case:] $(G_0,G_1)$ are drawn from the distribution $\dnull$ of two independent graphs from the \ER distribution $\bb{G}(n,p/2)$.

\item[Planted/structured case:] $(G_0,G_1)$ are drawn from the distribution $\dstruct(n,p;1/2)$. That is, we sample $\Base$ from $\bb{G}(n,p)$ and a random permutation $\pi \sim \cS_n$, and both $G_0$ and $G_1$ are independently subsampled subgraphs of $\Base$ where each edge is kept with probability $1/2$. The labels of the vertices of $G_1$ are additionally permuted according to $\pi$.
\end{compactdesc}

Before we present our approach to solve this problem, we explain some of the challenges.
In the $\dnull$ case the graphs $G_0,G_1$ are completely unrelated, and there is no permutation of the vertices so that $G_0$ and $G_1$ overlap on more than a $p$ fraction  of the edges, while in the $\dstruct$ case they are ``roughly isomorphic'', in the sense that  there is a permutation that will make them  agree on about a quarter of their edges.
Since random graphs are in fact an \emph{easy} instance of the graph isomorphism problem, we could perhaps hope that known graph isomorphism algorithms will actually succeed in this ``random noisy'' case as well.
Alas, it turns out not to be the case.

We now present some rough intuition why common graph isomorphism heuristics fail in our setting.
(If you are not interested in seeing why some algorithms fail but rather only why our approach succeeds, feel free to skip ahead to Section~\ref{sec:ourapproachsubsec}.)
Let's start with one of  the simplest possible heuristics for graph isomorphism: sort the vertices of $G_0$ and $G_1$ according to their degrees and then match them to each other.
If $G_0$ and $G_1$ are isomorphic via some permutation $\pi$ then it will of course be the case that the degree of every vertex $v$ of $G_0$ is equal to the degree of $\pi(v)$ in $G_1$.
Generally, even in a random graph, this heuristic will not completely recover the isomorphism since we will have many ties: vertices with identical degrees.
Nevertheless, this approach would map many vertices correctly, and in particular the highest degree vertex in a random graph is likely to be unique and so be mapped correctly.

However, in the \emph{noisy} setting, even the highest degree vertex is unlikely to be the same in both graphs.
The reason is that in a random graph of average degree $\Delta$, the degrees of all the vertices are roughly distributed as independent Poisson  random variable with expectation $\Delta$.
The vertex $v^*$ with highest degree in $G_0$ is likely to have degree which is  $k^* \sim \sqrt{\ln n}$ standard deviations higher than the mean.
But since the graphs are only $\tfrac{1}{4}$-correlated, the corresponding matched vertex $w^* = \pi(v^*)$ is likely to have degree which is only $\tfrac{1}{4}k^*$ higher than the mean.
It can be calculated that this means that $w^*$ is extremely unlikely to be the highest degree vertex of $G_1$.
In fact, we expect that about $n^{15/16}$ vertices will have degree larger than $w^*$s.

In the context of graph isomorphism algorithms, we often go beyond the degree to look at the \emph{degree profile} of a vertex $v$, which is the set of degrees of all the neighbors of $v$.
In the case that $G_0$ and $G_1$ are isomorphic via $\pi$, the degree profiles of $v$ and $\pi(v)$ are identical.
However, in the noisy case when $G_0$ and $G_1$ are only $\tfrac{1}{4}$-correlated, the degree profiles of $v$ and $\pi(v)$ are quite far apart.
About a quarter of the neighbors of $v$ and $\pi(v)$ will be matched, but for them the degrees are only roughly correlated, rather than equal.
Moreover the other three quarters of neighbors will not be matched, and for them the degrees in both graphs will just be independent Poisson variables.

Another common heuristic for  graph isomorphism is to match $G_0$ and $G_1$ by taking their top eigenvectors and sorting them (breaking ties using lower order eigenvectors).
Once again this will fail in our case, because even if the permutation was the identity, the top eigenvector of $G_0$ is likely to be very different from the top eigenvector of $G_1$.
This is for similar reasons as before: the top eigenvector of $G_0$ is the vector $v_0$ such that the quantity $v_0^\top A_0 v_0$ is  $k^*$ standard deviations higher than the mean for some particular value $k^*$.
However, it is likely that the $v_0^\top A_1 v_0$ will only be $k^*/4$ or so standard deviations higher than the mean, and hence $v_0$ will \emph{not} be the top eigenvector of $A_1$.

One could also imagine using a different heuristic, such as cycle counts, to distinguish---in \cref{sec:distinguish}, we discuss the shortcomings of such ``simple'' heuristics in detail.

\subsubsection{The ``black swan'' approach} \label{sec:ourapproachsubsec}

Now that we have appreciated the failure of the canonical graph isomorphism algorithms, we describe our approach.
Our approach can be thought of as \emph{``using a flock of black swans''}.
Specifically, suppose that $H$ is an $O(1)$-sized subgraph that is a ``black swan,'' in the sense that it has extremely low probability $\mu \ll 1$ of appearing as a subgraph of a random graph $G_0$ drawn from $\bb{G}(n,p/2)$.\footnote{For technical reasons, for the distinguishing section we actually take $\mu = O(1)$ and only use $\mu \ll 1$ for recovery, but we discuss here the case $\mu \ll 1$ for intuition.}
Another way to say it is that $\E \num{H}{G} = \mu$ where $\num{H}{G}$ is the {\em subgraph count} of $H$ in $G$, or the number of subgraphs of $G$ isomorphic to $H$.\footnote{More formally, $\num{H}{G}$ is the number of \textit{injective homomorphisms} of $H$ to $G$, divided by the number of {\em automorphisms} of $H$. That is,  if $H$ has vertex set $[v]$ and $G$ has vertex set $[n]$, then $\num{H}{G} = \frac{1}{|\aut(H)|}\sum_{\sigma:[v] \rightarrow [n] \text{1-to-1}}\prod_{(i,j) \in E(H)} G_{\sigma(i),\sigma(j)}$. \label{fn:homomorph}}

\begin{figure}[h]
    \centering
    \includegraphics[width=.25\textwidth]{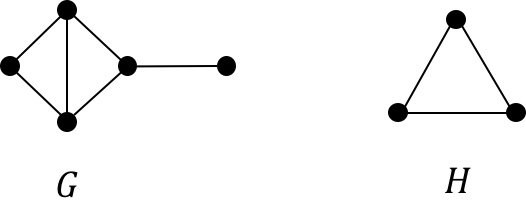}
    \caption{In this example, $\num{H}{G}=2$.}
    \label{fig:count}
\end{figure}

If we are lucky and $H$ appears in both $G_0$ and $G_1$, then we can conclude that it is most likely that the vertices of $H$ in $G_0$ are mapped to the vertices of $H$ in $G_1$, since the probability that both copies appear by chance is $\mu^2$.
This is much smaller than $\mu(\tfrac{1}{2})^{O(1)}$, which is the probability that $H$ appears in $G_0$ and those edges were not dropped in $G_1$.
If we are further lucky (or chose our swan carefully) so that $H$ has no non-trivial automorphism, then it turns out that we can in such a case deduce the precise permutation of the vertices of $H$ in $G_0$ to the vertices of $H$ in $G_1$.

The above does not seem helpful in designing an algorithm to recover the permutation, or even to distinguish between $\dnull$ and $\dstruct$ since by its nature as a ``black swan'', most of the times $H$ will \emph{not} appear as a subgraph of $G_0$, and hence we would not be able to use it.
Our approach is to use a \emph{flock} of such swans, which are a set $\cH$ of graphs $H_1,\ldots,H_t$ such that the probability of every individual graph $H_i$ occurring as a subgraph is very small, but the probability of \emph{some} graph $H_i$ occurring is very high.
We carefully designate properties of the family (which we call the \textit{test graph family}) so that when $i \neq j$ the events that $H_i$ and $H_j$ occur as subgraphs are roughly independent.

Already, this allows us to use the common occurrences of graphs in this family to deduce  whether $(G_0,G_1)$ came from the null distribution (in which case such occurrences will be rare) or whether they came from the structured distribution (in which case they will be more frequent).
In particular, if we define $\cH = \{ H_1,\ldots, H_t\}$, and the polynomial
$$
p_{\cH}(G_0,G_1) = \sum_{H\in \cH} (\num{H}{G_0}-\mu)(\num{H}{G_1}-\mu),
$$
then the value of $p_{\cH}$ will be noticeably higher when $(G_0,G_1)$ are drawn from $\dstruct$ than when they are drawn from $\dnull$.
This will result in an efficient algorithm to distinguish the two distributions.
We also use the ``swans'' for recovery, as we will discuss below.

\subsubsection{Constructing the flock of black swans}

It turns out that demonstrating the existence\footnote{We note that since the graphs in the family will be of size $v=O(\log n)$, and counting the number of occurrences of a graph on $v$ vertices takes time $n^{O(v)}$, once we fix $v$ we can perform brute-force enumeration over all graphs on $v$ vertices with negligible effect on the asymptotic runtime.
For this reason, demonstrating the existence of a family is enough. (The construction need not be algorithmic, though ours is).}
 of a family of ``swans,'' or  ``test graphs,'' is a delicate task, as we need to satisfy several properties that are in opposition to one another.
On one hand, we want the family to be \textit{large}, so that we can compensate for the fact that each member is a ``black swan'' and appears with very small probability.
On the other hand, we need each member of the family to have a \textit{small} number of edges.
Suppose that one of our swans, $H$, has $e$ edges.
If it appears in the base graph $\Base$, then it only survives in both $G_0$ and $G_1$ with probability $\noise^{2e}$.
That is, the correlation between $\num H {G_0}$ and $\num H {G_1}$ decays {\em exponentially} in the number of edges of $H$, and if $H$ is too large we cannot expect it to help us recover.
As a side effect, keeping $H$ small helps with the computational efficiency of the task of finding these occurrences.
A third constraint is that we need the events that each member of the family occurs to be roughly \textit{independent}.
It is very easy to come up with a large family of graphs for which the events of them co-occurring together are highly correlated, but such a family would not be useful for our algorithm.
Ensuring this independence amounts to obtaining control over the edge density of the common subgraphs that appear in pairs of distinct test graphs.
Luckily, we are able to demonstrate the existence of families of graphs achieving the desired properties, though this require some care.

The above discussion applies to the distinguishing problem of telling $\dnull$ and $\dstruct$ apart.
However, if we can ensure that the joint occurrences of our  family of test graphs \emph{cover} all the vertices of $G_0$ and $G_1$, then we can actually recover the permutation.
This underlies our \textit{recovery} algorithm.
To simultaneously ensure these conditions we need to make the number of vertices of each $H_i$ logarithmic rather than constant, which results in a \textit{quasipolynomial} time algorithm.

\paragraph{Properties of the test graphs.} We now describe more  precisely (though still not in full formality, see Section~\ref{sec:test-graphs}) the properties that our family $\cH$  of ``black swans'' or \emph{test graphs} needs to satisfy so the above algorithm will succeed:

\begin{compactdesc}

    \item[Low likelihood of appearing.] Each graph in our test family will have $v$ vertices and $e$ edges. To ensure that it is indeed a ``black swan'', we require that  $n^vp^e = \mu$ for $\mu$ slightly subconstant. In particular, in the regime $p=n^{-1+\delta}$ this will require $e\ge (1-\delta)^{-1}v$. In fact, we will set $e$ to be almost exactly $(1-\delta)^{-1}v$. Note that when, say, $\delta < 1/10$ these are graphs with less than, say $\tfrac{10}{9}\cdot v$ edges, so that the average degree is close to $2$.

\item[Strict balance.] This condition is well-known in the random graphs literature, and it ensures that the random variable $\num{H}{G}$ is well behaved. It states that for every $H$ in the family $\cH$, every induced subgraph $H'$ of $H$ with $v'$ vertices and $e'$ edges has strictly smaller edge density, $e'/v' < e/v$.
    We will actually require a strengthened, quantitative notion of strict balance, in which the density of $H'$ is related to its size.

\item[Intersection balance.] To ensure that for every pair of distinct graphs $H_1,H_2$ in our family the random variables $\num{H_1}{G}$ and $\num{H_2}{G}$ will be asymptotically independent when $\mu \ll 1$, we will need to have even tighter control over the density of their common subgraphs. We will require that for every two such graphs, any subgraph $H'$ of their intersection satisfies the stronger condition that $e'/v' < e/v - \alpha$ for some sufficiently large $\alpha>0$.

\item[No non-trivial automorphism.] To ensure we can recover the permutation correctly from an occurrence of $H$ in $G_0$ and an occurrence in $G_1$, we require that every $H$ in $\cH$ has no non-trivial automorphism.

\item[Largeness.] Finally to ensure that there actually will be many subgraphs from this family in our graph, we will require that $|\cH|\mu \noise^{e} > n$. (For distinguishing, it will suffice that $|\cH| \noise^{2e} =\Omega(1)$.)
\end{compactdesc}

We conjecture that a family achieving these properties can be obtained with any density $e/v>1$ (see \cref{conjecture:combinatorial}). 
However, at the moment we only demonstrate the existence of such families of graphs with certain densities, which is why our algorithms do not work for all ranges of $p$.\footnote{More accurately, we do have conjectured constructions that demonstrate the existence of such graphs for all densities, but have not yet been able to analyze them in all regimes.}

We now illustrate one such construction.
First and foremost, it can be shown that for integer $d \ge 3$, random $d$-regular graphs $H$ satisfy the requirements of strict balance and trivial automorphism group. Further, a sufficiently large fraction of the set of $d$-regular random graphs will satisfy the intersection balance property.
So for graphs with $e = (1-\delta)^{-1}v$ where $(1-\delta)^{-1} = d/2$, we easily have such a family.

However, the above construction does not give us graphs of all densities, and in particular does not allow us to handle the most interesting  regime of sparse \ER correlated graphs (e.g., $p < n^{-0.999}$ or so) which requires test graph of density roughly $1+\delta$ for some small $\delta$ and in particular a \textit{non integer} average degree of roughly $2+2\delta$.
Here is one example for such a construction when $\delta$ is $1/(3k+3)$ for some large integer $k$.
We start with a random $3$-regular graph $H'$ on $v'$ vertices (and hence $1.5v'$ edges).
We then subdivide every edge by inserting $k$ intermediate vertices into it, and so turning it into a path of length $k+1$.
The resulting graph $H$ will have $e=1.5v'(k+1)$ edges and  $v=v'+ke'=v'+1.5kv'$ vertices, and one can verify that $e=(1-\delta)^{-1}v$.
Moreover, it can be shown that the densest subgraphs of $H$ will ``respect'' the underlying structure, in the sense that for every original edge of $H'$, a subgraph of $H$ maximizing the density will either include all the corresponding path or none of it.
Using this observation, and the expansion properties of random graphs, it is possible to show that strict balance condition and even the intersection balance condition hold.
Moreover, we can also use known properties of random graphs to rule out non-trivial automorphism.
Finally, since the number of 3-regular graphs on $v'$ vertices is $v'^{\Omega(v')}$, for $v'=\Omega(v)$ we get a \emph{super exponential} (i.e., $2^{\omega(v)}$) number of graphs, which will allow us to get a sufficiently large family.
We will furthermore need to make the notion of strict balance quantitative. For this and the remaining details, see Section~\ref{sec:test-graphs}, where we also give our constructions for other values of $(1-\delta)^{-1}$.

\subsection{Related work}

As mentioned above, there is an extremely large body of literature on the graph matching problem. We discussed above the works on  correlated \ER graphs, but people have also studied other generative models such as power law graphs and others  (e.g., see \cite{ji2015your}).

On a technical level, our work is inspired by recent works on sum-of-squares, and using low degree polynomials for inference problems~\cite{hopkins2017power}.
In particular, our starting point is a low degree distinguisher from the planted and structured distributions.
However, there are some differences with the prior works.
These  works typically searched for objects such as cuts, vectors, or assignments that are less structured than searching for permutations.
Moreover, unlike prior works where the polynomial distinguishers used fairly simple polynomials (such as counting edges, triangles, cycles, etc..), we need to use subgraphs with more complex structure.
This is related to the fact that despite this inspiration, our algorithm at the moment is \textit{not} a sum-of-squares algorithm.
It remains an open problem whether the natural sum-of-squares relaxation for (\ref{eq:matchingperm}) captures our algorithm.

In our analysis we draw on the vast literature on analyzing the distribution of subgraph counts (e.g., see~\cite{Janson}).
Our setting  is however somewhat different as we need to construct a family of graphs with related but not identical  properties to those studied in prior works; in particular, some differences arise because of the fact that the graphs $G_0,G_1$ are correlated, and the fact that we work in the regime where the graphs have size growing with $n$ and appear $o(1)$ times in expectation.

\subsection{Organization}

In \cref{sec:distinguish} we give our \textit{distinguishing} algorithm between $\dnull$ and $\dstruct$ (\cref{thm:distinguish:intro}).
Then in \cref{sec:recovery} we build on this algorithm to obtain a \textit{recovery} algorithm that recovers the ``ground truth'' permutation from $(G_0,G_1)$ drawn from $\dstruct$.
This algorithm builds upon and extends the techniques of our distinguishing algorithm.
Both the recovery algorithm and distinguishing algorithms use as a ``black box'' the existence of families $\cH$ of ``test graphs'' (the black swans) that satisfy certain properties.
In Section~\ref{sec:test-graphs} we show how to construct such test families.

\paragraph{Notation.} For a graph $G = (V,E)$ and a subset of the vertices $S \subseteq V$, we use $G[S]$ to denote the vertex-induced subgraph on $S$ and we use $E[S]$ to denote the set of edges with both endpoints in $S$. We use $\Var(X)$ to denote the variance of the random variable $X$. For an event $\cE$, $\Ind[\cE]$ is the 0-1 indicator that $\cE$ occurs. For a two graphs $G,H$, we use $G \cong H$ to indicate that $G$ and $H$ are isomorphic, and $H \subseteq G$ to indicate that $G$ contains $H$ as an edge-induced subgraph. We use $\ff{n}{k}$ to denote the falling factorial, $\ff{n}{k} = n(n-1)\cdots(n-k+1)$. 
We will also use standard big-$O$ notation, and we will use $f(n) \ll g(n)$ to denote that $\lim_{n \to \infty}\frac{f(n)}{g(n)} \to 0$.

\section{Distinguishing the null and structured distributions} \label{sec:distinguish}

In this section, we give an algorithm for the following distinguishing problem:

\begin{problem}[Distinguishing]\label{prob:distinguishing}
We are given two $n$-vertex graphs $G_0,G_1$, sampled equally likely from one of the following distributions:
\begin{compactitem}
    \item The {\em null distribution}, $\dnull$: $G_0$ and $G_1$ are sampled independently from $\bb{G}(n,p\noise)$.
    \item The {\em   structured distribution}, $\dstruct$: First, a graph $\Base \sim \bb{G}(n,p)$ is sampled.
    Then, we independently sample $G_0,\tilde{G}_1$ from $G$ by subsampling every edge with probability $\noise$.
    Finally, we set $G_1$ to be a copy of $\tilde{G}_1$ in which the vertex labels have been permuted according to a  uniformly random permutation $\pi$.
\end{compactitem}
Our goal is to decide with probability $\ge 0.9$ whether $G_0,G_1$ were sampled from $\dnull$ or $\dstruct$.
\end{problem}

This section will be devoted to a proof of the following theorem, which is a generalization and directly implies Theorem~\ref{thm:distinguish:intro}:

\begin{theorem}[Distinguishing algorithm, restatement]\label{thm:distinguish}
    For arbitrarily small $\epsilon,\delta>0$, if $p\noise \in [\tfrac{n^\delta}{n},\tfrac{n^{1/153}}{n}]$ or $p\noise \in [n^{-1/3},n^{-\epsilon}]$ and if $\gamma = \Omega(\log^{-b} n)$ for constant $b$, there is a $n^{\gamma^{-O(1)}}$ time algorithm $A$ that distinguishes with probability at least\footnote{We can amplify this to probability $1-\delta$ by incurring extra runtime, gaining a dependence on $\delta$ in the exponent of $n$.} $0.9$ between the case that $(G_0,G_1)$ are sampled from $\dstruct(n,p;\noise)$ and the case that they are sampled from $\dnull(n,p\noise)$.

    In particular, if $\gamma = \Omega(1)$ then the algorithm runs in polynomial time.
\end{theorem}

Recall that for graphs $G,H$, we define the {\em subgraph count} $\num{H}{G}$ to be the number of subgraphs of $G$ isomorphic to $H$.
Since $G_0,G_1$ sampled from $\dstruct$ are correlated, if a subgraph $H$ appears in $G_0$ then it is more likely to also appear in $G_1$, and the subgraph counts are correlated.
The following lemma uses this approach to give  a certain ``test'': a polynomial $p_H(G_0,G_1)$ that has zero mean when $(G_0,G_1)$ is chosen from the null distribution, but positive mean when they are chosen from the structured distribution.
This test will not be good enough, since even in the structured case, it will be unlikely that the polynomial takes a non-zero value, but will serve as our starting point.

\begin{lemma}\label{lem:dist-exp}
    Let $H$ be a graph with $v$ vertices and $e$ edges,  define the {\em subgraph count-deviation correlation polynomial}
\[
    p_H(G_0,G_1) = \big(\num{H}{G_0}-\E[\num{H}{G_0}]\big)\big(\num{H}{G_1} -\E[\num{H}{G_1}]\big),
\]
where $G_0,G_1$ are two $n$ vertex graphs and the expectation is taken over $G_0,G_1$ from the \ER  distribution  $\bb{G}(n,p\noise)$.
Then  in the structured distribution,
\begin{equation}
    \E_{\dstruct(n,p;\noise)}[p_H(G_0,G_1)] = \Theta(1)\cdot \frac{\big(\E_{\dstruct}[\num{H}{G_0}]\big)^2}{\E_{G\sim \bb{G}(n,p)}[\num{K}{G}]}\label{eq:struct-expectation}
\end{equation}
where  $K \subseteq H$ is the subgraph  of $H$ which minimizes $\E_{G\sim\bb{G}(n,p)}[\num{K}{G}]$.
\end{lemma}

\begin{proof}
Note that under the null distribution, $G_0,G_1$ are independent.
Therefore
\begin{align*}
\E_{\dnull}[p_H(G_0,G_1)]
\,=\,
    \left(\E_{\dnull}\big[\num{H}{G_0} - \E[\num{H}{G_0}]\big]\right)^2
\,=\, 0.
\end{align*}

On the other hand, in the structured distribution, $G_0$ and $G_1$ are correlated.
That is,
\begin{align}
\E_{\dstruct}[p_H(G_0,G_1)]
    &= \E_{\dstruct}\big[\num{H}{G_0}\cdot \num{H}{G_1}\big] -\E[\num{H}{G_0}]\cdot\E[\num{H}{G_1}].\label{eq:cor}
\end{align}
    For an ordered subset of vertices $S \subset V(G_0)$ of size $v$, we define $\isH{S}{H}(G_0)$ to be the indicator that $G_0[S]$ contains $H$ as a labeled subgraph (at times we will drop the parameter $G_0$ for the sake of conciseness).
    Expanding $\num{H}{G_0}$ and $\num{H}{G_1}$ into sums of such indicators, we have
\begin{align}
\E_{\dstruct}\big[\num{H}{G_0}\cdot \num{H}{G_1}\big]
    &= \frac{1}{\aut(H)^2} \cdot\sum_{S_0 \in \ff{V(G_0)}{v}, S_1 \in \ff{V(G_1)}{v}} \E_{\dstruct}\left[\isH{S_0}{H}(G_0) \cdot \isH{S_1}{H}(G_1)\right],\label{eq:exp1}
\end{align}
    where we use $\ff{V(G_0)}{v}$ to denote all ordered subsets of $v$ vertices of $G_0$. The $\frac{1}{\aut(H)^2}$ is due to the fact that the sum is over ordered subset of $V(G_0)$ and $V(G_1)$ of size $v$ and thus it counts the number of labeled ordered copies of $H$ in $G_0$ as well as $G_1$. To avoid over-counting, we divide the number of automorphisms of $H$ and get the $\frac{1}{\aut(H)^2}$ factor.

    We recall that originally, we identified $V(G_0)$ and $V(G_1)$ both with the set $[n]$.
    For each summand, the value of the expectation is determined by the number of edges shared between the realization of $H$ on $S_0$ and the realization of $H$ on $\pi^{-1}(S_1)$, where $\pi$ is the random permutation we applied to the vertices of $G_1$.
    Without loss of generality, suppose that $\pi$ was the identity permutation (for notational convenience).
    Then let $\cE_H(S_0,S_1)$ be the number of edges in the intersection of $H$ as realized on $S_0$ and $S_1$ when both are identified with $[n]$.
    Then letting $\aut(H)$ be the number of automorphisms of $H$, we have
    \begin{align}
	\cref{eq:exp1}
	&= \frac{1}{\aut(H)^2} \cdot \sum_{k = 0}^{v} \sum_{\ell = 0}^e \sum_{\substack{S_0 \in \ff{V(G_0)}{v},\, S_1 \in \ff{V(G_1)}{v}\\ |S_0 \cap S_1| = k, \, \cE_H(S_0,S_1) = \ell}} \noise^{2e} p^{2e - \ell}.\label{eq:exp2}
\end{align}
We can more elegantly express this quantity as a sum over all subgraphs $J \subseteq H$, upon which the copy of $H$ on $S_0$ and the copy of $H$ on $S_1$ may intersect. So we may re-group the sum according to these unlabeled edge-induced subgraphs $J$ that give the intersection. Further, for each $J$ we can define the number $c_J(H)$ to be the number of ways one can obtain a graph by taking two ordered, labeled copies of $H$ and intersecting them on a subgraph isomorphic to $J$.

Specifically, to have such graphs with $J$ as an intersection, one must (a) choose a copy of $J$ in $H$ for the $G_0$ copy, (b) choose a copy of $J$ in $H$ for the $G_1$ copy, (c) choose an automorphism between the copies. 
Thus, for each subgraph $J$, we have $c_J(H) = \num{J}{H}^2 \cdot \aut(J)$.

Now, let us move from the summation over ordered subsets to the summation over unlabeled edge-induced subgraphs $J$, we have
\begin{align}
\E[\num{H}{G_0}\cdot \num{H}{G_1}] 
&= \frac{1}{\aut(H)^2} \sum_{J \subseteq H} c_J(H) \cdot \ff{n}{2v - |V(J)|} \cdot \gamma^{2e}\cdot p^{2e-|E(J)|},
\end{align}
as there are $c_J(H)$ ways of intersecting two copies of $H$ on the subgraph $J$, and for each such type of intersection there are $\ff{n}{2v - |V(J)|}$ choices of vertices for $S_0,S_1$.

    To finish off the proof, we observe that by following an identical sequence of manipulations we can re-write the squared expectation in the same manner,
    \begin{align*}
	\E_{\dstruct}\big[\num{H}{G_0}\big]\cdot\E_{\dstruct}\big[\num{H}{G_1}\big]
	&= \frac{1}{\aut(H)^2}\sum_{J \subseteq H} c_J(H) \cdot \ff{n}{2v - |V(J)|} \cdot \noise^{2e} p^{2e}.
    \end{align*}
The difference is of course that because the expectations were taken separately, the intersection has no effect on the exponent of $p$.
    This allows us to re-write \cref{eq:cor}:
    \begin{align}
	\cref{eq:cor}
	&=\frac{1}{\aut(H)^2} \sum_{\substack{J \subset H \\ |E(J)| \ge 1}} c_J(H) \cdot \ff{n}{2v - |V(J)|} \cdot \noise^{2e} p^{2e} \left(p^{-|E(J)|} - 
1 
\right),\label{eq:fin-expect}
    \end{align}
    where we have used that if $|E(J)| = 0$, the terms cancel.

    To obtain the final conclusion, we use that $p$ is bounded away from $1$, and that $c_J(H)$ and $\aut(H)$ are independent of $n$.
\end{proof}

\paragraph{The need for test sets.}
Lemma~\ref{lem:dist-exp} guarantees that the count-deviation polynomial $p_H$ has larger expected value under $\dstruct$ than $\dnull$.
However, this fact alone does not prove that $p_H$ is a distinguisher.
To illustrate this point, let us for simplicity suppose that we are in the regime where $n^v ({\noise}p)^e = C$ for $C$ constant.
In this case, \cref{lem:dist-exp} gives us that
\[
    \E_{\dstruct}\big[p_H(G_0,G_1)\big] \approx C \cdot {\noise}^e,
\]
up to lower-order terms (assuming that $H$ has no subgraph $K$ with $\E\num{K}{G} < \E\num{H}{G}$).
On the other hand, a simple calculation gives an optimistic bound on the standard deviation of $p_H$ under $\dnull$ of
\[
    \Var_{\dnull}\left(p_H(G_0,G_1)\right)^{1/2} \approx C.
\]
So the standard deviation in the null case is too large for us to reliably detect the offset expectation.

\medskip

Our solution is to identify a ``test set'' of graphs $\cH$, such that the estimators $p_H$ for $H \in \cH$ are close to independent.\footnote{Here we mean in the sense that the variance of their average is asymptotically equal to an average of independent estimators.}
If we had $|\cH| = T$ trials, intuitively we expect the standard deviation to decrease by a factor of $\sqrt{T}$.
So long as we satisfy
\[
    \frac{C}{\sqrt{T}} < C\cdot {\noise}^e,
\]
the variance in the null case may be sufficiently small that we reliably distinguish.

In order to translate this cartoon sketch into a reality, we will require some additional properties of our test set which will be crucial in controlling the variance.

\subsection{Test subgraphs for distinguishing}
Given a graph $H$ with $e$ edges and $v$ vertices, in expectation there are $\ff{n}{v} \cdot p^e/\aut(H)$ copies of $H$ in $\bb{G}(n,p)$.
Because of our prevailing intuition that random graphs are well-behaved, we might naively expect that the number of copies of $H$ is concentrated around its mean; however some simple examples demonstrate that this is not always the case.

\paragraph{Example: the need for balance.}
Consider for example the graph $H$ given by appending a ``hair'', or a path of length $1$, to a clique of size $4$ (see \cref{fig:necessity balance}).
$H$ has $5$ vertices, $7$ edges, and $3!$ automorphisms, so in $G(n,p)$ with $p = 2n^{-5/7}$, we have
\[
    \E_{G\sim \bb{G}(n,p)}[\num{H}{G}] = \frac{1}{3!}\cdot \ff{n}{5} \cdot p^7 = \frac{2^7}{3!} \cdot (1-o_n(1)).
\]
So in expectation, $H$ appears in $G$ $\tfrac{2^7}{3!}$ times.
\bigskip
\begin{figure}[h]
    \centering
    \includegraphics[width=0.15\textwidth]{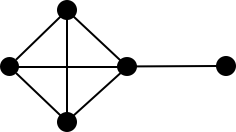}
    \caption{Necessity of being balanced. In this example, when $p = O(n^{-5/6})$, the expectation of the count of $H$ is $O(1)$, however, with high probability $H$ does not have a single occurrence in $G(n,p)$.}
    \label{fig:necessity balance}
\end{figure}
\bigskip
However, if we restrict our attention to the clique $K_4$ which is a subgraph of $H$, in expectation
\[
    \E_{G \sim \bb{G}(n,p)}[\num{K_4}{G}] = \frac{1}{4!}\cdot \ff{n}{4} \cdot p^6 = O(n^{-2/7}).
\]
So while the expected number of copies of $H$ is constant, $K_4$ is not expected to appear even once!
The large expected count of $H$ is due to a small-probability event; if $K_4$ appears (with polynomially small probability), then we will see many copies of $H$.

This issue is well known  in the study of random graphs (it is among the first topics in the textbook  of Janson et al. \cite{Janson}).
From that well-developed literature, we borrow the following concept:

\begin{definition}[Balanced graph]
The \textit{density} of a graph is the ratio of its edges to the vertices.

A graph $H$ with edge density $\alpha$ is called {\em balanced} if it has no strict subgraphs of density larger than $\alpha$.
If all strict subgraphs of $H$ have density strictly smaller than $\alpha$, then $H$ is called {\em strictly balanced}.
\end{definition}

If a graph $H$ is expected to appear at least once, then the {\em balancedness} of a graph $H$ is what determines whether the number of copies of $H$ in $\bb{G}(n,p)$ is well-concentrated around its mean.
For example, \cref{lem:dist-exp} already allows us the following observation:

\begin{observation}
    If $H$ is a graph of fixed size such that $\E_{G \sim \bb{G}(n,p)}[\num{H}{G}] = \Theta(1)$, then if $H$ is not balanced, $\Var_{G \sim \bb{G}(n,p)}[\num{H}{G}] = \omega(1)$.

    This follows from applying \cref{lem:dist-exp} with ${\noise} = 1$, in which case $\E_{\dstruct}[p_H(G_0,G_1)] = \Var_{G \sim \bb{G}(n,p)}[\num{H}{G}]$, and taking $K$ to be the densest subgraph of $H$, which must have $\E[\num{K}{G}] = o(1)$.
\end{observation}

To ensure asymptotic independence, we will require that each of the graphs n our test set $\cH$ be {\em strictly balanced}.

\begin{theorem*}[see \Cref{thm:graph family}]
    Let $\delta \in (0,1)$ be a rational number, so that $\tfrac{1}{1-\delta} \in (1,1+\tfrac{1}{153}] \cup \{\tfrac{3}{2},2,\tfrac{5}{2}\}$ or $\tfrac{1}{1-\delta} \ge 3$.
Let $v$ be a sufficiently large even integer.
    There exists a test set of graphs $\cH = \{H\}$, each on $v$ vertices and containing $e = \tfrac{1}{1-\delta}v$ edges, which satisfies the following properties:
\begin{compactenum}
    \item Every $H \in \cH$ is {\em strictly balanced}.
    \item Every $H \in \cH$ has no non-trivial automorphisms.
    \item $|\cH| \ge v^{c\frac{1}{1-\delta}v}$, for a constant $c$ independent of $v$.
\end{compactenum}
\end{theorem*}
We will prove this proposition in \cref{sec:test-graphs}; for now, assuming we have such a family, we proceed to prove \cref{thm:distinguish:intro}.

\subsection{Concentration for distinguisher}

We are now ready to prove that there is a poly-time computable distinguishing polynomial with bounded variance.
The existing results on concentration of subgraph counts is not sufficient for us here, because we are interested in the setting of {\em correlated} graphs.
We will bound the variance directly.
\begin{theorem}\label{thm:variance}
    Suppose that $p = n^{\delta - 1}$ for $\delta \in (0,1)$ with $\tfrac{1}{1-\delta} \in (1,1+\tfrac{1}{153}] \cup \{\tfrac{3}{2},2,\tfrac{5}{2}\}$ or $\tfrac{1}{1-\delta} \ge 3$.
    Then there exists a polynomial $P(G_0,G_1)$ such that $\E_{\dnull}[P(G_0,G_1)] = 0$, and
    \[
	\E_{\dstruct}\left[P(G_0,G_1)\right] \ge 40\cdot \max\left(\Var_{\dstruct}\left(P(G_0,G_1)\right)^{1/2}, \Var_{\dnull}\left(P(G_0,G_1)\right)^{1/2}\right).
    \]
Further, $P(G_0,G_1)$ is a sum of subgraph count-deviation correlation polynomials for subgraphs of size $v = \noise^{-O(1)}$, and is computable in time $n^{O(v)}$, where the $O(\cdot)$ in the exponent of $n$ hides only a dependence on the size of the representation of $(1-\delta)^{-1}$ as a ratio of two integers.
When $\gamma = \Omega(1)$, the algorithm runs in polynomial time.
\end{theorem}

\begin{proof}[Proof of \Cref{thm:variance}]
    Choose $v$ to be a sufficiently large even integer such that $v^{c} > 400/(\noise)^2$, where $c$ is the constant from \cref{thm:graph family} so that $|\cH| \ge v^{\frac{c}{1-\delta} v}$.
    Let $\cH$ be the test set of subgraphs guaranteed by \cref{thm:graph family} with $v$ vertices and $e$ edges, so that $\frac{e}{v} = (1-\delta)^{-1}$.
    By this setting of parameters we have then that $\E[\num H {G_0}] = 1 \pm o(1)$.

Define the polynomial $P$ to be the average of $p_H$ over $H \in \cH$,
\[
    P(G_0,G_1) = \frac{1}{|\cH|}\sum_{H \in \cH} p_{H}(G_0,G_1).
\]

    Since every $H \in \cH$ is strictly balanced, every strict subgraph $K\subset H$ has $\E_{G\sim \bb{G}(n,p)}[\num K G] = \omega(1)$.
    So by \Cref{lem:dist-exp} (or by the more precise \Cref{eq:fin-expect}) we have that
    \[
	\E_{\dstruct}\left[P(G_0,G_1)\right] \ge \ff{n}{v} \noise^{2e} (p^{e} - p^{2e}),
    \]
where we have also used that $\aut(H) = 1$ for every $H \in \cH$.

    We define the following quantity, which will show up repeatedly in our variance bounds:
    \begin{align}
	\rho = \frac{v^2}{(p\noise)^{(1-\delta)^{-1}} n} = O(1)\label{eq:cond-p}.
    \end{align}

The bounds from \cref{lem:dist-exp} gives us the expectation of $P$ under $\dnull$ and $\dstruct$.
Using the balancedness properties of our test set, we will bound the variance of $P$ under $\dnull$ and $\dstruct$.
    \paragraph{Variance bound for $\dnull$.}
Because $G_0$ and $G_1$ are independent and identically distributed, we have
    \begin{align}
	\E_{\dnull}\left[P(G_0,G_1)^2\right]
	&= \frac{1}{|\cH|^2}\sum_{H,H' \in \cH} \E_{\dnull}[ p_H(G_0,G_1) \cdot p_{H'}(G_0,G_1)]\nonumber\\
	&= \frac{1}{|\cH|^2}\sum_{H,H' \in \cH} \left(\E_{\dnull} \left[\big(\num{H}{G_0} -\E[\num{H}{G_0}]\big)\cdot\big(\num{H'}{G_0}- \E[\num{H'}{G_0}]\big)\right] \right)^2\nonumber\\
	&= \frac{1}{|\cH|^2}\sum_{H,H' \in \cH} \left(\E_{\dnull} \big[\num{H}{G_0}\cdot \num{H'}{G_0}\big] - \E[\num{H}{G_0}]\cdot\E[\num{H'}{G_0}] \right)^2.\label{eq:varnull}
    \end{align}
    Now, we will re-write the expression within the square as we did in the proof of \cref{lem:dist-exp}, when bounding the expectation of $p_H$ under $\dstruct$.
    We will sum over graphs $J$ which are subgraphs of both $H$ and $H'$.
    For each such $J$, let $\kappa_J = \num{J}{H}$ and let $\kappa'_J = \num{J}{H'}$.
    Then, letting $H \cup_J H'$ be the graph given by taking the union of $H$ and $H'$ and identifying the vertices and edges on the subgraph $J$, the first term in~\cref{eq:varnull} can be re-written as follows.
    \begin{equation}
	\E_{\dnull} \big[\num{H}{G_0}\cdot \num{H'}{G_0}\big]
	= \sum_{J \subseteq H,H'} \kappa_J \cdot \kappa_J'\cdot \aut(J) \cdot \ff{n}{2v - |V(J)|} \cdot \E_{G \sim \bb{G}(n,p\noise)}\big[\Ind[H \cup_J H' \text{ on } 2v - |V(J)| \text{ vertices}]\big],\label{eq:nullcor}
    \end{equation}
    where we have also used that $\aut(H) = \aut(H') = 1$.
    In $\bb{G}(n,p\noise)$, the event that $H \cup_J H'$ is present on $2v - |V(J)|$ vertices occurs with probability $(p\noise)^{2e - |E(J)|}$.
    This gives us
    \begin{equation}
	\cref{eq:nullcor}
	= \sum_{ J \subseteq H,H'} \kappa_J \cdot \kappa_J'\cdot \aut(J) \cdot \ff{n}{2v - |V(J)|} \cdot (p\noise)^{2e - |E(J)|}.\label{eq:deg2}
    \end{equation}
    Re-writing the term $\E[\num{H}{G_0}]\E[\num{H'}{G_1}]$ similarly and combining with \cref{eq:deg2},
    \begin{align}
&\E_{\dnull} \big[\num{H}{G_0}\cdot \num{H'}{G_0}\big] - \E[\num{H}{G_0}]\cdot\E[\num{H'}{G_0}]\nonumber\\
	&= \sum_{\substack{J \subseteq H,H'\\ |E(J)| \ge 1}} \kappa_J\cdot \kappa_J'\cdot \aut(J) \cdot \ff{n}{2v - |V(J)|} \cdot (p\noise)^{2e}\big( (p\noise)^{-|E(J)|} - 1\big).\label{eq:diff-null}
    \end{align}

    Now, when $H\neq H'$, if there is a subgraph $J \subseteq H,H'$, $K$ must be a strict subgraph.
    Since $H,H'$ are strictly balanced, every $J \subseteq H,H'$ has density strictly less than that of $H$ and $H'$.
    Therefore, we can assume that $|E(J)| \le \density |V(J)| - \beta_J$ for some fixed (as a function of $n$) positive $\beta_J$.
    Thus returning to \Cref{eq:diff-null}, when $H \neq H'$,
    \begin{align}
	\cref{eq:diff-null}
	&\le \sum_{\substack{J \subseteq H,H'\\ |E(J)| \ge 1}} \kappa_J\cdot \kappa_J'\cdot \aut(J)\cdot n^{2v - |V(J)|} (p\noise)^{2e -|V(J)|\density + \beta_J}\nonumber\\
	&\le n^{2v}(p\noise)^{2e + \min_J \beta_J}\cdot \sum_{i=2}^v \left(\frac{v^{2}}{n (p\noise)^{\density}}\right)^i,\label{eq:geom}
    \end{align}
    where we have grouped subgraphs $J$ according to $i = |V(J)|$, then bounded the contribution of the $\kappa_J$ using that there are at most $\binom{v}{i}$ subgraphs $J$ of size $i$, so
    \begin{equation}
    \sum_{\substack{J \subseteq H\\ |J| = s}} \kappa_J\cdot \kappa_J'\cdot \aut(J)
    \le \left(\sum_{\substack{J \subseteq H \\|J| = s}} \num{J}{H}\cdot \aut(J)\right)\cdot\left(\sum_{\substack{J \subseteq H \\|J| = s}} \num{J}{H'}\cdot \aut(J)\right)
    \le \left( \binom{v}{s} \cdot s!\right)^2 
    \le v^{2s}.\label{eq:cJ upper bound}
    \end{equation}
Since the final sum in \Cref{eq:geom} is geometric, using \Cref{eq:cond-p} we have that when $H \neq H'$,
\[
    \cref{eq:diff-null} \le (p\noise)^{\min_J \beta_J} \cdot n^{2v} (p\noise)^{2e} \cdot \rho^2 \frac{\rho^{v-1} - 1}{\rho - 1}
\]
where the extra $(p\noise)^{\min_J \beta_J}$ factor is due to the strictly balanced property which implies $\beta_J$ is a positive constant independent of $n$ for each $J$.
When $H = H'$, an identical calculation bounds the terms in \Cref{eq:diff-null} for which $J \neq H$; when $J = H$, then (since $H$ has no non-trivial automorphisms) the term contributes $n^{v}(p\noise)^e$.
Thus we have
\begin{equation}
	\cref{eq:diff-null}
	\le
	\begin{cases}
	    n^v(p\noise)^e + (p\noise)^{\min_J \beta_J} \cdot n^{2v} (p\noise)^{2e} \cdot \rho^2 \frac{\rho^{v-1} - 1}{\rho - 1} & H = H'\\
	    (p\noise)^{\min_J \beta_J} \cdot n^{2v} (p\noise)^{2e} \cdot \rho^2 \frac{\rho^{v-1} - 1}{\rho - 1} & H \neq H'.
	\end{cases}
\end{equation}

Finally, combining with \Cref{eq:varnull}, this gives a bound on the variance:
\begin{equation}
    \Var_{\dnull}(P(G_0,G_1))
    \le (p\noise)^{2\min \beta_J} \cdot \left(n^{2v} (p\noise)^{2e} \cdot \rho^2 \frac{\rho^{v-1} - 1}{\rho - 1}\right)^2 + \frac{1}{|\cH|} n^{2v} (p\noise)^{2e}.\label{eq:varnull-bound}
\end{equation}
Note that in~\Cref{eq:varnull-bound}, $\rho=O(1)$ and $n^{2v}(p{\noise})^{2e}=O(1)$ while $(p\gamma)^{\min_J \beta_J}=o(1)$, thus, the first term in \Cref{eq:varnull-bound} is $o(1)$. Furthermore $|\cH| \ge 200\cdot \noise^{-2e}$, therefore the second term is at most $\frac{1}{200} n^{2v}p^{2e}\noise^{4e} = \frac{1}{200}\E[P(G_0,G_1)]^2$.
Thus, for sufficiently large $n$ we have
\[
\E_{\dstruct}\left[P(G_0,G_1)\right] \ge 40\cdot\Var_{\dnull}\left(P(G_0,G_1)\right)^{1/2}.
\]

    \paragraph{Variance bound for $\dstruct$.}
    In the structured case, the correlation of $G_0$ and $G_1$ introduces some additional complications.
    In particular, because $G_0$ and $G_1$ are not independent, the expectation does not factor nicely.
    We'll expand the expression for $\E_{\dstruct}[P(G_0,G_1)^2]$; we'll use the shorthand $X_H$ for $\num H G$ to save space when $G$ doesn't matter.
    We have
    \begin{align}
	&\Var_{\dstruct}[P(G_0,G_1)]\nonumber\\
	&=\E_{\dstruct}[P(G_0,G_1)^2] - \E_{\dstruct}[P(G_0,G_1)]^2\nonumber\\
	&= \frac{1}{|\cH|^2}\sum_{H,H' \in \cH}\E_{\dstruct}\left[\big(\num{H}{G_0}-\E[X_H]\big)\big(\num{H'}{G_0}-\E[X_{H'}]\big)\big(\num{H}{G_1}-\E[X_H]\big)\big(\num{H'}{G_1}-\E[X_{H'}]\big)\right]\label{eq:bigexp}\\
	& \qquad - \E_{\dstruct}\left[\big(\num{H}{G_0}-\E[X_H]\big)\big(\num{H}{G_1}-\E[X_H]\big)\right]\E_{\dstruct}\left[\big(\num{H'}{G_0}-\E[X_{H'}]\big)\big(\num{H'}{G_1}-\E[X_{H'}]\big)\right].\nonumber
    \end{align}
    As in the proof of \cref{lem:dist-exp}, we will write $\num H G$ as a sum of indicators that ordered subsets of vertices contain $H$ as a subgraph.
    We apply such a transformation to the first half of the summand above, and we have
    \begin{align}
	&\E_{\dstruct}\left[\big(\num{H}{G_0}-\E[X_H]\big)\big(\num{H'}{G_0}-\E[X_{H'}]\big)\big(\num{H}{G_1}-\E[X_H]\big)\big(\num{H'}{G_1}-\E[X_{H'}]\big)\right]\nonumber\\
	&= \sum_{\substack{A \in \ff{V(G_0)}{v}\\B \in \ff{V(G_0)}{v} \\ C \in \ff{V(G_1)}{v}\\ D \in \ff{V(G_1)}{v}}} \E_{\dstruct}\big[(\isH{A}{H} - \E[\isH{A}{H}])(\isH{B}{H'}-\E[\isH{B}{H'}])(\isH{C}{H}-\E[\isH{C}{H}])(\isH{D}{H'}-\E[\isH{D}{H'}])\big].\label{eq:4term}
    \end{align}
    Rather than expand the product above into 16 terms, we liken it to a process of {\em inclusion-exclusion}.
    We have already seen that sums of the form $\sum\E[ \isH{A}{H} \isH{B}{H'}]$ are equivalent to weighted sums over subgraphs $J \subseteq H,H'$, which correspond to different ways of taking the union of the graphs $H,H'$ by choosing the intersection to be equal to the subgraph $J$.
    Specifically, the weight for each subgraph $J$ is proportional to $n^{-|V(J)|} p^{-|E(J)|}$.
    The term $\sum\E[\isH{A}{H}]\E[\isH{B}{H'}]$, on the other hand, corresponds to a similar sum, but the weight for the $J$th term is proportional merely to $n^{-|V(J)|}$.
    For terms in which $J$ contains no edges, these weights are equivalent; when $|E(J)| > 0$, the weight of the first term dominates the weight of the second.
    Therefore, when the second term is subtracted from the first, we are ``exclude'' the intersections $J$ which contain no edges (up to lower order terms).

    An analogous phenomenon transpires above.
    The leading term, $\E[\isH{A}{H}\isH{B}{H'}\isH{C}{H}\isH{B}{H'}]$ includes {\em all} ways of intersecting the two copies of $H$ and the two copies of $H'$.
    The following ``order-1 exclusion term'', $\E[\isH{A}{H}]\E[\isH{B}{H'}\isH{C}{H}\isH{D}{H'}]$, is subtracted to exclude terms in which the copy of $H$ in $G_0$ does not intersect with any of the other subgraphs on any edges.
    After subtracting all three of the order-1 exclusion terms, we must add the six ``order-2 inclusion terms'', (e.g. $\E[\isH{A}{H}]\E[\isH{B}{H'}]\E[\isH{C}{H}\isH{D}{H'}]$), and so on, until we are finally left with a summation in which every copy of $H$ and $H'$ must intersect with at least one other subgraph on at least one edge.\footnote{For example, if $H$ in $G_0$ contains the edge $(u,v)$, and $H'$ in $G_1$ contains the edge $(\pi(u),\pi(v))$, this term is included in the summation. On the other hand if $H$ in $G_0$ does not share any edges with the second copy of $H$ in $G_1$ or either copy of $H'$, the term is excluded.}

    To disambiguate which ``copy'' of $H$ and $H'$ we refer to (that in $G_0$ or $G_1$), we now attach subscripts $H_0$, $H_1$).
    We can write our expression as a sum over subgraphs $J \subseteq H_0,H_0'$, $K \subseteq H_1,H_1'$, and $L \subseteq H_0\cup_J H_0', H_1\cup_K H_1'$, which satisfy the condition $\cE(J,K,L)$ that $H_0,H_1,H'_0$ and $H_1'$ all contain at least one edge in $J \cup K \cup L$:
    \begin{align}
	\cref{eq:4term}
	&\le \sum_{\substack{J \subseteq H_0,H_0'\\ K \subseteq H_1,H_1' \\ L \subseteq H_0 \cup_J H_0', H_1 \cup_K H_1' \\ \cE(J,K,L)}} c_{J,K,L}(H,H')\cdot \ff{n}{4v - |V(J)| - |V(K)| - |V(L)|}\cdot (\noise p)^{4e}({\noise}p)^{- |E(J)| - |E(K)|}p^{- |E(L)|} .\nonumber
    \end{align}
    where we have let $c_{J,K,L}(H,H')$ to be the number of ways one can obtain a graph by (a) taking two ordered, labeled copies of $H$ and intersecting them on a subgraph isomorphic to $J$, (b) taking two ordered, labeled copies of $H'$ and intersecting them on a subgraph isomorphic to $K$, and (c) the intersection of the above two intersecting parts is isomorphic to $L$. 
For our bounds, we will ultimately only use that $c_{J,K,L}(H,H') \le \kappa_J \kappa_K \kappa'_J \kappa'_K \kappa^{J}_L\kappa^K_L\cdot \aut(J)\cdot\aut(K)\cdot\aut(L) = O(1)$ is independent of $n$. 

    To obtain the expression for the variance, we now subtract the second half of the summand from \Cref{eq:bigexp}, $\E_{\dstruct}[p_H(G_0,G_1)]\cdot \E_{\dstruct}[p_{H'}(G_0,G_1)]$.
    This term, too, can be written as a summation over subgraphs $J,K,L$, in which there must be an intersection between $H_0,H_1$ and between $H_0',H_1'$.
    Since it is only re-weighted according to the intersections of $H_0,H_1$ with each other and $H_0',H_1'$ with each other, subtracting this term will fully cancel any combination of $J,K,L$ in which the exclusive intersections occur between copies of the identical graph (i.e. both copies are $H$ or $H'$).
    Therefore, letting $\cF(J,K,L)$ be the condition that the exclusive intersections in $J\cup K \cup L$ are not between $H_0,H_1$ and $H_0',H_1'$, we have that
    \begin{align}
	&\Var(P(G_0,G_1))\nonumber\\
	&\le \frac{1}{|\cH|^2} \sum_{H,H' \in \cH} \sum_{\substack{J \subseteq H_0,H_0' \\ K \subseteq H_1,H_1'\\L \subseteq H_0 \cup_J H_0', H_1\cup_K H_1'\\\cE(J,K,L)\wedge \cF(J,K,L)}} c_{J,K,L}(H,H') \cdot \ff{n}{4v - |V(J)| - |V(K)| - |V(L)|}\cdot ({\noise}p)^{4e}({\noise}p)^{- |E(J)| - |E(K)|}p^{- |E(L)|}.
    \end{align}

Now we claim that if $H\neq H'$, then the union of $J,K,L$ must be strictly sparser than $H$ or $H'$, or equivalently, that the graph union $H_U$ of $H_0,H_1,H_0',H_1'$ given by identifying the edges together is strictly denser than $H$ and $H'$.
The conditions $\cE(J,K,L)$ and $\cF(J,K,L)$ ensure that we cannot have $H_U$  isomorphic to a disjoint union of copies of $H$ and $H'$.
At this point, we appeal to the following claim:
    \begin{claim}\label{claim:intersections}
	Suppose that $A$ is a graph of density $\alpha$, and let $B$ be a strictly balanced graph of density $\beta \le \alpha$.
	Then if $J \subset A,B$ is a non-empty proper subgraph of both $A$ and $B$, and we form the graph $A \cup_J B$ by identifying the vertices corresponding to $J$ in $A$ and $B$, then $A \cup_J B$ has density strictly larger than $\beta$.
    \end{claim}
    \begin{proof}
	We have that the density of $A \cup_J B$ is given by
	\[
	    \frac{|E(A)| + |E(B)| - |E(J)|}{|V(A)| + |V(B)| - |J|} = \frac{\alpha|V(A)| + \beta|V(B)| - |E(J)|}{|V(A)| + |V(B)| - |J|}.
	\]
	Since $B$ is strictly balanced, by definition $J\subset B$ must be sparser than $B$.
	Thus, $|E(J)| < \beta |J|$, and the conclusion follows.
    \end{proof}
We can apply this claim iteratively if we take one intersection at a time (first intersecting a copy of $H'$ with a copy of $H$). That is, the density of $H_U$ is strictly larger than $(1-\delta)^{-1}$.
Therefore, $|E(L)| + |E(J)| + |E(K)| \le \density(|V(J)| + |V(K)| + |V(L)|) - \min_{J,K,L}\beta_{J,K,L}$, where $\beta_{J,K,L}$ is a positive number independent of $n$.
So when $H \neq H'$, we have that the value of the summation is bounded by
\begin{align*}
    &\le \sum_{\substack{J \subseteq H_0,H_0' \\ K \subseteq H_1,H_1'\\L \subseteq H_0 \cup_J H_0', H_1\cup_K H_1'\\\cE(J,K,L)\wedge \cF(J,K,L)}} c_{J,K,L}(H,H')\cdot \ff{n}{4v - |V(J)| - |V(K)| - |V(L)|}\cdot \noise^{2e} p^{4e - \density(|E(J)| + |E(K)|+ |E(L)|) + \min_{J,K,L}\beta_{J,K,L}}\\
    &\le n^{4v} {\noise}^{2e} p^{4e+\min_{J,K,L} \beta_{J,K,L}} \cdot \sum_{i=2}^{3v}\left(\frac{4v^6}{np^{\density}}\right)^i
    \, = \,
    n^{4v} {\noise}^{2e} p^{4e+\min_{J,K,L} \beta_{J,K,L}} \cdot \rho^2 \cdot \frac{\rho^{3v-1} - 1}{\rho - 1}.
\end{align*}
where we have bounded the $c_{J,K,L}(H,H')$ using the same trick as we did in~\Cref{eq:cJ upper bound} and finally applied \Cref{eq:cond-p}.

When $H = H'$, the above applies except for the case when $J = \emptyset, K = \emptyset$, and $L$ is the disjoint union of $H$'s, so that $H_0$ intersects fully with $H'_0$ and the identical thing happens in $G_1$, and also when all of the copies intersect fully.
For those two terms, the value is $n^{2v}({\noise}p)^{2e} + n^v p^e {\noise}^{2e}$.

Thus, our bound on the variance is
\begin{align}
    \Var_{\dstruct}(P(G_0,G_1))
    \le n^{4v}\noise^{2e}p^{4e + \min_J\beta_J} \cdot \rho^2 \cdot \frac{\rho^{3v -1} -1}{\rho-1} + \frac{1}{|\cH|} \cdot \left(n^{2v}(\noise p)^{2e} + n^v p^e \noise^{2e}\right).
\end{align}
Here, the first term has magnitude $O(p^{\min_{J,K,L} \beta_{J,K,L}}/{\noise}^{2e})\cdot \E[\num H G_0]^4 = o(1)$, since $p = o(1)$ and $\min_{J,K,L} \beta_{J,K,L},\rho,\noise,\E[\num H G_0] = O(1)$.
Since $|\cH| = (c_1 v)^{c_2 e} > \left(\frac{400}{\noise^2}\right)^e$,
the second term is at most $\tfrac{1}{400}\E[P(G_0,G_1)]^2 + \tfrac{\noise^{3e}}{400}\E[\num H G_0]$.
Therefore for $n$ sufficiently large we have that
\[
    \E_{\dstruct}[P(G_0,G_1)] \ge 40 \cdot \Var_{\dstruct}(P(G_0,G_1))^{1/2}.
    \]
This completes the proof.
\end{proof}

\subsection{Putting everything together}

Given \Cref{thm:variance}, we can complete our algorithm and prove  \Cref{thm:distinguish}:

\begin{proof}[Proof of \Cref{thm:distinguish}]
    The algorithm is as follows: compute the value of $P(G_0,G_1)$, and if the value is larger than $\frac{1}{3}\E_{\dstruct}[P(G_0,G_1)]$, say that $G_0$ and $G_1$ come from $\dstruct$, otherwise say that $G_0$ and $G_1$ come from $\dnull$.
    Clearly the algorithm runs in polynomial time: it requires counting the number of occurrences of constant-sized subgraphs in two graphs of size $n$.

    The claim is that in both cases, we are correct with probability at least $0.99$.\footnote{Notice that had we taken $v = \omega(1)$, we could amplify this probability to $1-o(1)$.}
    The reason is that if we apply Chebyshev's inequality to $P(G_0,G_1)$, by \Cref{thm:variance} and \Cref{lem:dist-exp} we have that in the structured distribution, $P(G_0,G_1)$ is unlikely to be too small:
    \[
	\Pr_{\dstruct}\left( P(G_0,G_1) < \frac{1}{2} \E_{\dstruct}[P(G_0,G_1)] \right) \le \frac{4}{1600}.
    \]
At the same time, Chebyshev's inequality also guarantees that $P(G_0,G_1)$ is not too large with good probability:
    \[
	\Pr_{\dnull}\left( P(G_0,G_1) \ge \frac{1}{4} \E_{\dstruct}[P(G_0,G_1)] \right) \le \frac{16}{1600}.
    \]
   This concludes the proof.
\end{proof}

\section{Recovering the ground truth permutation} \label{sec:recovery}

In this section, we will solve the following recovery problem:
\begin{problem}\label{prob:recovery}
    Let  $G_0,G_1$ be $n$-vertex graphs sampled from $\dstruct(n,p;\noise)$ according to the following procedure: God samples a ``base'' graph $\Base\sim \bb{G}(n,p)$, then twice independently subsamples each edge of $\Base$ with probability $\noise$ to form the graphs $G_0,\tilde G_1$, and finally applies a random permutation $\pi^*$ to $\tilde G_1$ to obtain $G_1$.
    We are given $G_0,G_1$ without knowledge of $\pi^*$, our goal is to recover $\pi^*$.
\end{problem}

We give a quasipolynomial time algorithm for the above problem; below we state the precise guarantees.

\begin{theorem}[Recovery]
\label{thm:recovery}
    For an arbitrarily small constant $\delta>0$, if $(G_0,G_1,\pi) \sim \dstruct(n,p;\gamma)$ so that $p \in [\tfrac{ n^{\varepsilon}}{n},\tfrac{n^{1/153}}{n}]$ where\footnote{We did not make an effort to optimize the constant in the exponent $n^{1/153}$.} $\varepsilon = \Omega\left(\left(\frac{\log\log^4 n}{\log n}\right)^{\frac{1}{2}}\right)$ or $p \in [n^{-1/3},n^{-\delta}]$, and $\noise\ge\left(\frac{1}{\log n}\right)^{o(1)}$, then there is a randomized $n^{O(\log n)}$-time algorithm $A$ that recovers $\pi$ exactly with high probability over the input $(G_0,G_1)$ and over the choices of $A$.
\end{theorem}
Note that \cref{thm:recovery} implies \cref{thm:recovery:intro}.\\

Our algorithm first finds a rough estimate of $\pi^*$ using test subgraphs (as we did for our distinguishing algorithm), then completes the estimate via a boosting procedure.
We fill in the details of our strategy in \cref{sec:test-rec,sec:alg}.
In \cref{sec:unique,sec:coverin,sec:partial,sec:moments}, we prove that our algorithm successfully estimates $\pi^*$, and in \cref{sec:boost} we give the boosting algorithm.
Finally, we  tie everything together and prove \cref{thm:recovery} in \cref{sec:putting-together}.

\subsection{Test subgraphs for recovery}\label{sec:test-rec}

For the purposes of distinguishing, we used the counts of subgraphs belonging to a specially constructed ``test set'' of strictly balanced graphs with no non-trivial automorphisms.
However, we only used subgraphs of size $O(1)$ to distinguish, each of which crucially appeared only ${O}(1)$ times.
For a fixed constant size $s$, there are only $2^{s^2}$ possible subgraphs on $s$ vertices; if each appears only with multiplicity ${O}(1)$, most vertices in the graph will not participate in such subgraphs.

To recover the permutation $\pi^*$, we will use the same ``test set'' of subgraphs we used for distinguishing, only now we will choose the subgraphs to have size $\Omega(\log n)$ (and consequently we will have a larger test set, so that more vertices are covered).
Also, we will set their density so that every subgraph which does appear in both $G_0$ and $G_1$ almost certainly appears exactly once.
If we see an occurrence of such a subgraph $H$ in both $G_0$ and $G_1$, we can be confident that the vertices of $H$ in $G_0$ and $G_1$ correspond to each other in $\pi^*$.
Finally, we fix any mis-matched vertices with a boosting procedure.

We require stronger conditions from our test set, which we delineate below.
\begin{theorem}[see \Cref{thm:graph family,prop:large-d,prop:small-d}]\label{thm:meta-graph-fam}
Let $d' = d + \lambda$ for an integer $d \ge 2$ and a rational number $\lambda  \in [0,1]$.
Let $v,e$ be integers such that $e=\frac{d'}{2}v$.
Then so long as $d = 2$ and $\lambda \le \frac{1}{76}$ or $d \ge 6$ (with arbitrary $\lambda$), then there exists a test set of graphs $\cH_{d'}^v = \{H\}$, each on $v$ vertices and containing $e$ edges, which satisfies the following properties:
\begin{enumerate}
    \item Every $H \in \cH$ is {\em strictly balanced}, in a strong sense: for any $\alpha \in (0,1)$, every subgraph of $H$ on at most $\alpha v$ vertices has density $(\frac{d'}{2} - f(\alpha,d'))$ for an explicit function $f(\alpha,d') > 0$.
    \item Every $H\in\cH$ has no non-trivial automorphisms.
    \item The size of the family is $|\cH|\geq  v^{\Omega(C_{d'}\cdot v)}$, for a constant $C_{d'}$ depending on $d'$.
    \item For every pair of distinct graphs $H,H' \in \cH$, every common subgraph $J \subset H,H'$ is less dense than $H$ and $H'$ in a strong sense: $|E(J)|\leq (\frac{d'}{2} - g(d'))\cdot|V(J)|$ for an explicit function $g(d') >0$.
\end{enumerate}
\end{theorem}
Because our constructions differ depending on the value of $d'$, we have stated the theorem in generality here; for precise statements, see \cref{sec:test-graphs}.

\subsection{The recovery algorithm}\label{sec:alg}
The recovery algorithm consists of four steps.
\begin{algorithm}[Recovery]\label{alg:recover}
    On input $(G_0,G_1)\sim \dstruct(n,p;\noise)$:
\begin{enumerate}
    \item Generate a family of test graphs $\cH := \cH_{d'}^v$ that satisfies the properties guaranteed in \cref{thm:meta-graph-fam}, choosing $v,d'$ as follows:
	\begin{itemize}
	    \item If $p = n^{\delta-1} \in [\frac{n^\varepsilon}{n},\frac{n^{1/153}}{n}]$, choose $v = \Theta(\log n)$ to be the smallest even integer so that $\lambda v$ is also an integer, for some $\lambda$ chosen from the window $\lambda \in \left(\frac{2\delta}{1-\delta},\frac{2\delta}{1-\delta} + \frac{\log\log n}{\log n}\right)$.
		Choose $d' = 2 + \lambda$.
	    \item If $p=n^{\delta-1}\in[n^{-1/3},n^{-\epsilon}]$: Choose $v = \Theta(\log n)$ to be the smallest even integer so that there is some $d'\in \left(\frac{2}{1-\delta},\frac{2}{1-\delta}+\frac{\log\log n}{4\log n}\right)$, so that $(d'-\lfloor d'\rfloor)v$ is also an integer.
	\end{itemize}
	In \cref{lem:partial recovery}, we verify that these conditions are feasible.
    \item For each vertex $u \in V(G_0)$, we find all $H \in \cH$ such that $u$ is incident on a copy of $H$ in $G_0$, and so that $H$ also appears in $G_1$.
	If $u$ is incident on at least $\frac{1}{2}|\cH|\cdot v \cdot n^{v-1} q^{e}$ distinct $H$, then we choose one uniformly at random, and set $\pi(u)$ equal to the corresponding vertex in the copy of $H$ in $G_1$.
	If a collision occurs so that $\pi(u) = \pi(u')$ for $u \neq u'$, ties are broken arbitrarily.
    \label{step:match}
    \item Run a boosting algorithm (\cref{algo:boosting}, see~\cref{sec:boost}) with the partial map obtained from step 3 .
\end{enumerate}
Output the resulting permutation $\pi$.
\end{algorithm}
\medskip

To prove that \cref{alg:recover} works, we'll establish the following three claims:
\begin{itemize}
    \item The probability that a vertex $u \in V$ appears in a test graph $H \in \cH$ which appears more than once in the base graph $\Base$, and also appears at least once in both $G_0,G_1$, is small.
    \item Each vertex $u \in V$ appears in some test graph $H \in \cH$ in both $G_0,G_1$ with large enough probability.

	Since most $H \in \cH$ that appear in both $G_0,G_1$ will appear uniquely in the base graph $\Base$, together, these claims imply that in step \ref{step:match} we will match many vertices to each other, and make few mistakes.
    \item Given a one-to-one map between an $\tilde\Omega(1)$-fraction of the vertices of $G_0,G_1$ which is consistent with $\pi^*$ on all but $o(1)$ of the matches, the boosting algorithm recovers the optimal permutation $\pi^*$.
\end{itemize}

\subsection{Few test graphs introduce errors in the algorithm}\label{sec:unique}
In this subsection, we will prove that few test graphs $H \in \cH$ appear more than once in the base graph, and also survive the subsampling in both $G_0$ and $G_1$.
We will bound the probability of this event in a straightforward manner, taking advantage of properties of the test graphs in $\cH$.

\begin{lemma}\label{lem:recovery uniquness single subgraph}
    Suppose that $G_0,G_1 \sim \dstruct(n,p;\gamma)$ are both subsampled from the base graph $\Base = (V,E)$.
    Define $q = p\noise^2$.
    Suppose we have chosen $\cH$ to be a test set comprised of graphs on $v$ vertices as constructed in \cref{sec:test-graphs}, so that $H \in \cH$ has average degree $d'$, and so that $np^{d'/2} < 1$.

    For a vertex $u \in V$, let $B_u$ be the number of ``bad'' $H \in \cH$ that contain $u$, such that $H \in \cH$ appears at least twice in the base graph $\Base$, and that a copy of $H$ survives in both $G_0$ and in $G_1$.

    If $d' = 2 + \lambda \le 2 + \frac{1}{76}$ and further $v \log^2 v \ll \lambda^2 \log^2 n$, $\frac{1}{v^8} < np^{d'/2}$ and $p = n^{\delta - 1}$ for $\delta < \frac{1}{3}$, then:
	    \[
	    \E[B_u]
	    \le |\cH|\cdot v\cdot n^{v-1} q^e\cdot O\left(\frac{1}{\sqrt{v}}\right),
	    \]

    And if $d' \ge 6$, then if $2v^{2+2(d'+2)} q^{\frac{1}{50}}\le v^{2(d'+2)} nq^{d'/2} \le \frac{1}{2}q^{-\frac{1}{50}}$, then:
    \[
	\E[B_u]
	    \le |\cH|\cdot v\cdot n^{v-1} q^e\cdot  16 \left(v^{2(d'+2)} \cdot np^{\frac{1}{2}d' + \frac{1}{50}}\right).
    \]
\end{lemma}
\begin{proof}
    Fix some $H \in \cH$, and let $B^H_u$ be the bad event that the vertex $u$ is contained in $H$, and $H$ appears more than once in the base graph $\Base$, and also that at least one copy of $H$ survives in both $G_0$ and $G_1$.
    Let $\isH{H}{S}$ be the indicator that the ordered vertex set $S \subset V$ contains a labeled copy of $H$ in $\Base$.
    If there is a copy of $H$ on the ordered vertex set $S \subset V$ where $S$ contains $u$, and also a distinct copy of $H$ on a different ordered set of vertices $S' \subset V$ so that $|S \cap S'| < v$, then the product of indicators $\isH{H}{S}^{(0)} \cdot \isH{H}{S'}^{(1)}$ yields 1.
    Now, define the indicators $\Ind_{S,H}^{G_0},\Ind_{S,H}^{G_1}$ to be the indicators that $H$ survives on $S$ in $G_0,G_1$ respectively.

    The probability that $u$ is contained in $H$ which appears more than once in $\Base$ and which has at least one copy surviving in both $G_0,G_1$ is at most
    \begin{align}
	\Pr[B^H_u]
	&\le \sum_{\substack{S,S' \subseteq V\\|S| = |S'| = v\\ |S \cap S'| < v\\u \in S}} \E\left[\isH{H}{S}\cdot\isH{H}{S'}\cdot\left(\Ind_{S,H}^{G_0} + \Ind_{S',H}^{G_0}\right)\cdot\left(\Ind_{S,H}^{G_1} + \Ind_{S',H}^{G_1}\right)\right]\nonumber\\
	&= \sum_{\substack{S,S' \subseteq V\\|S| = |S'| = v\\ |S \cap S'| < v\\u \in S}} \E\left[\left(\Ind_{S,H}^{G_0} + \Ind_{S',H}^{G_0}\right)\cdot\left(\Ind_{S,H}^{G_1} + \Ind_{S',H}^{G_1}\right)~\Big{|}~\isH{H}{S}\cdot\isH{H}{S'}\right]\cdot \Pr[\isH{H}{S}\cdot\isH{H}{S'}]
	\label{eq:expectation-bad}
    \end{align}
    The copies of $H$ on $S,S'$ may have shared edges (though because $S \neq S'$, they cannot share the entire graph $H$).
    Therefore, if the copies of $H$ and $H'$ intersect on the edge-induced strict subgraph $J \subset H$, then
    \[
	\Pr[\isH{H}{S}\cdot\isH{H'}{S}] = p^{2e - |E(J)|}.
    \]
Now, because the edges of $G_0,G_1$ are subsampled from $\Base$ independently,
    \[
\E\left[\left(\Ind_{S,H}^{G_0} + \Ind_{S',H}^{G_0}\right)\cdot\left(\Ind_{S,H}^{G_1} + \Ind_{S',H}^{G_1}\right)~\Big{|}~\isH{H}{S}\cdot\isH{H}{S'}\right]
    = \E\left[\left(\Ind_{S,H}^{G_0} + \Ind_{S',H}^{G_0}\right)~\Big{|}~\isH{H}{S}\cdot\isH{H}{S'}\right]^2
    = \left(2\gamma^e\right)^2.
    \]
Thus, we may refine the sum over $S,S'$ in \cref{eq:expectation-bad} into a sum over the proper subgraph $J$ upon which $S,S'$ intersect:
    \begin{align}
	\cref{eq:expectation-bad}
	~=~ \sum_{\substack{J \subset H}} \sum_{\substack{S,S' \subseteq V\\ S \cap S' = J\\u \in S}} 4 \gamma^{2e} p^{2e - |E(J)|}
	&\le \sum_{\substack{J \subset H}} \num{J}{H}^2\cdot\aut(J)  \cdot v\cdot \ff{(n-1)}{2v - 1 - |V(J)|} \cdot 4\gamma^{2e} p^{2e - |E(J)|}\nonumber\\
	&\le 4v\cdot n^{2v-1}\gamma^{2e}p^{2e}\sum_{\substack{J \subset H}}\num{J}{H}^2\cdot\aut(J)\cdot  n^{- |V(J)|} \cdot p^{- |E(J)|}.\label{eq:need-bound}
    \end{align}
    where we have used that the number of $S,S'$ that intersect on $J$ is at most $\num{J}{H}^2\cdot\aut(J)$ for the choice of the position of $J$ in $H$ for each set and the automorphism between the two copies of $J$, $v$ for the choice of the position of $u$ within $S$, and $\ff{(n-1)}{2v - 1 - |V(J)|}$ for the choice of the identities of the remaining vertices in $S$ and $S'$.

	The strict balancedness property of $H$ is no longer sufficient to ensure that this sum is small; that is because the density of $H$ is such that $nq^{e/v} \ll 1$, and therefore we must be careful that the terms of \cref{eq:need-bound} corresponding to the densest, and largest, subgraphs are not too numerous.
	Because we utilize different constructions for our test set $\cH$ depending on the density of $H$, we will require distinct arguments for bounding $\cref{eq:need-bound}$ for the case when $|E(H)|/|V(H)| \le 1.5$ and $|E(H)|/|V(H)| > 1.5$.
    We now apply the following lemmas:
	\begin{restatable}{lemma}{sparseequal}\label{lem:sparseequal}
	   When the average degree of $H$ is $d' = 2 + \lambda$ with $\lambda \in [0,\frac{1}{76}]$, then if $\frac{1}{v^8} \le nq^{d'/2}$, for $v$ sufficiently large,
    \[
	\sum_{\substack{J \subset H}} \num{J}{H}^2\cdot\aut(J)\cdot n^{-|V(J)|}q^{-|E(J)|}
        \le \sum_{s=1}^{(1 - \theta)v} \left(\frac{v^2q^{\frac{1}{100}\lambda \theta}}{n q^{1 + \frac{1}{2}\lambda}}\right)^{s} + \theta v\left(\frac{1}{nq^{1 + \frac{1}{2}\lambda}}\right)^v \left(v^8n q^{1 + \frac{1}{2}\lambda}\right)^{\theta v} q^{\frac{1}{100}\lambda(1 - \theta)}.
    \]
	\end{restatable}

\begin{restatable}{lemma}{denseequal}
    \label{lem:denseequal}
    When the average degree of $H$ is $d' \ge 6$, so that $2v^{2+2(d'+2)} q^{\frac{1}{50}}\le v^{2(d'+2)} nq^{d'/2} \le \frac{1}{2}q^{-\frac{1}{50}}$ and $(nq^{d'/2})^{v-1} \le v^{2(d'+2)}q^{\frac{1}{50}}$, then
    \[
	\sum_{\substack{J \subset H}} \num{J}{H}^2\cdot\aut(J)\cdot n^{-|V(J)|}q^{-|E(J)|} \le  4  (v^{2(d'+2)} q^{\frac{1}{50}})\cdot \left(\frac{1}{nq^{d'/2}}\right)^{v-1}.
    \]
\end{restatable}
We will prove these lemmas in \cref{sec:moments}; first we will complete the theorem.

    \paragraph{Sparse case.}
    In the sparse case, we will apply \cref{lem:sparseequal} (under the assumption that $np^{d'/2} \ge v^{-8}$) with setting $q=p$ and $\theta = \frac{600 \log v}{\lambda \log n}$.
    At this parameter setting, if $p = n^{\delta -1}$, then
    \[
    \log \left(v^2 p^{\frac{1}{100}\theta\lambda} \right)
    = 2 \log v - \frac{1}{100}\theta \lambda (1-\delta)\log n
    < \left(2 - 6(1-\delta)\right)\log v < 0,
    \]
    so long as $\delta < \frac{1}{3}$.
    The summation term from the bound in \cref{lem:sparseequal} is thus $\le \frac{1}{\sqrt{v}}n^{-v}p^{-e}$.
Also, for the second term,
\begin{align*}
\log\left(\theta v \cdot (v^8 n p^{1+\frac{1}{2}\lambda})^{\theta v} p^{\frac{1}{100}\lambda(1-\theta)}\right)
&= \log \theta v + \theta v(8\log v - (\delta - \frac{\lambda}{2} + \delta \frac{\lambda}{2})\log n) - \frac{1}{100}\lambda(1-\theta)(1-\delta) \log n\\
&\le \log \theta v + \frac{4800}{\lambda\log n}v\log^2 v - \frac{1}{300}\lambda \log n
\ll 0,
\end{align*}
where the final inequality follows for $n$ sufficiently large given our assumption that $v \log^2 v \ll \lambda^2 \log^2 n$.
It follows that the second term is $\le \frac{1}{n^{\frac{\lambda}{300}}}\cdot n^{-v}p^{-e}$.
Therefore,
    \begin{align}
	    \cref{eq:need-bound}
	    &\le 4v \cdot n^{2v-1} \gamma^{2e} p^{2e} \cdot o\left(\frac{1}{\sqrt{v} n^vp^e}\right)
	    \le O\left(v^{-1/2}\right)\cdot v \cdot n^{v-1} q^e,\label{eq:sBbd}
    \end{align}
    where we have used that $q = p\noise^2$.

	By linearity of expectation, summing over all $H \in \cH$, the number of $H$ which $u$ participates in that have at least two distinct appearances in $\Base$ and survive at least once in each of $G_0,G_1$ is at most (by \cref{eq:sBbd}):
	\[
	    \E[B_u]
	    = \E\left[\sum_{H \in \cH} B^H_u \right]
	    \le \sum_{H \in \cH} \Pr\left[B_u^H\right]
	    \le |\cH|\cdot v\cdot n^{v-1} q^e\cdot O\left(v^{-1/2}\right)
	\]
	which completes the proof for the sparse case.
    \paragraph{Dense case.}
	In the dense case, since $J$ is always a strict subgraph of $H$, we can now apply \cref{lem:denseequal} (since we have assumed that $2v^{2+2(d'+2)} q^{\frac{1}{50}}\le v^{2(d'+2)} nq^{d'/2} \le \frac{1}{2}q^{-\frac{1}{50}}$)
	to bound  \cref{eq:need-bound}:
	\begin{align}
	    \cref{eq:need-bound}
	    &\le 4v\cdot n^{2v-1} \gamma^{2e}q^{2e}  \cdot 4 \cdot  (v^{2(d'+2)} p^{\frac{1}{50}})\cdot \left(\frac{1}{np^{d'/2}}\right)^{v-1}
	    = 4 v\cdot n^{v-1}q^e\cdot  4 \left(v^{2(d'+2)} \cdot np^{\frac{1}{2}d' + \frac{1}{50}}\right).\label{eq:Bbound}
	\end{align}
	where we have used that $q = p\noise^2 $ and $e = \frac{d'}{2}v$.

	Now by definition of $B_u$ and by \cref{eq:Bbound},
	\[
	    \E[B_u]
	    = \E\left[\sum_{H \in \cH} B^H_u \right]
	    \le \sum_{H \in \cH} \Pr\left[B_u^H\right]
	    \le |\cH|\cdot 4v\cdot n^{v-1} q^e\cdot  4 \left(v^{2(d'+2)} \cdot np^{\frac{1}{2}d' + \frac{1}{50}}\right),
	\]
	which completes the proof for the dense case.
\end{proof}

Therefore we have bounded the probability that each vertex participates in a test graph $H$ introduces errors in the map constructed by the algorithm (that is, which appears more than once in $\Base$ and survives subsampling in both graphs).

\subsection{Many vertices appear in test subgraphs}\label{sec:coverin}

To guarantee that the algorithm recovers a $\tilde\Omega(1)$ fraction of $\pi^*$, we will show that many vertices in $\Base$ are incident on a copy of $H$ which is present in {\em both} $G_0$ and $G_1$.
It will be convenient for us to think of the {\em intersection graph} of $G_0$ and $G_1$:
\begin{definition}[Intersection graph]\label{def:intersection}
    For $G_0,G_1 \sim \dstruct$, the {\em intersection graph} $\Inter$ is the subgraph of the base graph $\Base$ induced by edges that were sub-sampled in both $G_0$ and $G_1$.
    Note that the distribution of intersection graph is the same as $\mathbb{G}(n, p\noise^2)$.
\end{definition}
We will lower bound the probability that a vertex is contained in at least one test graph inside the intersection graph.

\begin{lemma}\label{lemma:recovery cover vertex}
    Define $q := p\gamma^2$, and let $\Inter \sim G(n,q)$.
    Furthermore let $\cH$ be a test set with $H \in \cH$ having $v$ vertices and average degree $d' = d + \lambda$ for an integer $d = 2$ or $d \ge 6$, and $\lambda \in [0,1]$,
    so that $nq^{d'/2} < 1$ and $|\cH|\cdot v n^{v-1} q^e \ge 1$.
    If $d = 2$, then we also require $\lambda \in(0 ,\frac{1}{76}]$, $v\log^2 v\ll \lambda^2 \log^2 n$, $v^8 nq^{d'/2} \ge 1$, and
	$v^2 = o( n q^{1 + \frac{1}{3}\lambda})$.
    If $d\ge 6$, then we require that $2v^{2+2(d'+2)} q^{\frac{1}{50}}\le v^{2(d'+2)} nq^{d'/2} \le \frac{1}{2}q^{-\frac{1}{50}}$, and $nq^{(1-\beta)\frac{1}{2}d'} \gg v^{2}$.

    For a vertex $u \in V(\Inter)$, let $N_u$ be the number of $H \in \cH$ that $u$ appears in in $G$.
    When these conditions are met, then
    \[
	\Var[N_u] \le o(\E[N_u]^2).
    \]
\end{lemma}

\begin{proof}

    Our proof is via the second moment method.
	First, the mean of $N_u$ can be calculated as follows:
\begin{align}
    \E[N_u]
    = \E\Big[\sum_{H\in\cH}\sum_{\substack{S\subseteq V,\\|S|=v,\\u\in S}} \isH{S}{H} \Big]
    = |\cH| \cdot \sum_{\substack{S\subseteq V,\\|S|=v,\\u\in S}} \E[\isH{S}{H} ]
    &= |\cH| \cdot v\cdot\ff{(n-1)}{v-1} q^{e},\label{eq:mean}
\end{align}
    since there are $v$ choices for the position of $u$ in $H$, then $\ff{(n-1)}{v}$ choices for the (ordered) remaining vertices of $H$, as $H$ has no non-trivial automorphisms.

    Now, we will bound the second moment.
    We expand the expression in terms of the indicators $\isH{S}{H}$, the indicator that the ordered vertex subset $S \subset V$ which includes $u$ has $H$ as an edge-induced subgraph in $\Inter$.
\begin{align}
    \E[N_u^2]
    = \E\Big[\Big(\sum_{H\in\cH}\sum_{\substack{S\subseteq V,\\|S|=v,\\u\in S}}\isH{S}{H}\Big)^2 \Big]
    &= \E\Big[\sum_{H,H'\in\cH}\sum_{\substack{S,S'\subseteq V,\\|S|=|S'|=v,\\u\in S,S'}}\isH{S}{H}\cdot \isH{S'}{H'} \Big]\nonumber\\
    &= \sum_{H\in\cH}\sum_{\substack{S,S'\subseteq V,\\|S|=|S'|=v,\\u\in S,S'}}\E[\isH{S}{H}\isH{S'}{H}] + \sum_{H\neq H'\in\cH}\sum_{\substack{S,S'\subseteq V,\\|S|=|S'|=v,\\u\in S,S'}}\E[\isH{S}{H}\isH{S'}{H'}].\label{eq:recovery cover 1}
\end{align}
In the final step, we have split the terms corresponding to products of indicators of the same subgraph $H$, and products of indicators of distinct subgraphs $H,H'$.
We will bound each of these summations individually.

\paragraph{Contribution of copies of the same subgraph $H$.}

	The first term of~\cref{eq:recovery cover 1} is the second moment of the number of labeled copies of $H$ incident on $u$.
	To bound the expectation $\E[\isH{S}{H}\isH{S'}{H}]$, we need to consider the overlap of the ordered vertex sets $S$ and $S'$ (and specifically, the overlap of the copies of $H$ on those sets).
	For each vertex-induced subgraph $J \subseteq H$, we may have the copies of $H$ on $S$ and $S'$ overlapping on $J$.
	In that case, the indicators $\isH{S}{H}, \isH{S'}{H}$ are not independent, and we have
	\[
	    \E[\isH{S}{H}\isH{S'}{H}] = q^{2e - |E(J)|}.
	\]
	Recall that $\num{J}{H}$ denotes the subgraph count of $J$ in $H$.
	Taking this into account,
	\begin{align}
	    \sum_{\substack{S,S' \subset V\\|S| = |S'| = v\\ u \in S,S'}}\E[\isH{S}{H}\isH{S'}{H}]
	    \, =\, \sum_{\substack{J \subseteq H\\|V(J)|\ge 1}}\sum_{\substack{S,S'\subset V\\ |S| = |S'| = v\\ u\in S,S',\, S\cap S' = J}} q^{2e - |E(J)|}
	    &\le \sum_{\substack{J \subseteq H\\|V(J)|\ge 1}} \num{J}{H}^2\cdot\aut(J) \cdot \ff{n}{2v -1 - |V(J)|}\cdot q^{2e - |E(J)|}\nonumber\\
	    &\le n^{2v-2}q^{2e}\cdot \sum_{\substack{J \subseteq H\\|V(J)|\ge 1}} \num{J}{H}^2\cdot\aut(J) \cdot n^{-|V(J)|+1}\cdot q^{-|E(J)|} \label{eq:same-expand}
	\end{align}
	where we have $\num{J}{H}^2$ ways that $J$ can appear as a subgraph of $H$ in $S$ and $S'$ and create a collision, $\aut(J)$ ways that maps $J$ from a copy to another, and for each such occurrence we choose $\ff{(n-1)}{2v - 1 - |V(J)|}$ vertices for the vertices of $S,S'$; $2v-1$ initially, since we are forced to include $u$, and then subtracting the vertices in $J$.

	As in the proof of \cref{lem:recovery uniquness single subgraph}, since $nq^{e/v} \ll 1$ we must carefully bound \cref{eq:same-expand} so that the densest $J$ do not cause it to blow up.
	Since we use two distinct constructions for $\cH$, we require different arguments depending on the density of $H$.
	We use \cref{lem:sparseequal,lem:denseequal}, mentioned above and proven in \cref{sec:moments}.

	It remains to bound the contribution of $H \neq H'$.

\paragraph{Contribution of distinct subgraphs $H,H'$.}
For a fixed pair $H,H' \in \cH$, expressing the second term \cref{eq:recovery cover 1} as a sum over subgraphs $J$ as before, we have
\begin{align}
    \sum_{\substack{S,S'\subseteq V,\\|S|=|S'|=v,\\u\in S,S'}}\E[\isH{S}{H}\isH{S'}{H'}]
     &= \sum_{\substack{J\subseteq H,H'\\|V(J)| \ge 1}}\sum_{\substack{S,S'\subseteq V,\\|S|=|S'|=v,\\u\in S,S',\, S\cap S'=J}}\E[\isH{S}{H}\isH{S'}{H'}]\nonumber\\
     &= \sum_{\substack{J\subseteq H,H'\\|V(J)| \ge 1}}\sum_{\substack{S,S'\subseteq V,\\|S|=|S'|=v,\\u\in S,S',\, S\cap S'=J}} q^{2e - |E(J)|}\nonumber\\
     &\le \sum_{\substack{J\subseteq H,H'\\|V(J)| \ge 1}}\num{J}{H}\cdot \num{J}{H'}\cdot\aut(J)\cdot \ff{n}{2v -1 -|V(J)|}\cdot q^{2e - |E(J)|},\label{eq:left off}
\end{align}
where to obtain the last inequality we have used that there are at most $\num{J}{H} \cdot \num{J}{H'}$ ways for $S,S'$ to intersect on $J$, $\aut(J)$ ways that maps $J$ from a copy to another, and there are at most $\ff{(n-1)}{2v - 1 - |V(J)|}$ choices of vertices for $S,S'$ once we are forced to include $u$.

To bound \cref{eq:left off}, we need to use the property of our graph family that for every $H\neq H'$ inside the family, the density of the densest common subgraph is appreciably less than the density of $H$.

Again because our guarantees differ depending on the density of $H$, we bound this quantity separately for the case when $H$ has density at most $1.5$ and the case when $H$ has density at least $3$.
\begin{restatable}{lemma}{sparseneq}\label{lem:sparseneq}
	Let $d' = 2 + \lambda$ for $\lambda \in (0,\frac{1}{76}]$.
	Suppose that $v^2 = o( n q^{1 + \frac{1}{3}\lambda})$.
	Then for distinct $H,H' \in \cH_{d'}^v$,
	\[
	\sum_{\substack{J \subseteq H,H'\\|V(J)|\ge 1}} \num{J}{H}\cdot \num{J}{H'}\cdot\aut(J)\cdot n^{-|V(J)|+1}q^{-|E(J)|}
    \le v^2 + nq^2(1 + o(1)).
	\]
\end{restatable}

\begin{restatable}{lemma}{denseneq}\label{lem:denseneq}
    Let $d' = d + \lambda$ for $d \ge 6$ and $\lambda \in [0,1]$.
    Define $\beta =  (\lambda + \frac{1}{26})\frac{1}{d'}$, and suppose that $nq^{(1-\beta)\frac{1}{2}d'} \gg v^{2}$.
    Then for distinct $H,H' \in \cH_{d'}^v$,
    \[
	\sum_{\substack{J \subseteq H,H'\\|V(J)|\ge 1}} \num{J}{H}\cdot \num{J}{H'}\cdot\aut(J)\cdot n^{-|V(J)|+1}q^{-|E(J)|} \le v^2 + nq^{d'/2}\cdot (1+o(1)).
    \]
\end{restatable}
We will prove these lemmas in \cref{sec:moments}; with these bounds in hand, we complete the proof of the theorem.

\paragraph{Bounding the variance of $N_u$ when $H$ is sparse.}
We are now equipped to bound the variance of $N_u$.
We return to \cref{eq:same-expand}; we wish to apply \cref{lem:sparseequal}, which bounds this sum for all $J \neq H$, and to it we will add the term for $J = H$, which gives an additional $n^{-v}q^{-e}$.
So When $H$ has density at most $1 + \frac{1}{152}$,
\[
\cref{eq:same-expand} \le n^{2v-1} q^{2e}\left(\frac{1}{n^v q^e} + \sum_{s=1}^{(1-\theta) v}\left(\frac{v^2 q^{\frac{1}{100}\theta \lambda}}{nq^{1+\frac{\lambda}{2}}}\right)^s + \frac{1}{n^v q^e}\cdot \theta v (v^8 n q^{1+\frac{\lambda}{2}})^{\theta v} q^{\frac{1}{100}\lambda(1-\theta)}\right).
\]
where we have applied \cref{lem:sparseequal} under the assumption that $v^8 nq^{1+\frac{\lambda}{2}} \ge 1$.
Now, we choose $\theta = \frac{800 \log v}{\lambda \log n}$; at this parameter setting, using that $q = n^{\delta - 1}$ for $\delta < \frac{1}{2}$,
\[
\log \left(v^2 q^{\frac{1}{100}\theta\lambda} \right)
= 2 \log v - \frac{1}{100}\theta \lambda \cdot (1-\delta) \log n
< \left(2 - 4\right)\log v < 0,
\]
so the sum is at most $\frac{1}{n^v q^e}$.
Also, for the third term, using that $nq^{1 + \lambda/2} < 1$ and that $(1-\delta)(1-\theta) > \frac{1}{2}$,
\begin{align*}
\log\left(\theta v \cdot (v^8 n q^{1+\frac{1}{2}\lambda})^{\theta v} q^{\frac{1}{100}\lambda(1-\theta)}\right)
    &\le \log\left(\theta v \cdot v^{8\theta v}\cdot q^{\frac{1}{100}\lambda(1-\theta)}\right)\\
    &= \log \theta v + 8\theta v\log v - \frac{1}{100}\lambda(1-\delta)(1-\theta) \log n\\
&\le \log \theta v + \frac{6400}{\lambda\log n}v\log^2 v - \frac{1}{200}\lambda \log n
< 0,
\end{align*}
where the final inequality follows for $n$ sufficiently large given our assumption that $v \log^2 v \ll \lambda^2 \log^2 n$.
It follows that $\cref{eq:same-expand} \le 3 n^{v-1}q^e$.

Thus, using \cref{eq:mean,eq:recovery cover 1,eq:same-expand,eq:left off} together with \cref{lem:sparseequal} and \cref{lem:sparseneq} we have
\begin{align}
    \E[N_u^2] - \E[N_u]^2
    &\le |\cH|^2 \cdot n^{2v -2} q^{2e}\left(v^2 + nq^{2}(1+o(1))\right) +|\cH|\cdot 3 n^{v-1}q^{e} - |\cH|^2 v^2 n^{2v-2} q^{2e}\\
    &\le |\cH|^2 \cdot n^{2v -2}q^{2e} \cdot n q^2 (1+o(1))+ 3\cdot |\cH|\cdot n^{v-1}q^{e},
\end{align}
where we have subtracted the third term from the first term and simplified the second term.
Because we have set the density so that $nq^2 \ll nq^{d'/2} \ll 1$, and because $\E[N_u] = |\cH| \cdot v\cdot \ff{(n-1)}{v-1} \cdot q^e$, the first term has magnitude $o(\E[N_u]^2)$.

We compare the second term against $\E[N_u]^2$; we have that
\[
    \frac{4\cdot |\cH| \cdot n^{v-1} q^e}{\E[N_u]^2} \le \frac{6\cdot|\cH| n^{v-1}q^e}{v^2 |\cH|^2 n^{2(v-1)} q^{2e}}
    = \frac{6n}{v^2}\cdot \frac{1}{|\cH|n^{v} q^e}
    \le o(1),
\]
by assumption.
Thus, if these conditions are satisfied,
\[
    \Var(N_u) = o(\E[N_u]^2),
\]
which completes the proof for sparse $H$.

\paragraph{Bounding the variance of $N_u$ when $H$ is dense.}
When $H$ has density at least $3$, using \cref{eq:mean,eq:recovery cover 1,eq:same-expand,eq:left off} together with \cref{lem:denseequal} (to which we add the term for $J = H$, which contributes $n^{-v}q^{-e}$ and dominates the contribution of the other terms) and \cref{lem:denseneq} (we meet the assumptions since we have assumed $v^2 < nq^{d'/2 - 1/50} < nq^{(1-\beta)d'/2}$) we have
\begin{align*}
    \E[N_u^2] - \E[N_u]^2
    &\le |\cH|^2 \cdot n^{2v -2}q^{2e}\left(v^2 + nq^{d'/2}(1+o(1))\right)
     +|\cH|\cdot n^{2v -1}q^{2e}\cdot 4 \left(\frac{1}{nq^{d'/2}}\right)^v      - |\cH|^2 v^2 n^{2v - 2} q^{2e}\\
    &\le |\cH|^2 \cdot n^{2v -2}q^{2e}\cdot (nq^{d'/2})\cdot(1+o(1)) + |\cH|\cdot n^{2v -1}q^{2e} \cdot 4 \left(\frac{1}{nq^{d'/2}}\right)^v,
\end{align*}
where to obtain the final line we have subtracted the third term from the first term.
Because we have set the density of $H$ so that $nq^{d'/2} = o(1)$, and since $\E[N_u] = |\cH|\cdot v\cdot \ff{(n-1)}{v-1}\cdot q^e$, the first term has magnitude $o(\E[N_u]^2)$.
To bound the second term, using \cref{eq:mean}, we have
\begin{align*}
    \frac{|\cH|\cdot n^{2v -1}q^{2e} \cdot 4 \left(\frac{1}{nq^{d'/2}}\right)^v }{\E[N_u]^2}
    = \frac{4 n \left(\frac{1}{nq^{d'/2}}\right)^v }{|\cH|\cdot v^2}
    = \frac{4n}{v^2}\cdot \frac{1}{|\cH| n^v q^e}
    \le o(1),
\end{align*}
by assumption.

It follows that if these conditions are satisfied,
\[
    \Var(N_u) = o(\E[N_u]^2),
\]
which completes the proof for dense $H$.
\end{proof}

From~\cref{lemma:recovery cover vertex}, we know that with good probability each vertex $u$ will be covered by at least one graph from the test set $\cH$ in the intersection graph.

\subsection{Partial solution}\label{sec:partial}
The previous two lemmas allow us to argue that step~\ref{step:match} of \cref{alg:recover} what we call a {\em partial solution}.
We start by defining the notion of a $(\theta,\eta)$ partial solution, which corresponds to getting $\eta$-accurate information about a $\theta$ fraction of the vertices:

\begin{definition}[Partial solutions] \label{def:partial}
If $(G_0,G_1,\pi^*)$ are sampled from $\dstruct(n,p;\gamma)$, and $s \in \{1,\ldots, n\}$  and $0 < \eta \leq 1$ are some constants then we say that a partial function $\pi:[n] \rightarrow [n]$ is an \textit{$(s,\eta)$ partial solution} if $\pi$ is a one-to-one function that is defined on at least $s$ inputs, and such that for at least $\eta$ fraction of the inputs $u$ on which $\pi$ is defined, $\pi(u)=\pi^*(u)$.
\end{definition}

\begin{lemma}\label{lem:partial recovery}
    Suppose that $G_0,G_1 \sim \dstruct(n,p;\gamma)$.
    Then under the conditions of \cref{thm:recovery} with probability $1-o(1)$ over the choice of $G_0,G_1\sim\dstruct(n,p;\gamma)$ and the randomness of the algorithm, step~\ref{step:match} of \cref{alg:recover} recovers a $(\frac{n}{\log v},1-\frac{1}{v^{1/8}})$-partial solution.
\end{lemma}
\begin{proof}
    We begin by noting that the parameters chosen at the start of \cref{alg:recover} satisfy the requirements of \cref{lem:recovery uniquness single subgraph,lemma:recovery cover vertex}:
    \paragraph{Sparse case requirements.}
	    If $p = n^{\delta-1} \in [\frac{n^\eps}{n},\frac{n^{1/153}}{n}]$, we have chosen $d' = 2 + \lambda$ for $\lambda$ in the range
    $
	\lambda \in \left(\frac{2\delta}{1-\delta}, \frac{2\delta}{1-\delta} +  \frac{\log \log n}{\log n}\right).
    $
    It can be verified that $\lambda \le \frac{1}{76}$.
We have chosen $v = \Theta(\log n)$ to be the smallest even integer so that $\lambda v$ is also an integer.
	    We must verify the following conditions:
	    \begin{itemize}
		\item For sufficiently large $n$, there exists such a choice of $\lambda$ and $v$ so that $v$ and $\lambda v$ are integers:

		   We have that so long as $v = \Theta(\log n)$, the interval of choices of $\lambda$ is such that there is at least one value such that $\lambda v$ is an integer.

		\item $v \log^2 v \ll \lambda^2 \log^2 n$.
		    We have required that $\delta \ge \sqrt{\frac{(\log\log n)^4}{\log n}}$, from which this follows.
		    This condition is used by~\cref{lem:recovery uniquness single subgraph}.
		\item $np^{d'/2} < 1$.

		    Using that $\frac{d'}{2} = 1+\frac{1}{2}\lambda$ and $p = n^{\delta -1}$, taking logarithms we have that this is equivalent to
		    \[
			\log n - (1-\delta)(1+\frac{1}{2}\lambda)\log n = (\delta -\frac{1}{2}\lambda + \frac{1}{2}\delta \lambda) \log n < 0 \iff \frac{2\delta}{(1-\delta)} < \lambda.
			\]
			For our choice of $\lambda$, this requirement is satisfied.
			This condition is used by~\cref{lem:recovery uniquness single subgraph}.
		\item $v^2 \ll nq^{1 + \frac{1}{3}\lambda}$.

		    Again taking a logarithm, we have that this condition is equivalent to
		    \[
			\log n - (1+\frac{1}{3}\lambda)(1-\delta)\log n - 2\log v \gg 0,
			~\Longleftarrow~
			\lambda < \frac{3\delta}{1-\delta} - \frac{6\log v}{\log n} -\Omega\left(\frac{\log v}{\log n}\right).
			\]
			Since we have required $\delta \ge \sqrt{\frac{(\log\log n)^4}{\log n}} \gg \frac{\log v}{\log n}$ and $\lambda < \frac{2\delta}{1-\delta} + O\left(\frac{\log v}{\log n}\right)$, we satisfy this condition.
			This condition is used by~\cref{lemma:recovery cover vertex} and ~\cref{lem:sparseneq}.

		\item $nq^{d'/2} v^8 \ge 1$.

		    Again taking logarithms, we have that this condition is equivalent to
		    \[
			\left(\delta - \frac{1}{2}\lambda + \frac{1}{2}\lambda \delta\right)\log n - (2+\lambda)\log \frac{1}{\gamma} + 8\log v> 0
		    \]
		    which is in turn implied by
		    \[
			\log \frac{1}{\gamma} \ll \log v,\quad \text{and}\quad \left(\delta - \frac{1}{2}\lambda + \frac{1}{2}\lambda \delta\right) > -1.9\frac{\log v}{\log n},
		    \]
		    where the latter condition can be strengthened to
		    \[
			\frac{2\delta}{1-\delta - 1.9\frac{\log v}{\log n}} \ge \frac{2\delta }{1-\delta} + 1.9\frac{\log v}{\log n}> \lambda.
		    \]
		    This latter condition is met by our choice of $\lambda$.
		    By assumption, $\gamma^{-1} = o(\log^c n)$ for any constant $c$, so the first condition is met as well.
		    This condition is used by~\cref{lem:recovery uniquness single subgraph},~\cref{lem:sparseequal}, and~\cref{lemma:recovery cover vertex}.
		    
		\item $1 \le |\cH| \cdot v n^{v-1} q^{\frac{d'}{2} v}$.

		    We first address the lower bound.
		    By \cref{prop:small-d}, we can take $|\cH|$ to be of size up to $(\lambda v)^{c\lambda v}$ for $c = \frac{69}{100}$.
		    Taking logarithms, this is equivalent to
		    \begin{align*}
			0 &< c \lambda v \log(\lambda v) + \log v + (v-1)\log n - v(1-\delta)(1+\frac{1}{2}\lambda)\log n - v(2+\lambda) \log \frac{1}{\gamma} \intertext{which is in turn implied by}
			0 &< c\lambda \left(\log v - \log \frac{1}{\lambda}\right) + \frac{\log v}{v} + \left(\delta - \frac{1}{2} \lambda+ \frac{1}{2}\lambda \delta - \frac{1}{v}\right)\log n - (2+\lambda)\log \frac{1}{\gamma},
		    \end{align*}
		    And this holds when
		    \[
			\log\frac{1}{\gamma} \ll \log v, \quad \frac{\log n}{v} \ll \log v, \quad \log{\frac{1}{\lambda}} \ll \log v,\quad \text{and} \quad \left(\delta - \frac{1}{2}\lambda + \frac{1}{2}\lambda \delta \right) > -0.9 \cdot c\lambda\frac{\log v}{\log n}.
		    \]
		    And (since $\delta$ is sufficiently small) the latter condition can be strengthened as follows:
		    \[
			\frac{2\delta}{(1-\delta - 1.8 c\frac{\log v}{\log n})} > \frac{2\delta }{1-\delta} + 1.8\cdot c\frac{\log v}{\log n}> \frac{2\delta }{1-\delta} + \frac{\log v}{\log n}\ge \lambda,
		    \]
		    This is met by our choice of $\lambda$, and we also have $\lambda^{-1} \le O(\delta^{-1}) = o\left(\sqrt{\frac{\log n}{(\log\log n)^4}}\right)$.
		    Therefore, we satisfy the lower bound.
\
		    In fact, under these conditions, we satisfy the following, stronger condition:
		    \begin{claim}
			In the sparse case when $d' = 2 + \lambda$, $|\cH|v n^{v-1}q^{d'/2} \ge v^{\lambda v/100}$.
		    \end{claim}
		    This condition is used by~\cref{lemma:recovery cover vertex}.
		\item $p = n^{\delta - 1}$ for $\delta < \frac{1}{3}$.
		    This holds by assumption.
		    This condition is used by~\cref{lem:recovery uniquness single subgraph}.
	    \end{itemize}

    \paragraph{Dense case requirements.}
    If $p=n^{\delta-1}\in[n^{-1/3},n^{-\epsilon}]$, we have chosen an even integer $v = \Theta(\log n)$ and $d'\in \left(\frac{2}{1-\delta}, \frac{2}{1-\delta}+\frac{\log\log n}{4\log n}\right)$ so that $d'v$ is an integer.
    It can be verified that $d' \ge 6$.
    We must verify the following conditions:
    \begin{itemize}
	\item There exists a choice of $d'$ so that $d'v$ is an integer.

	    We first notice that this implies that $(d'-\lfloor d' \rfloor)v$ is an even integer; this is because we have chosen $v$ even, and because $\lfloor d'\rfloor v$ is also an integer.

	    Now, we see that so long as $v = \Omega(\log n)$, for $n$ sufficiently large there is at least one integer in the allowed interval of $d' v$.
        \item $np^{d'/2} < 1$.

	    Taking the logarithm, we equivalently require that
	    \[
		\log n - (1-\delta)\frac{d'}{2}\log n < 0 \iff d' > \frac{2}{1-\delta}.
		\]
		Our choice of $d'$ satisfies this requirement.
	    When $v \ll n$ and $\frac{d'}{2}(1-\delta) - 1 = O(\frac{\log v}{\log n})$, as in our case, these inequalities are easily satisfied.
	    This condition is used by~\cref{lem:recovery uniquness single subgraph}.
        \item $nq^{(1-\beta)\frac{1}{2}d'} \gg v^{2}$.
        
        Take logarithm, this is equivalent to requiring that 
        \[
        \log n-(1-\delta)\frac{d'}{2}\log n +\frac{1}{2}(1-\delta)(\lambda+\frac{1}{26})\log n\gg\log v.
        \]
        Which holds by our choice of $\lambda$ and $d'$. This condition is used by~\cref{lemma:recovery cover vertex} and~\cref{lem:sparseneq}.
        \item $2v^{2+2(d'+2)} p^{\frac{1}{50}}\le v^{2(d'+2)} np^{d'/2} \le \frac{1}{2}p^{-\frac{1}{50}}$.

	    Taking the logarithm, this is equivalent to requiring that
	    \[
		1 + \left(2+2(d'+2)\right)\log v - \frac{1}{50}(1-\delta)\log n < \left(2d'+4\right) \log v + \log n - \frac{d'}{2}(1-\delta)\log n <- 1 + \frac{1}{50}(1-\delta)\log n.
	    \]
	    Which, using that $\frac{d'}{2}(1-\delta) - 1 = O\left(\frac{\log v}{\log n}\right)$, is equivalent to
	    \[
		1 + \left(3+2(d'+2)\right)\log v - \frac{1}{50}(1-\delta)\log n < \pm O(\log v) <- 1 + \frac{1}{50}(1-\delta)\log n.
	    \]
	    Using that $\delta < 1-\epsilon$ for $\epsilon$ fixed as a function of $n$, and that $v = O(\log n)$, gives the desired conclusion.
        \item $2v^{2+2(d'+2)} q^{\frac{1}{50}}\le v^{2(d'+2)} nq^{d'/2} \le \frac{1}{2}q^{-\frac{1}{50}}$.

	    Since we have assumed that $\gamma = o(\log^c n)$ for any constant $c$, and because $p = n^{\delta-1}$ for $\delta$ a constant, the previous condition implies this one as well. This condition is used by~\cref{lem:recovery uniquness single subgraph},~\cref{lem:denseequal}, and~\cref{lemma:recovery cover vertex}.

	\item $1 \le |\cH|\cdot v n^{v-1} q^{\frac{d'}{2}v}$.

	    By \cref{prop:large-d}, we can take $|\cH| = v^{C_d\cdot v}$, where $C_d = \left(\frac{d'}{4}(1-\beta) - \frac{1}{2}\lambda\right)$ and $\beta = (\lambda + \frac{1}{26})\frac{1}{d'}$.
	    When $d' \ge 6, C_d > \frac{1}{2}$.
	    Taking the logarithm and dividing by $v$, we have that the lower bound is equivalent to
	    \[
		0 < C_d \log v + \frac{\log v}{v} + (1-\frac{1}{v})\log n - \frac{d'}{2}(1-\delta)\log n - \frac{d'}{2}\log \frac{1}{\gamma}.
	    \]
	    This condition is met so long as
	    \[
		\frac{\log v}{ v} \ll \log v,
		\quad \frac{\log n}{v}\ll \log v,
		\quad \log \frac{1}{\gamma} \ll \log v,
		\text{and}\quad \left(\frac{d'}{2}(1-\delta) - 1\right) \le \frac{\log v}{4\log n}< \frac{1}{2}C_d\frac{\log v}{\log n};
	    \]
	    The first three conditions we have by assumption, and the latter condition can be strengthened to
	    \[
		d' < \frac{2}{1-\delta} + \frac{\log v}{2\log n},
	    \]
	    which our condition on $\lambda$ satisfies.

	    In fact, we can deduce the following stronger consequence:
	    \begin{claim}
		In the sparse case when $d' \ge 6$, $|\cH| v n^{v-1} q^{d'/2} \ge v^{\frac{1}{2}C_d v} \ge v^{\frac{1}{4} v}$.
	    \end{claim}
	    This condition is used by~\cref{lemma:recovery cover vertex}.
    \end{itemize}

    \medskip

    \paragraph{The algorithm returns a partial solution.}
    For each $u \in V$, let $N_u$ be the number of distinct $H \in \cH$ which contain the vertex $u$ and which appear in both $G_0,G_1$.
    In both the sparse and dense regimes, we have that $\E[N_u] = |\cH| \cdot v n^{v-1} q^e$, which we have set to be at least $v^{\Omega(v)}$ (where the $\Omega$ hides a dependence on $\lambda$ in the sparse case).

    \begin{claim}\label{claim:single}
	For $u \in V$,
	let $N_u$ be the number of $H \in \cH$ that contain $u$, so that $H$ survives into both $G_1$ and $G_2$.
	Then $\Pr[N_u < \frac{1}{2}\E[N_u]] = o(1)$.
    \end{claim}
    \begin{proof}
	By \cref{lemma:recovery cover vertex}, $\Var(N_u) = o(\E[N_u])$.
	The claim follows by Chebyshev's inequality.
    \end{proof}

    \begin{claim}\label{claim:high-frac}
	With probability $1-o(1)$, at least a $\frac{1}{\log v}$ fraction of vertices $u \in V$ appear in at least $\frac{1}{2}\E[N_u]$ distinct $H \in \cH$ that survive the subsampling in both $G_0,G_1$.
    \end{claim}
    \begin{proof}
	Let $\mathbb{I}_u$ be the indicator that $u$ is contained in at least $\frac{1}{2}\E[N_u]$ such $H$.
	With~\cref{claim:single} we have that $\sum_{u\in V}\E [\mathbb{I}_u]\geq (1-o(1))n$.
	We can then apply the following averaging argument to conclude that with probability $1-\theta - o(1)$, a $\theta$-fraction of the vertices are covered by at least one $H \in \cH$:
\begin{align*}
    (1-o(1))n &\leq \sum_{u\in V}\E_{G_0,G_1}[\mathbb{I}_u] = \E\left[\sum_{u\in V}\mathbb{I}_u\right]
\leq \theta\cdot n\Pr\left[\sum_{u\in V}\mathbb{I}_u<\theta n\right] + n\cdot\Pr\left[\sum_{u\in V}\mathbb{I}_u\geq\theta n\right].
\end{align*}
    Thus, $\Pr\left[\sum_{u\in V}\mathbb{I}_u\geq\theta n\right]\geq1-\theta -o(1)$ as desired.
    Taking $\theta = \frac{1}{\log v}$ gives us the desired conclusion.
    \end{proof}

    \begin{claim}\label{claim:low-error}
	For $u \in V$, let $B_u$ be the number of $H \in \cH$ that contain $u$, so that $H$ appears at least twice in the base graph $\Base$, and so that at least one copy of $H$ survives into both $G_0$ and $G_1$.
	Then with probability at least $1-o(1)$, at most $\frac{n}{v^{1/4}}$ of $u \in V$ have $B_u \ge \frac{1}{v^{1/4}} \E[N_u]$.
    \end{claim}
    \begin{proof}
	Let $\mathbb{J}_u$ be the event that $B_u \ge v^{-1/4} \E[N_u]$.
	From \cref{lem:recovery uniquness single subgraph} and Markov's inequality, we have that
	\[
	\Pr[\mathbb{J}_u]
	\le \frac{\E[B_u]}{v^{-1/4}\E[N_u]}
	\le \frac{v^{-1/2}\E[N_u]}{v^{-1/4}\E[N_u]}
	\le O(v^{-1/4}).
	\]

	Then by Markov's inequality we get
	$$
	\Pr\left[\sum_{u \in V} \mathbb{I}_u \ge \frac{n}{\log^2 v} \right] \leq O(v^{-1/4}\log^2 v),
	$$
	as desired.
    \end{proof}

    In our algorithm, independently for every vertex which participates in at least $\frac{1}{2}\E[N_u]$ distinct $H \in \cH$ that appear in both $G_0,G_1$, we choose a uniformly random $H \in \cH$ that it participates in, and match it based on that copy of $H$.

    Given these conditions, \cref{claim:high-frac} (using \cref{lemma:recovery cover vertex}) proves that step \ref{step:match} of \cref{alg:recover} will match at least an $\Omega(\frac{1}{\log v})$ fraction of the vertices of $G_0,G_1$; since the $H \in \cH$ have no non-trivial automorphisms, the matches are always correct (and one-to-one) unless $H$ is ``bad'', i.e. there is more than one copy of the $H$ that was chosen in the base graph $\Base$.

    By \cref{claim:low-error}, with high probability the proportion of such ``bad'' $H$ incident on each $u \in V$ in the intersection graph is at most $\frac{\log^2 v}{v^{1/4}}$; since we make the choice of mapping for each vertex independently, a Chernoff bound implies that for all but a ${v^{-1/8}}$ fraction of the matched vertices, we will make the match based on an $H$ that appears only once in the base graph $\Base$.
    This completes the proof.
\end{proof}

\subsection{Bounds on the order-2 moments of subgraph counts}\label{sec:moments}
Here we prove \cref{lem:sparseequal,lem:denseequal,lem:sparseneq,lem:denseneq}, each of which bounds correlations of the counts of subgraphs in the test sets $\cH$, in the sparse and dense cases respectively.

\subsubsection{Sparse case}
\sparseequal*
\begin{proof}
    We will make use of \cref{lem:qbalsparse}, which states that if $J$ has  $(1-\alpha)v$ vertices, then it has density at most $1 + \frac{\lambda}{2} - \frac{1}{100}\lambda \alpha$; taking $|V(J)| = s$ then gives that $|E(J)| \le (1 + \frac{\lambda}{2} - \frac{1}{100}\lambda(1 - \frac{s}{v}))s$. First, we are going to remove the $\aut(J)$ with the following lemma.
    
    \begin{lemma}\label{lem:qbalsparse removing aut}
    	Suppose that $H$ is a connected graph on $v\ge 1$ vertices with maximum degree $\Delta$.
    	Suppose also that $pv^2 < 1$, and that the densest subgraph $J^* \subset H$ on $s$ vertices has $\alpha s$ edges.
    	Then,
    	\[
    	\sum_{\substack{J \subset H, |V(J)| = s}} \num{J}{H}^2 \aut(J) \cdot n^{-|V(J)|} p^{-|E(J)|}
    	\le v^{2\Delta (v-s)} \sum_{J \subset H, |V(J)| = s}  \num{J}{H}^2\cdot p^{-\alpha |V(J)|} n^{-|V(J)|}
    	\]
    \end{lemma}
    
    We will prove~\cref{lem:qbalsparse removing aut} in the end of this subsubsection. Now, note that there are only $\binom{v}{s} \le \min(v^s/s!,v^{v-s})$ vertex induced subgraphs of $H$ of size $s$, $\aut(J)$ can be trivially upper bounded by $|V(J)|!$, and the maximum degree of $J$ is at most $3$.
    Therefore, splitting the sum by the size of the subgraph $J$,
    \begin{align*}
        \sum_{\substack{J \subset H\\|V(J)|\ge 1}} \num{J}{H}^2\cdot\aut(J) \cdot n^{-|V(J)|}\cdot q^{-|E(J)|}
        &\le \sum_{s=1}^{(1 -\theta) v} v^{2s} \cdot n^{-s} q^{-(1+\frac{\lambda}{2} - \frac{1}{100}\lambda(1-\frac{s}{v}))s} + \sum_{s=(1 -\theta) v + 1}^{v-1} v^{8(v-s)} \cdot n^{-s} q^{-(1+\frac{\lambda}{2} - \frac{1}{100}\lambda(1-\frac{s}{v}))s},
        \intertext{And taking the maximal density in each of the sums,}
        &\le \sum_{s=1}^{(1 - \theta)v} \left(\frac{v^2}{n q^{1 + \frac{1}{2}\lambda - \frac{1}{100}\lambda \theta}}\right)^{s} + v^{8v}\sum_{s=(1 - \theta) v + 1}^{v-1} \left(\frac{1}{v^8 n q^{1 + \frac{1}{2}\lambda - \frac{1}{100}\lambda \frac{1}{v}}}\right)^{s}
        \intertext{Since we have assumed that $v^8 nq^{1+\frac{1}{2}\lambda - \frac{1}{100}\lambda \frac{1}{v}} \ge 1$, the first term in the second sum dominates:}
        &\le \sum_{s=1}^{(1 - \theta)v} \left(\frac{v^2}{n q^{1 + \frac{1}{2}\lambda - \frac{1}{100}\lambda \theta}}\right)^{s} + v^{8v}\theta v \left(\frac{1}{v^8 n q^{1 + \frac{1}{2}\lambda - \frac{1}{100}\lambda \frac{1}{v}}}\right)^{(1 - \theta)v} \\
        \intertext{And finally simplifying both terms,}
        &\le \sum_{s=1}^{(1 - \theta)v} \left(\frac{v^2q^{\frac{1}{100}\lambda \theta}}{n q^{1 + \frac{1}{2}\lambda}}\right)^{s} + \theta v\left(\frac{1}{nq^{1 + \frac{1}{2}\lambda}}\right)^v \left(v^8n q^{1 + \frac{1}{2}\lambda}\right)^{\theta v} q^{\frac{1}{100}\lambda(1 - \theta)}.
    \end{align*}
    This gives us our conclusion.
\end{proof}

\sparseneq*
\begin{proof}
We split up the sum into $J$ corresponding to an isolated vertex, for which $\num{J}{H} = \num{J}{H'} = v$, and into $J$ with more vertices (and therefore possibly edges):
\begin{align*}
	\sum_{\substack{J \subseteq H,H'\\|V(J)|\ge 1}} \num{J}{H}\cdot \num{J}{H'}\cdot\aut(J)\cdot n^{-|V(J)|+1}q^{-|E(J)|}
     &\le v^2 + \sum_{k=2}^v\sum_{\substack{J\subseteq H,H'\\|V(J)| = k}}\num{J}{H}\cdot \num{J}{H'}\cdot\aut(J)\cdot n^{-|V(J)|+1}\cdot q^{2e - |E(J)|}\\
     \intertext{Since we have by \cref{prop:small-d} that $|E(J)| \le (1+\frac{1}{3}\lambda)|V(J)| - 2$,}
     &\le v^2 + \sum_{k=2}^v\sum_{\substack{J\subseteq H,H'\\|V(J)| = k}}\num{J}{H}\cdot \num{J}{H'}\cdot\aut(J)\cdot n^{-k+1}\cdot q^{- (1+\frac{1}{3}\lambda)k + 2}
     \intertext{and finally because there are at most $\ff{v}{k}$ vertex-induced subgraphs $J$ on $k$ vertices and each of them has at most $v^{|V(J)|}$ non-trivial automorphisms, we obtain the crude upper bound}
     &\le v^2 + nq^2\sum_{k=2}^v v^{2k}\cdot n^{-k}\cdot q^{- (1+\frac{1}{3}\lambda)k}
     \,=\, v^2 + nq^2\sum_{k=2}^v \left(\frac{v^{2}}{n\cdot q^{(1+\frac{1}{3}\lambda)}}\right)^k.
\end{align*}
    So long as $v^2 = o(n q^{1 + \frac{1}{3}\lambda})$, the conclusion holds.
\end{proof}

\begin{proof}[Proof of~\cref{lem:qbalsparse removing aut}]
Let us start with an observation on the number of non-trivial automorphsims for a vertex-induced subgraph of a graph that has no non-trivial automorphism.
\begin{claim}\label{claim:auto}
	Suppose that $H$ is a $v$-vertex graph with no non-trivial automorphisms, and let $J$ be a graph obtained from $H$ by removing $\ell$ edges.
	Then $J$ has at most $v^{2\ell}$ non-trivial automorphsims.
\end{claim}
\begin{proof}
	Let $S \subset V(J)$ be the ordered set of endpoints of the $\ell$ edges removed.
	We have $|S| \le 2\ell$. 
	Now suppose by way of contradiction that $J$ has more than $v^{2\ell}$ non-trivial automorphisms. 
	Then by a pigeonhole argument, there must be an ordered set of vertices $T \subset V(J)$ and two distinct automorphisms, $\pi$ and $\sigma$, such that $\pi(S) = \sigma(S) = T$.
	But then, $\pi \circ \sigma^{-1}$ is a non-trivial automorphism of $J$ that maps $S$ to $S$, and can therefore be extended to a nontrivial automorphism of $H$.
	This is a contradiction. 
\end{proof}

Let $\cS_{v-s}$ be the set of vertex-induced subgraphs of $H$ obtainable by removing $(v-s)$ vertices, and for each $J \in \cS_{v-s}$, let $\cS(J)_{\ell}$ be the set of edge-induced subgraphs obtainable from $J$ by removing $\ell$ edges.
If there is a subgraph $K \subset H$ which appears in more than one set $\cS(J)_{\ell}$, we remove some of its appearances so that the $\cS(J)_{\ell}$ form a partition.
Re-writing the sum on the left-hand side, we have
\begin{align}
\sum_{\substack{J \subset H \\ |V(J)| = s}} \num{J}{H}^2 \aut(J) \cdot n^{-|V(J)|} p^{-|E(J)|}
&= \sum_{J \in \cS_{v-s}} \sum_{\ell = 0}^{\Delta v/2} \sum_{K \in \cS(J)_{\ell}} \num{K}{H}^2 \aut(K) \cdot n^{-|V(K)|} p^{-|E(K)|}\\
&= \sum_{J \in \cS_{v-s}} \sum_{\ell = 0}^{\Delta v/2} \sum_{K \in \cS(J)_{\ell}} \num{K}{H}^2 \aut(K) \cdot n^{-|V(K)|} p^{-|E(J)| + \ell}\\
\intertext{Now we bound $\aut(K)$: if we remove $v-s$ vertices from $H$ this removes at most $\Delta (v-s)$ edges, and then we remove at most $\ell$ additional edges, which removes at most $\Delta (v-s) + \ell$ edges.
	Then applying \cref{claim:auto}:}
&\le \sum_{J \in \cS_{v-s}} \sum_{\ell = 0}^{\Delta v/2} \sum_{K \in \cS(J)_{\ell}} \num{K}{H}^2 v^{2(\Delta (v-s) + \ell)} \cdot n^{-|V(K)|} p^{-|E(J)| + \ell}\\
&\le \sum_{J \in \cS_{v-s}} \left(\max_{\ell \in [\Delta v/2]} (v^2p)^\ell \right)\sum_{\ell = 0}^{\Delta v/2} \sum_{K \in \cS(J)_{\ell}} \num{K}{H}^2 v^{2\Delta s} n^{-|V(K)|} p^{-|E(J)|} 
\intertext{and so long as $pv^2 < 1$, this is maximized for $\ell = 0$:}
&\le v^{2\Delta (v-s)}\sum_{J \in \cS_{v-s}} \sum_{\ell = 0}^{\Delta v/2} \sum_{K \in \cS(J)_{\ell}} \num{K}{H}^2  n^{-|V(K)|} p^{-|E(J)|} \\
&\le v^{2\Delta (v-s)}\sum_{J \subset H, |J| = s}  \num{J}{H}^2 \cdot p^{-\alpha|V(J)|} n^{-|V(J)|},
\end{align}
where in the final line we have used that the $S(J)_{\ell}$ form a partition, and that $\alpha$ is the maximum density of any subgraph $J$.
\end{proof}

	\subsubsection{Dense case}
	\denseequal*
\begin{proof}
    When $H\in\cH$ have average degree $d'= d + \lambda$ for $d > 2$, our construction is as follows: we sample a $d$-regular graph on $v$ vertices, then add a uniformly random matching on $\lambda v$ vertices.

    We will rely upon the following claim, which bounds the density of subgraphs of $H$.
    \begin{claim*}[consequence of \cref{claim:densebalance}]
	For $H \in \cH_{d'}^v$ with $d' = d + \lambda$ for $d \ge 6$ and $\lambda \in [0,1]$,
	For any $J \subseteq H$ of with $|V(J)| = (1-\theta)v$ for $\theta \in [0,1]$,
	\[
	    |E(J)| \le
	    \begin{cases}
		\left(\frac{d'}{2} - \frac{1}{50}\right)|V(J)| & \text{ if } \theta < \frac{1}{2},\\
		\frac{d'}{2}|V(J)| - \frac{1}{50}\theta v & \text{ if } \theta \ge \frac{1}{2}.
	    \end{cases}
	    \]
    \end{claim*}

    Thus, we can split $J$ according to size.
    We have that:
    \begin{align*}
	&\sum_{\substack{J \subseteq H}} \num{J}{H}^2\cdot\aut(J)\cdot n^{-|V(J)|}q^{-|E(J)|}\\
	&\le  \sum_{k=0}^{v/2-1}\sum_{\substack{J \subseteq H\\|V(J)| = k}}\num{J}{H}^2\cdot\aut(J)\cdot n^{-|V(J)|}q^{-|E(J)|}
	+ \sum_{k=v/2}^{v-1} \sum_{\substack{J \subseteq H\\|V(J)| = k}}\num{J}{H}^2\cdot\aut(J)\cdot n^{-|V(J)|}q^{-|E(J)|}.
    \end{align*}
    We first bound the first summation.
    We have that
    \begin{align*}
	\sum_{k=0}^{v/2-1}\sum_{\substack{J \subseteq H\\|V(J)| = k}}\num{J}{H}^2\cdot\aut(J)\cdot n^{-|V(J)|}q^{-|E(J)|}
	&\le \sum_{k = 0}^{v/2-1} \left(\frac{v^2}{n q^{\frac{d'}{2} - \frac{1}{50}}}\right)^{k}
	\le 2,
    \end{align*}
    where we have upper bounded $\num{J}{H}$ by $v!/|V(J)|!$ and upper bounded $\aut(J)$ by $|V(J)|!$ and obtained the final inequality used by our choice of $d'$, $nq^{\frac{d'}{2} - \frac{1}{50}} \ge 2v^2$.

    The second summation we bound in the same manner, but now using that if $k = (1-\theta)v$ then $\theta v = v - k$, and there are at most $\ff{v}{v-k}$ subgraphs of size $k$:
    \begin{align*}
	\sum_{k=v/2}^{v-1}\sum_{\substack{J \subseteq H\\|V(J)| = k}}\num{J}{H}^2\cdot\aut(J)\cdot n^{-|V(J)|}q^{-|E(J)|}
	&\le \sum_{k=v/2}^{v-1} v^{2(d'+1)(v-k)}\sum_{\substack{J \subseteq H\\|V(J)| = k}}\num{J}{H}^2 n^{-|V(J)|}q^{-\left(\frac{d'}{2}k - \frac{1}{50}(v-k)\right)}\\
	&\le \sum_{k=v/2}^{v-1} \frac{v^{2(d'+2)(v-k)}}{n^k q^{\frac{d'}{2}k - \frac{1}{50}(v-k)}}\\
	&\le v^{2(d'+2)v} q^{\frac{1}{50}v}\sum_{k=v/2}^{v-1} \left(\frac{1}{v^{2(d'+2)} \cdot n\cdot q^{(\frac{d'}{2} + \frac{1}{50})}}\right)^k\\
	&\le v^{2(d'+2)v} q^{\frac{1}{50}v}\cdot 2\left(\frac{1}{v^{2(d'+2)} \cdot n\cdot q^{(\frac{d'}{2} + \frac{1}{50})}}\right)^{v-1} = v^{2(d'+2)} q^{\frac{1}{50}}\cdot 2 \left(\frac{1}{nq^{d'/2}}\right)^{v-1}.
    \end{align*}
	Where to obtain the first inequality we have applied~\cref{lem:qbalsparse removing aut} and to obtain the final inequality we have bounded this geometric sum by its leading term where $k = v-1$, and have used that $v^{2(d'+2)} n q^{d'/2 + \frac{1}{50}} \le \frac{1}{2}$.

	Combining these two terms, the term from the second summation has larger magnitude than the term from the first summation, as we have assumed $(nq^{d'/2})^{v-1} \le v^{2(d'+2)} q^{1/50}$.
	    This completes the proof.
\end{proof}

\denseneq*

\begin{proof}
    We will use \cref{prop:large-d}, from which we have that for $\beta = (\lambda + \frac{1}{26})\frac{1}{d'}$, every subgraph $J \subseteq H,H'$ has at most $|E(J)| \le \frac{1}{2}(1-\beta)d'\cdot |V(J)| - \frac{d'}{2}$.
    Since $\frac{1}{2}(1-\beta)d' < \frac{1}{2}d' - \xi$, by applying the trivial bound $\num{J}{H} \le v^{|V(J)|}/|V(J)|!$ and $\aut(J)\leq|V(J)|!$ (and pulling out the special case $|V(J)| = 1$),
    \begin{align*}
	&\sum_{\substack{J \subseteq H,H'\\|V(J)|\ge 1}} \num{J}{H}\cdot \num{J}{H'}\cdot\aut(J)\cdot n^{-|V(J)|+1}q^{-|E(J)|}\\
	&\le v^2 + nq^{d'/2}
	\cdot \sum_{k=2}^v \sum_{\substack{J \subseteq H,H'\\|V(J)|=k}} \num{J}{H}\cdot \num{J}{H'}\cdot\aut(J)\cdot n^{-|V(J)|}q^{-(1-\beta)\frac{d'}{2}|V(J)|}\\
	&\le v^2 + nq^{d'/2} \cdot \left(\sum_{k=2}^v \left(\frac{v^2}{nq^{(1-\beta)\frac{d'}{2}}}\right)^{k-1}\right),
    \end{align*}
    and using that $v^2 \ll nq^{(1-\beta)\frac{d'}{2}}$, we have that the geometric sum is $1+o(1)$, which gives us our conclusion.
\end{proof}

\subsection{Boosting from a partial assignment}\label{sec:boost}

So far, we have shown that our algorithm partially recovers the ground truth permutation (with some errors).
In this section we show how to ``boost'' this to full recovery of the ground truth permutation.

There are many algorithms in the literature on the graph matching problem that deal with recovering the ground truth permutation from such a  partial matching on a subset of the vertices,  which is typically known as a \textit{seed set}~\cite{yartseva2013performance,lyzinski2014seeded,lyzinski2014other}.
These works typically study a simple percolation, which can be shown to succeed in a random graph~\cite{janson2012bootstrap}.
However, our setting is somewhat different:

\begin{itemize}
    \item These algorithms typically assume that the seed set is chosen at \textit{random}, while in our setting the seed set (i.e., partial solution)  is learned from the graph and has arbitrary correlation with it.

    \item These algorithms typically (though not always~\cite{kazemi2015growing}) assume that the permutation is known \textit{perfectly} on the seed set (i.e., that they are given an $(s,1)$ partial solution), while it will be more convenient for us to allow some probability of error.
\end{itemize}

For this reason, we reprove here that a percolation-like algorithm succeeds in our more general setting (albeit with worse bounds than those of \cite{janson2012bootstrap}).
Specifically, we show that we can boost an $(\tilde{\Omega}(n),1-o(1))$ partial solution to a the full ground truth:

\begin{lemma}[Boosting from partial knowledge] \label{lem:boost}
    Let $p,\gamma,n,\epsilon,c,\theta$ be such that $p \gamma n \ge \log^c n$ for $c > 1$, $\epsilon\theta = o(\gamma^2)$ and $\theta = \Omega(\log^{1-c} n)$.
    Then there is a polynomial-time algorithm $A$ such that with probability $1-o(1)$ over the choice of $(G_0,G_1,\pi^*)$ from $\dstruct(n,p;\gamma)$, if $A$ is given $G_0,G_1$ and any  $(\theta n,1-\epsilon)$ partial solution $\pi$, then it outputs the ground truth permutation $\pi^*$.
\end{lemma}

\cref{lem:boost} completes the proof of \cref{thm:recovery}. Note the order of quantifiers in \cref{lem:boost}: the partial solution $\pi$ can be \textit{adversarially chosen},  as long as it satisfies  the desired conditions.
The proof of \cref{lem:boost} is obtained by analyzing the following percolation algorithm.
This is the same  algorithm used in prior works except that we set our parameters in a particular way, and  the ``seed set'' can depend arbitrarily on the graph.
We will denote by $\Inter$ the ``intersection graph'' of $G_0$ and $G_1$ (see \cref{def:intersection}). Note that $\Inter$ is distributed according to $\bb{G}(n,q)$ where  $q=\noise^2 p$.
We let $c$ be some sufficiently large constant.

\begin{algorithm}[Boosting from a non-random seed].\label{algo:boosting}

    {\bf Input:} $(G_0,G_1,\pi^*)\sim \dstruct(n,p;\noise)$ with and a partial one-to-one map $\pi:V(G_0) \rightarrow V(G_1)$, defined on at least $\theta n$ vertices of $V(G_0)$, such that for all but  an $\epsilon$ fraction of the $u$'s on which $\pi$ is defined, $\pi(u)=\pi^*(u)$.

    {\bf Goal: } Recover $\pi^*$
    \begin{enumerate}

\item Let $u$ be a vertex of $G_0$ and $v$ be a vertex of $G_1$. We define the number of ``common neighbors'' of $u,v$, denoted as  $N(u,v)$, to be the number of $w$'s such that \textbf{(i)} $w$ is a neighbor of $u$, \textbf{(ii)}  $\pi$ is defined on $w$, and \textbf{(iii)} $\pi(w)$ is a neighbor of $v$.

\item Let $\Delta \in \mathbb{N}$ be a parameter, which we set as $\lfloor \theta \gamma^2 np/100 \rfloor$ where $\theta$ is the fraction of vertices on which $\pi$ is defined.

\item \textbf{Step 1: Completing the solution.} Repeat the following until we cannot do so anymore: If there is a pair $(u,v)$  where $\pi$ is not defined at $u$ with $N(u,v) \geq \Delta$ then define $\pi(u)=v$.

\item If the permutation is not yet complete at the end of step 1, complete it to a full permutation arbitrarily.

\item \textbf{Step 2: Fixing the solution.} We let $\Delta' = \lfloor \gamma^2 np/100 \rfloor$. Repeat the following until we cannot  do it anymore: if there exists a pair $u,v$ with $N(u,v) \geq \Delta'$ such that $N(u,\pi(u)),N(\pi^{-1}(v),v) \leq \Delta'/10$ then modify the permutation $\pi$ to map $u$ to $v$ and map $\pi^{-1}(v)$ to $\pi(u)$.

\item Output $\pi$

\end{enumerate}

\end{algorithm}

This algorithm will indeed run in polynomial time: In Step 1 we keep increasing the number of inputs on which $\pi$ is defined, and in Step 2 we increase in every iteration the potential $\sum_u N(u,\pi(u))$ and so we can only run in polynomial time.
Lemma~\ref{lem:boost} will follow from the following two claims:

\begin{claim}[Step 1 analysis] \label{clm:steponeboost}
With probability $1-o(1)$ over the choice of the graphs $(G_0,G_1)$, by the end of Step 1, the permutation $\pi$ is defined on all inputs, and there are at most $10\epsilon n$ vertices $v$ on which $\pi(v) \neq \pi^*(v)$.
\end{claim}

\begin{claim}[Step 2 analysis] \label{clm:steptwoboost}
With probability $1-o(1)$ over the choice of the graphs $(G_0,G_1)$, by the end of Step 2, the permutation $\pi$ will equal $\pi^*$.
\end{claim}

The following lemmas will be useful for our proof.

\begin{lemma}[Few common neighbors]\label{lem:recoveryboundedsharedneighbors}
    Let $(G_0,G_1,\pi^*)\sim\dnull(n,p;\gamma)$. Then with probability at least $1-o(1)$, for every pair $u$ and $v$ such that $v \neq \pi^*(u)$, the number of joint neighbors in $\Inter = G_0 \cap {\pi^*}^{-1}(G_1)$ of $u$ and $v$ is at most $10 (\gamma p)^2n\cdot \log n$.
\end{lemma}
\begin{proof}
    If $u\neq v$, then their neighbors in $\Inter$ are chosen independently at random and the expected number of them is $(\gamma p)^2 n$. The probability that this variable is $10 \log n$ times larger  than the expectation will be exponentially small in $10 \log n$ and so we can take a union over all $n^2$ pairs of vertices.
\end{proof}

In our setting of parameters, with $p  < n^{-\Omega(1)}$, it will always be the case that $10(p\gamma)^2n\log n  \ll \Delta = \lfloor\theta qn / 100 \rfloor$.

\begin{lemma}[Expansion] \label{lem:randomgraphexpansion} Let $\Inter \sim \bb{G}(n,q)$.
If $q  \geq c\log n /n$ then  then  with high probability over $\Inter$, for every non-empty set $S$, $|E(S,\overline{S})|=(1\pm \tfrac{10}{\sqrt{c}})q|S||\overline{S}|$.
\end{lemma}
\begin{proof}
Let $S$ with $|S|=\theta n$.
We assume $\theta \geq 1/2$ (otherwise move to the complement).
In a random graph, the random variable $X=|E(S,\overline{S})|$ is just the sum of at least $|S||\overline{S}|=\theta(1-\theta)n^2$ independent Bernoulli random variables each with expectation $q$, and so $\mu=\mathbb{E}[X] = \theta(1-\theta)qn^2$.
So, the probability that $|X- \mu| \geq \sqrt{k \theta q n^2}$ is exponentially small in $k$.
As long as $k$ is larger than $10\log \binom{n}{\theta n} \leq 10 \theta n \log n$ then we can take a union bound over all such sets $S$.
Hence the deviation will always be smaller than
$\sqrt{10 \log n \theta^2 q n^{3}} \leq  \theta \sqrt{q n  \log n} n  < 10\theta qn^2/\sqrt{c}$ as long as $q \geq c\log n/n$.
\end{proof}

We now turn to completing the analysis of the algorithm by proving   Claims~\ref{clm:steponeboost} and~\ref{clm:steptwoboost}.

\begin{proof}[Proof of \cref{clm:steponeboost}]
Let $S$ be the current set on which $\pi$ is defined, and let $T = [n] \setminus S$.
We write $S= A \cup B$ where $A$ is the ``good'' set, on which $\pi$ agrees with $\pi^*$, and $B$ is the ``bad'' set, on which $\pi$ disagrees with $\pi^*$. Under our assumptions, $|A| \geq \theta n/2$ and $|B| \leq \epsilon\theta n$.

Let $(u_1,v_1),\ldots,(u_k,v_k)$ be the pairs we add to the permutation $\pi$ during Step 1 in chronological order.
    We start by bounding the number of ``bad pairs'' where $\pi(u_i) \neq v_i$.
In such a bad pair $u_i$ and $v_i$ have $o(\Delta)$ shared neighbors under the ground truth $\pi^*$, and a necessary condition for  $N(u_i,v_i) \geq \Delta$ is for $u_i$ to have at least $\Delta/2$ neighbors that were themselves previously mismatched.
So, if there are $b$ ``bad pairs'', then it means that there is a sequence $u_1,\ldots,u_b$ such that each $u_i$ has at least $\Delta/2$ neighbors that are either in the original set $B$ or are of the form $u_j$ for $j<i$.
Suppose towards a contradiction that $b \geq |B|$. Then the set $B' = B \cup \{ u_1,\ldots, u_{|B|} \}$ has size $2|B| \leq 2\epsilon\theta n$, but half of the vertices in it have $\Delta/2$ edges going into $B'$.
Hence the average degree of the induced subgraph on $B'$ is at least $\Delta/4$.
Since the expected average degree would be $\epsilon \theta qn = o(\Delta)$ this corresponds to a deviation of about $\Delta |B'| = \Omega(\theta qn |B'|)$ in the number of edges between $B'$ and its complement.
If $qn \gg \log n/\theta$ then this contradicts \cref{lem:randomgraphexpansion}.

The above shows that the number of mistakes we make is at most $O(\epsilon \theta n)$ but we still need to rule out the possibility that we simply get ``stuck'' and can't find a pair $(u,v)$ with $N(u,v) \geq \Delta$.
Suppose towards a contradiction we get stuck after we have defined $\pi$ on a set $S' = A' \cup B'$ where $A'$ is the ``good'' set on which we defined $\pi$ correctly. and $B'$ is the bad set on which we defined it incorrectly.
Let $\alpha = |A'|/n$. We know that $\alpha \geq \theta$.
Also, the argument above showed that  $|B'| \leq 10\epsilon \theta n$.
There are about $q|A'|(n-|A'|)$ edges between $A'$ and its compliment, and so (since the degree of the graph is $qn$), there will be at least $\alpha/10$ fraction of the vertices in the compliment of $A'$ that have at least $\alpha q n/10$ edges into $A'$. 
This means that as long as $n-|A'| \geq 20\epsilon \theta n/\alpha$ (or equivalently, $\alpha(1-\alpha) \geq 20\epsilon \theta$), there will be more than $2|B'|$ vertices $u$ in the compliment of $A'$ that have more than $\Delta$ edges into $A'$ in the intersection graph $\Inter$.
Hence, even if we ignore all vertices $u$ such that either $u\in B'$ or $\pi^*(u) \in \pi(B')$, we still can find one vertex $u$ such that $N(u,\pi^*(u)) \geq \Delta$ and neither $u$ nor $\pi^*(u)$ were defined.

In particular this means that at the end of step 1, at most $20\epsilon \theta n$ vertices are undefined and for them we complete the permutation arbitrarily.
\end{proof}

\begin{proof}[Proof of \cref{clm:steptwoboost}]
In this step we assume we have an almost optimal permutation $\pi$ such that $\pi(v)=\pi^*(v)$ for all but $\delta$ fraction of $v$'s (for $\delta =O(\epsilon\theta)$)\footnote{In the previous version, we did not take care of the possibility of some weird correlation between the error pattern and the almost optimal permutation we are starting with.}.

The main idea is that with high probability over \textit{the choice of random graphs}, they will satisfy expansion properties that guarantee that our boosting algorithm does not get stuck.
Note that this is with respect to an  \textit{arbitrary initial error pattern}. That is, with high probability over the choice of random graphs, no matter which almost optimal permutation we are starting with, the algorithm will not get stuck.
Concretely, the following lemma shows that with high probability over the choice of random graphs, \textbf{Step 2: Fixing the solution} can keep making progress in fixing mismatched vertices:

\begin{lemma}\label{lem:boosting step 2 analysis}
	Let $(G_0,G_1,\pi^*)\sim\mathcal{D}_{null}(n,p;\noise)$ and $\delta\in(0,1/10)$. Then with probability at least $1-o(1)$ over the choice of $(G_0,G_1,\pi^*)$, for any permutation $\pi$ that agrees with $\pi^*$ on at least $1-\delta$ fraction, after one round of the \textbf{Step 2: Fixing the solution} in Algorithm 3.20, the fixed permutation will agree with $\pi^*$ on at least $1-\delta/2$ fraction. Further, if $\delta=o(p\noise)$, then the fixed permutation will equal $\pi^*$.
\end{lemma}

Before we prove~\autoref{lem:boosting step 2 analysis}, note that it addresses the reviewer's (valid!) concern in our original proof of Claim 3.22. The reason is that \textbf{Step 2: Fixing the solution} is able to keep making progress and fixes every mismatched vertices in the end.

\begin{proof}[Proof of~\autoref{lem:boosting step 2 analysis}]
	The main idea is based on the following expansion property of random graphs.
	\begin{claim}\label{claim:boosting expansion}
		Let $G\sim \mathbb{G}(n,q)$ where $q\geq20000\log n/n$. Then with high probability, for any subset $S$ of size $|S|\leq n/20000$, $|E(S,S)|\leq\frac{|S|\Delta'}{100}$ where $\Delta'=\floor*{\frac{qn}{100}}$.
	\end{claim}
	\begin{proof}[Proof of~\autoref{claim:boosting expansion}]
		Let $S$ with $|S|=\theta n$ where $\theta\leq1/100$. In random graph $G\sim\mathbb{G}(n,q)$, we let random variable $X=|E(S,S)|$ and $\mu=\mathbb{E}[X]=q\binom{\theta n}{2}$. By Chernoff's bound, we know that the probability that $X-\mu>\sqrt{k\theta^2n^2q}$ is exponentially small in $k$. Thus, as long as $k$ is larger than $10\log\binom{n}{\theta n}\leq10\theta n\log n$ then we can simply take an union bound over all possible $S$. As a result,
		\begin{align*}
		X&\leq\mu+\sqrt{10\theta n\log n\theta^2n^2q}\\
		&\leq \frac{qn|S|}{20000} + qn|S|\cdot\sqrt{\frac{10\theta\log n}{qn}}\leq\frac{|S|\Delta'}{100}
		\end{align*}
		with high probability.
	\end{proof}
	
	Next, apply averaging argument on~\autoref{claim:boosting expansion} and $G$ to be the \textit{intersection graph} and $S$ to be the set of mismatched vertices, there are at least $|S|/2$ vertices in $S$ having less than $\Delta'/50$ edges to vertices in $S$. Thus, \textbf{Step 2: Fixing the solution} would correct $|S|/2$ many vertices.
	
	As for the case where $\delta=o(p\noise)$, the size of mismatched vertices is $o(\Delta')$ and thus \textbf{Step 2: Fixing the solution} would fix all the remaining vertices.
\end{proof}
\end{proof}

\subsection{Putting it all together}\label{sec:putting-together}

Here, we tie up loose ends and prove \cref{thm:recovery}.
\begin{proof}[Proof of \cref{thm:recovery}]
    The proof of correctness follows together from \cref{lem:partial recovery} and from \cref{lem:boost}.

    The running time of~\cref{alg:recover} is $n^{O(\log n)}$, since it simply requires constructing $\cH$ (which cannot take more than $\text{poly}(\binom{v^2}{d'v}) = v^{O(v)}$ time), then looking for the subgraphs in $\cH$ in $G_1,G_2$ (which takes $\binom{n}{v}|\cH| = v^{O(v)} \cdot n^{O(v)} = n^{O(v)}$ time). The runtime of~\cref{algo:boosting} is polynomial time. Thus, we conclude that there is a quasi-polynomial time algorithm that solves~\cref{prob:recovery} with high probability in the specified parameter region.
\end{proof}

\section{Test Subgraphs}\label{sec:test-graphs}
Both the distinguishing and recovery algorithms rely on a carefully chosen ``test set'' of graphs with a set of delicate properties.
In this section, we construct test sets with these characteristics.

\begin{theorem}[General overview of test graph properties]\label{thm:graph family}
    For any rational scalar $d \in (2,2+\tfrac{1}{76})$ or $d \in \mathbb{Z}_{\ge 3}$ or $d \ge 6$, and sufficiently large even integer $v$, there exists a set $\cH_d^v$ of $v$-vertex graphs with the following properties.
\begin{enumerate}
    \item Every $H \in \cH_d^v$ has average degree $d$, or density $d/2$, \textit{i.e.,} $e = dv/2$.
    \item Every $H \in \cH_d^v$ is {\em strictly balanced}.
    \item Every $H\in\cH_d^v$ has no non-trivial automorphisms.
    \item For every pair of distinct graphs $H,H' \in \cH_{d}^v$ have no shared subgraphs $J \subset H,H'$ of edge density larger than $\left(\tfrac{d}{2} - \alpha(d)\right)|V(J)|$ (where $\alpha(d)> 0$ depends on the density $d$).
    \item The size of the family is $|\cH|=v^{\Omega_d(v)}$.
\end{enumerate}
\end{theorem}
In fact, \cref{thm:graph family} is the union of \cref{prop:integer-d,prop:large-d,prop:small-d}, each of which covers a different set of densities $d$ (and each of which contains a more precise quantitative statement).

When $d$ is an integer, random $d$-regular graph satisfy all but condition 4, and the graphs have sufficiently different edge sets on average that we can argue that a sufficiently large fraction of $d$-regular graphs can be chosen while still satisfying condition 5.

On the other hand, when $d$ is a non-integer, the strict balancedness property is less straightforward to satisfy.
Further, for recovery we require a strengthened, quantitative notion of strict balance.
In \cref{sec:int-avg}, we will prove that random $d$-regular graphs satisfy our requirements, and in \cref{sec:non-int-avg-large,sec:non-int-avg-small} we design constructions that leverage $d$-regular regular random graphs to obtain graphs of dense and sparse non-integer average degrees.

\subsection{Test subgraphs with integer average degree}\label{sec:int-avg}

Suppose that we were only interested in constructing test sets of graphs with integer average degree $d \ge 3$.
For those cases, random $d$-regular graphs will satisfy almost all of our requirements for the individual graphs $H$, and we will prove that there is a large subset of graphs on $d$ vertices so that every pair of graphs have no dense common subgraphs.

\begin{proposition}\label{prop:integer-d}
    For each integer $d \ge 3$ and $\alpha \in [\frac{12}{25},1]$, for $v$ sufficiently large, there exists a set of $d$-regular graphs on $v$ vertices $\cH^v_d$, with
    \[
	|\cH_d^v| = v^{\frac{1}{2}\alpha \cdot d\cdot v}
	\]
	such that every graph $H \in \cH$ is connected, has no non-trivial automorphisms, and for every distinct pair $H,H' \in \cH$, $H$ and $H'$ do {\em not} contain any common subgraphs $J$ with $|E(J)|> \alpha d\cdot |V(J)|-d$.
\end{proposition}

Random $d$-regular graphs are known to satisfy many of our requirements, either absolutely or with high probability.
For example, the following well-known fact:
\begin{fact}
For integer $d \ge 2$, any connected $d$-regular graph is strictly balanced.
\end{fact}
\begin{proof}
If $H$ is a $d$-regular graph and $J$ is a proper subgraph of $H$, then $J$ must have at least one vertex of degree less than $d$.
\end{proof}

Additionally, classical results assert that random $d$-regular graphs are with high probability connected, and also contain no non-trivial automorphisms.
\begin{theorem}[\cite{bollobas1981random, wormald1980some,luczak1989sparse}]
If $H$ is a random $d$-regular graph on $v$ vertices and  $3\leq d\leq O(v^{0.02})$, then $H$ is connected asymptotically almost surely.
\end{theorem}

\begin{theorem}[\cite{Bollobas-automorph,McKayWormwald,kim2002asymmetry}]
If $H$ is a random $d$-regular graph on $v$ vertices and $ 3 \le d \le v-4$, then almost surely $H$ has no non-trivial automorphisms.
\end{theorem}

It is not too difficult to see that there are at least $v^{\Omega(v)}$ $d$-regular graphs on $v$ vertices; however, a precise formula is also known.
The following is a consequence of a theorem of McKay and Wormwald.
\begin{fact}[\cite{MCKAY1990565}]\label{fact:gdsize}
For $3 \le d \le o(\sqrt v)$, there are at least
\[
\frac{(dv)!}{(dv/2)! 2^{dv/2} (d!)^v e^{d^2}}
\]
$d$-regular graphs on $v$ vertices.
\end{fact}

Now, we will prove that there is a set of at least $v^{\Omega(dv)}$ $d$-regular graphs that have all of the above properties, while not sharing any dense subgraphs.

\begin{proof}[Proof of \Cref{prop:integer-d}]
We construct a meta-graph $G$ whose vertices correspond to $d$-regular graphs.
    For two $d$-regular graphs $H,H'$, we add an edge between their vertices in $G$ if $H$ and $H'$ intersect on a subgraph $J$ of average degree at least $\alpha \cdot d $.

First, we argue that most $H$ do not have small subgraphs of density $\alpha$.
This follows from a result of Bollob\'{a}s, who proves that $d$-regular random graphs are good expanders.
\begin{theorem}[\cite{BOLLOBAS1988241}]\label{thm:expands}
    If $d\ge 3$ is a fixed integer and $H$ is a $d$-regular random graph on $v$ vertices, then almost surely every subset $S$ of $k\le v/2$ vertices of $H$ has at least $c_d \cdot dk$ edges to $\overline{S}$, where $c_d \ge \tfrac{3}{50}$.
    When $d \ge 6$, $c_d \ge \tfrac{13}{75}$.\footnote{In general, $c_d \to \frac{1}{2}$ as $d \to \infty$.}
\end{theorem}

For constant-sized subsets, we also use the following known lemma.
\begin{lemma}[e.g. \cite{wormald1999models}]\label{lem:smallsub}
    For any fixed integer $d \ge 3$, a $d$-regular random graph $H$ on $v$ vertices does not contain any subgraphs $J$ of fixed size with $|E(J)|>|V(J)|$, asymptotically almost surely.
\end{lemma}

    So to rule out dense shared subgraphs on fewer than $100$ vertices, we use \Cref{lem:smallsub}, to conclude that all but an $o(1)$ fraction of $d$-regular graphs cannot contain a subgraph on at most $100$ with density $>1$.

    For larger subgraphs, we apply \cref{thm:expands}, concluding that only a $o(1)$ fraction of all $d$-regular graphs contain subgraphs $J$ of size less than $v/2$ with more than $(d|V(J)| - \tfrac{3}{50} d|V(J)|)/2 = d\tfrac{47}{100}\cdot |V(J)|$ edges.
    Thus, for a ``typical'' pair of $d$-regular graphs $H,H'$ and for any $J\subset H,H'$ with $|V(J)| \in [100, v/2]$ the conclusion holds:
    \[
	|E(J)| < d\tfrac{47}{100}|V(J)| \le d\tfrac{48}{100} |V(J)| - d,
    \]
where the final inequality holds by our lower bound $|V(J)| \ge 101$.

    For convenience, fix $\alpha$ so that $\alpha \ge \tfrac{48}{100}$.
    For a ``typical'' $d$-regular graph $H$, we bound the number of $d$-regular graphs $H'$ which share a subgraph of more than $v/2$ vertices with at least $d\alpha \cdot |V(J)|$ edges.
    For each such subgraph $J \subseteq H$, $H$ intersects with at most $M_{(1-\alpha)dv}$ graphs $H'$, where $M_{2k}$ denotes the number of matchings on $2k$ elements.
    This is because once $J$ is fixed, there remain at most $\tfrac{d}{2}v - d|V(J)|\alpha \le \frac{1-\alpha}{2}dv$ edges in the rest of the graph, and the number of ways to arrange them is at most $M_{(1-\alpha)dv}$, the number of matchings on the $(1-\alpha)dv$ ``half edges'' that leave the vertices of degree less than $d$ once $J$ is removed.
    Again for convenience, let $\beta := 1 - \alpha$.
    There are at most at most $2^{dv/2}$ choices for of $J$, as that is the number of subsets of edges.
Thus, the degree of a typical $H$ (one which does not have dense subgraphs of size $< v/2$) in $G$ is at most $ D_G := 2^{dv/2} \cdot M_{\beta dv}$.

Now, we remove the $o(1)$ fraction of vertices of the meta-graph $G$ corresponding to $H$ with small dense subgraphs, nontrivial automorphisms, insufficiently expanding small subgraphs, or which are disconnected, to obtain a new meta-graph $G'$.
In the remaining graph $G'$ we find the maximum independent set.
    Since the maximum degree in $G'$ is at most $D_G$, the largest independent set has size at least $|G'|/(D_G + 1)$.
This independent set is our set $\cH$.

    Applying \Cref{fact:gdsize}, we have that
    \begin{align*}
	|\cH|
	\,\ge\, \frac{|G'|}{(D_G + 1)}
	\,\ge\, \frac{1}{4} \cdot \frac{|G|}{D_G}
	&\ge \frac{1}{4}\cdot \frac{(dv)!}{(\frac{1}{2}dv)! 2^{dv/2} (d!)^v e^{d^2}} \cdot \frac{(\frac{1}{2}\beta \cdot dv)!\cdot 2^{\beta dv/2}}{2^{dv/2} \cdot(\beta\cdot dv)!}\\
	&\ge \frac{1}{8 e^{d^2}} \cdot\left(\left(\frac{1}{e \cdot \beta^{\beta} \cdot 2\cdot d^{(1+\beta)}}\right)^{1/(1-\beta)}v\right)^{(1-\beta)dv/2},
    \end{align*}
    where in the final line we have applied Striling's approximation.
    The conclusion follows.
\end{proof}

\subsection{Test subgraphs with non-integer average degrees}

Now, we will use our integer-average-degree sets to construct a test set for graphs of non-integer average degrees.

\subsubsection{Test subgraphs with density at least 3}\label{sec:non-int-avg-large}
Here, we will show that for graphs of average degree at least $6$.
In this regime, the expansion properties of $d$-regular random graphs are so strong that they allow us to tread less delicately, and simply add a random matching on top of a $d$-regular random graph to obtain a graph with average degree between $d$ and $d+1$, without jeopardizing the strict balancedness property.

Our construction will be as follows:
\begin{proposition}\label{prop:large-d}
    Let $d \ge 6$ be an integer, and let $\lambda \in (0,1)$ be a rational number.
    Define $d' = \lambda \cdot (d+1) + (1-\lambda) \cdot d$.
    Then for any $\beta \in [0,(\lambda + \frac{1}{26})\frac{1}{d'}]$ and sufficiently large integers $v$, there exists a set of graphs on $v$ vertices with average degree $d'$, $\cH_{d'}^v$, with
    \[
	|\cH_{d'}^v| \ge v^{(\frac{1}{4}(1-\beta)\cdot d' - \lambda/2)\cdot v}
    \]
such that every $H \in \cH_{d'}^v$ has no non-trivial automorphisms and is strictly balanced, and so that every pair of distinct $H',H \in \cH_{d'}^v$ do not intersect on subgraphs $J$ with $|E(J)| \ge (1-\beta) \frac{1}{2}d'\cdot |V(J)| - \frac{d'}{2}$.
\end{proposition}
\begin{proof}
    Our construction will be as follows: we sample a $d$-regular random graph $H$, we add a random matching $M$ on $\lambda\cdot v$ random vertices to obtain the graph $H\cup M$, and we add it to the set $\cG$ if $H \cup M$ is simple and if $H \in \cH_{d}^v$.
    Then, we will find the largest subset of pairs $(H,M) \in \cH$ such that no pair of graphs $H\cup M$ and $H' \cup M'$ overlap on a subgraph with average degree larger than $d - \tfrac{1}{5}$.

    First, we prove that $H\cup M$ is simple with constant probability.
    \begin{claim}\label{claim:test dense simple}
	If $d$ is a fixed integer, $H$ is a $d$-regular random graph on $v$ vertices, and $M$ is a random matching on $\lambda v$ vertices, then with probability at least $\exp(-\frac{1}{4}(d^2 + 2d))$ the matching $M$ does not contain any edges that are included in $H$.
    \end{claim}
    \begin{proof}
	The probability that $H \cup M$ is simple is at least the probability that a random $(d+1)$-regular graph drawn from the configuration model is simple; this is known to occur with probability at least $\exp(\frac{1-(d+1)^2}{4})$ \cite{wormald1999models}.
    \end{proof}
    Each $H\cup M$ clearly has the desired average degree; we will prove that our other desired properties hold as well.

    \begin{restatable}{claim}{db}\label{claim:densebalance}
	With high probability, $H\cup M$ is strictly balanced.
	Further, for any subgraph $J \subseteq H \cup M$ of with $|V(J)| = (1-\theta)v$,
	\[
	    |E(J)| \le
	    \begin{cases}
		\frac{d'}{2}|V(J)| - \frac{1}{50}|V(J)| & \text{if }\theta > \frac{1}{2},\\
		\frac{d'}{2}|V(J)| - \frac{1}{50}\theta \cdot v & \text{if }\theta \le \frac{1}{2}.
	    \end{cases}
	    \]
    \end{restatable}
    \begin{proof}
	Applying \cref{thm:expands} with $c_d$ the expansion constant of $d$-regular random graphs, we have that in $H$, for every subgraph $S$ containing at most $v/2$ vertices, $S$ contains at most $(1-c_d)\frac{d}{2}|S|$ edges.

	Thus, if $|V(J)| = (1-\theta)v$ with $\theta > \frac{1}{2}$, because $J$ must expand in $H$, even if $J$ contains the maximum possible number of vertices in $M$ we have
	\begin{align}
	    |E(J)|
	    &\le \left(1 - c_d\right)\frac{d}{2}\cdot (1-\theta)v  + \min\left(\frac{1}{2}\lambda v, \frac{1}{2}(1-\theta)v\right)
	    \le \frac{d'}{2}\cdot (1-\theta)v  - \frac{1}{2}(d\cdot c_d - 1 + \lambda)(1-\theta)v.
	\end{align}
	Since for $d \ge 6$, $dc_d \ge 1 + \frac{1}{25}$, this quantity is strictly smaller than $\frac{d'}{2}|V(J)|$.

	If $|V(J)| = (1-\theta)v$ with $\theta \le \frac{1}{2}$, then again by the expansion of the complement of $V(J)$ in $H$, $J$ has at most $(1 -\theta -c_d\theta)\frac{d}{2}v$ edges in $H$.
	Thus, even if $J$ contains the maximum possible number of edges in $M$,
	\begin{align*}
	    |E(J)|
	    &\le \left(1-\theta - c_d\theta\right)\frac{d}{2} v + \min\left(\frac{1}{2}\lambda v,\frac{1}{2}(1-\theta)v\right)
	    \le \frac{d'}{2}(1-\theta)v + \frac{\theta}{2}\left(\lambda -d\cdot c_d\right) v.
	\end{align*}
	When $d \ge 6$, $d c_d \ge 1 + \frac{1}{25}$ and this is strictly less than $\frac{d'}{2}v$.
    \end{proof}

    \begin{claim}
	With high probability over $H \cup M \in \cG$, $H\cup M$ has no non-trivial automorphisms.
    \end{claim}
    \begin{proof}
    To establish our claim, we utilize the following result which is a corollary of a theorem of McKay and Wormwald.
    \begin{theorem}[Corollary 3.8 in \cite{McKay1984}]
 Let $d \ge 3$ be a fixed integer.
	With high probability, a random simple graph on $v$ vertices with $\lambda v$ vertices of degree $d+1$ and $(1-\lambda)v$ vertices of degree $d$ has no non-trivial automorphsism.
    \end{theorem}
    Standard arguments (see e.g. \cite{Janson}) establish that our distribution over graphs is contiguous with the uniform distribution over simple graphs on $v$ vertices with this degree sequence, therefore this conclusion holds for the set $\cG$ as well.
    \end{proof}

    \begin{claim}
	For any $d \ge 6$ and $\beta \in [0, (\lambda + \frac{1}{26})\frac{1}{d'}]$, there is a set $\cS$ of at least $v^{(\frac{1}{4}(1-\beta) d' - \lambda/2)v}$ of pairs $(H,M)$ such that no two pairs in $\cS$ share a common subgraph $J$ with $|E(J)| \ge \frac{1}{2}(1-\beta)d'\cdot|V(J)| - \frac{d'}{2}$.
    \end{claim}
    \begin{proof}
	We again construct a meta-graph $G$, with one vertex for every pair of $H \in \cH_d^v$ and matching $M$ on $(1-\lambda)v$ vertices.
	We place an edge between $(H,M)$ and $(H',M')$ in the meta-graph if there is a shared subgraph $J$ of $H\cup M, H'\cup M'$ such that $|E(J)| \ge \frac{1}{2}(1 - \beta)d'|V(J)| - \frac{d'}{2}$.
	We'll take $\cS$ to be the largest independent set in $G$; to obtain a lower bound on its size, we will bound the degree of $(H,M)$ in $G$.

	First, for subgraphs $J \subseteq H\cup M, H' \cup M'$ of size at most $|V(J)| \le 25\cdot 26\cdot \frac{d'}{2}$, we can apply \Cref{lem:smallsub} to conclude that with high probability over $(H,M)$ and $(H',M')$, every $J$ contains at most $|V(J)|$ edges from $H$, and we also have that it contains at most $|V(J)|/2$ edges from $M$.
	Thus, as long as $|E(J)|\le \frac{3}{2}|V(J)| \le \frac{1}{2}(1-\beta)d'|V(J)| - \frac{d'}{2}$ when $d\ge 6$ and $\beta \le \frac{1}{6}$ (this is using the fact that for $|V(J)| \le 2$, there cannot be more than one edge).

	Now we consider larger subgraphs.
	By construction, $H$ and $H'$ do not overlap on subgraphs $J$ with $|E(J)| \ge \alpha d|V(J)| - d$.
	Also, we can use \cref{thm:expands} to conclude that with high probability for any subgraph $J$ of $H\cup M$ and $H'\cup M'$ of size $< v/2$ has
	\[
	    |E(J)|
	    \le
	    \frac{1}{2}\left((1-c_d)d|V(J)| + |V(J)|\right) < \frac{1}{2}(d - \tfrac{1}{25})|V(J)| < \frac{1}{2}(d - \tfrac{1}{26}) |V(J)| - \frac{d'}{2},
	    \]
	    where the second inequality uses that $d \ge 6$ and the third inequality uses that $|V(J)| > 25 \cdot 26 \cdot \frac{d'}{2}$.

	Suppose therefore that there exists some subgraph $K \subseteq H \cup M,H' \cup M'$ with
	$|E(K)| \ge \frac{1}{2}(1-\beta)\cdot d'\cdot |V(K)|-\frac{d'}{2}$, for $\beta < \frac{\lambda + \frac{1}{26}}{d + \lambda}$;
	it must follow that $|V(K)| \ge v/2$.
	Furthermore, since the restriction of $K$ to $H$ and $H'$ must have fewer than $\alpha d|V(K)|- d$ edges, we can conclude that the matchings $M$ and $M'$ must intersect on at least $\left(\frac{1}{2}(1-\beta)d' - \alpha d\right)\cdot|V(K)| \ge \tfrac{1}{4}(1-\beta)d' v - \tfrac{1}{2}\alpha dv$ edges since $|V(K)|\ge\frac{v}{2}$.

	We now bound $D_{H,M}$, the degree of $(H,M)$ in the meta-graph $G$, or the number of pairs $(H',M')$ for which $(H,M)$ can share at least $\ell =  \left(\tfrac{1}{4}(1-\beta)d' - \tfrac{1}{2}\alpha d\right)\cdot v$ edges between $M'$ and $M$.
	If $\ell > \lambda v/2$, this is already a contradiction, and the graph $G$ has no edges.

	Otherwise, $D_{H,M}$ is at most the number of choices of $\ell$ edges in $M$, times the number choices for the $m = \lambda v - 2\ell$ vertices outside of these $\ell$ edges that could be in the matching in $M'$, times the number of matchings remaining on $m$ elements (corresponding to the leftover vertices).
	This is at most
	\begin{align*}
	    D_{H,M}
	    &\le \binom{\tfrac{1}{2}\lambda v}{\ell} \cdot \binom{v}{m} \cdot M_m.
	\end{align*}
	On the other hand, with high probability over our choice of matchings $M$, $G$ has at least $\exp(-\frac{1}{4}(d^2 + 2d))|\cH_d^v|/2$ vertices from~\cref{claim:test dense simple}.
	The size of the maximum independent set in $G$ is at least
	\begin{align*}
	    |\cS|
	    \,\ge\, \frac{|G|}{D_{H,M} + 1}
	    &\ge \frac{1}{2}\cdot\frac{\frac{1}{2}\cdot\exp(-\frac{1}{4}(d^2 + 2d))\cdot |\cH_d^v|}{\binom{\tfrac{1}{2}\lambda v}{\ell}\binom{v}{m}\cdot M_m}\\
	    &\ge \frac{1}{4} \cdot \exp(-\frac{1}{4}(d^2 + 2d))\cdot |\cH_d^v|\cdot \Theta(v) \cdot \frac{1}{2^{\frac{1}{2}\lambda v}\cdot 2^{v}\cdot \left(\frac{m}{e}\right)^{m/2}}\\
	    &= |\cH_d^v| \cdot \Omega(v)^{-\frac{\lambda v}{2} + \frac{1}{4}(1-\beta)d'v - \frac{1}{2}\alpha d v - \frac{4}{\log v}v},
	\end{align*}
	where we have bounded the binomial coefficients by $2^v$ and applied Stirling's inequality to estimate the number of matchings.
	Finally, because $|\cH_d^v| \ge v^{\frac{1}{2}\alpha\cdot dv}$ by \cref{prop:integer-d}, we have our conclusion.
    \end{proof}

    Finally, taking $\cH_{d'}^v$ to be the largest independent set $\cS$ given by the previous claim, then removing the $o(1)$ fraction of graphs with non-trivial automorphisms, concludes the proof.
\end{proof}

\subsubsection{Test subgraphs with density close to 1}\label{sec:non-int-avg-small}
In this subsection, we describe a graph family $\cc{H}$ of graphs of density close to 1 that satisfies the desired properties.
In this regime, our strategy is to sample a $3$-regular random graph and subdivide its edges into paths.
If we want to achieve an arbitrary average degree, the paths that we subdivide cannot all have the same length.
However, when we subdivide into paths that are sufficiently long, the expansion of the $3$-regular graph will compensate for the asymmetry.\footnote{Of course, one might hope that choosing the lengths of the edges randomly will also compensate for the asymmetry, regardless of the path lengths; we conjecture that this is the case but have not proven it as of yet.}
Using this intuition, we are able to prove that our construction satisfies the desired properties when our target density is close to 1 (few degree-3 vertices relative to the number of degree-2 vertices).

In this subsection, we design test sets with density close to 1.
\begin{algorithm}[Generating $\cH_{d'}^v$ for $d' \in [2,2+\eps)$]\label{alg:small-p}.\\
    {\bf Input:} An integer $v$ and a desired average degree $d' \in [2,3)$, so that $d' = 2(1-\lambda) + \lambda 3$ for a rational number $\lambda \in (0,1)$.
    Let $\left\lfloor \frac{(1-\lambda)v}{3\lambda v/2}\right\rfloor = k$, so that $\frac{(1-\lambda)v}{3\lambda v/2} = \alpha \cdot k + (1-\alpha) \cdot (k+1)$ for $\alpha \in (0,1)$; this is the average number of desired degree-2 vertices per edge of $H$.

    We require that $\lambda v$ is an integer , that $3 \lambda v $ is an even integer,
    and that $\alpha 3\lambda v/2$ is an integer.

    For each graph $H \in \cH_{3}^{v}$:
    \begin{enumerate}
	\item Choose $\alpha\cdot 3\lambda v/2$ edges of $H$ uniformly at random, and replace them with paths of length $k$ (subdivide the edge with $k$ vertices).
	\item For the remaining $(1-\alpha)\cdot 3\lambda v/2$ edges, replace them with $k+1$ length paths.
    \end{enumerate}
    Add the resulting graph $H'$ to the set $\cH_{d'}^v$.
\end{algorithm}
\begin{proposition}\label{prop:small-d}
    Let $d' = (1-\lambda)\cdot 2 + \lambda \cdot 3 = 2 + \lambda$ for $\lambda \in (0,\frac{1}{76}]$, and let $v$ be a sufficiently large integer with $\lambda 3 v$ an even integer.
    Then the set $\cH_{d'}^v$ produced by \cref{alg:small-p} is a set of graphs on $v$ vertices with average degree $d'$, with
    \[
	|\cH_{d'}^v| \ge (\lambda v)^{c\cdot \lambda v},
    \]
for a constant $c > \tfrac{69}{100}$, and furthermore every $H \in \cH_{d'}^v$ has no non-trivial automorphisms, is strictly balanced, and every pair of distinct $H,H' \in \cH_{d'}^v$ do not intersect on any subgraphs $J$ with $|E(J)| \ge (1 + \frac{1}{3}\lambda)|V(J)|-2$.
\end{proposition}
\begin{proof}
    Let $H \in \cH'$, and let $H_3$ be the corresponding $3$-regular graph (given by shrinking every $k$-- or $(k+1)$-path to a single edge).
    $H'$ has average degree $d'$ by construction.

    \begin{claim}
    $H$ has no non-trivial automorphisms.
    \end{claim}
    \begin{proof}
	This follows from the fact that $H_3$ has no non-trivial automorphisms; an automorphism of $H$ corresponds to an automorphism of $H_3$, as in an automorphism of $H$, every subdivided edge is mapped to a subdivided edge.
    \end{proof}

    \begin{restatable}{lemma}{qbalsparse}(Quantitative strict balance).\label{lem:qbalsparse}
    Let $d' = 2 + \lambda$, with $\lambda \in (0,\frac{1}{76}]$.
    Suppose that $H \in \cH_{d'}^v$ is a graph on $v$ vertices with average degree $d'$, produced by \cref{alg:small-p} by subdividing the edges of the $3$-regular graph $H_3$ into paths.
    Any subgraph $J \subset H$ with $(1 - \theta)v$ vertices has density at most $1 + \frac{1}{2}\lambda - \frac{1}{100}\theta \lambda$.
	In particular, $H$ is strictly balanced.
    \end{restatable}
    The proof of \cref{lem:qbalsparse} is somewhat involved, and we prove it below after completing the proof of the proposition.

	Finally, it remains to prove that many pairs $H,H' \in \cH$ do not overlap on dense subgraphs.
	\begin{claim}
	    If $\lambda < \frac{11}{109}$, then every distinct pair $H,H' \in \cH_{d'}^v$ overlap on no subgraphs $J$ with $|E(J)| \ge (1 + \lambda/4)|V(J)| - 2$.
	\end{claim}
	\begin{proof}
	    First, we use \cref{lem:smallsub} (and our construction of $\cH_3$) to conclude that for any shared subgraph $J$ of $H,H'$ with at most $100$ vertices in $H_3$ or $H'_3$, the density is bounded by $1$.

	    Since $H_3,H'_3$ come from $\cH_3^{\lambda v}$, we already have that $H_3,H'_3$ do not overlap on any subgraph $J_3$ with  $|E(J_3)| \ge \frac{12}{25}\cdot 3|V(J_3)| - 3$.
	    If $H_3,H'_3$ overlap on a subgraph $J_3$ of size at least $101$ with average degree $\delta$, then in $H,H'$ the corresponding subgraph $J$ will have density at most
	    \[
		\frac{|E(J)|}{|V(J)|}
		\le \frac{\tfrac{1}{2}\delta |V(J_3)|\cdot (k+1)}{|V(J_3)| + \tfrac{1}{2}\delta|V(J_3)|\cdot k}
		\le 1 + \frac{\tfrac{1}{2}\delta -1}{1 + \tfrac{1}{2}\delta\cdot k},
	    \]
	    and again if we use that $k > \tfrac{3(1-\lambda)}{2\lambda}-1$,
	    and that $\delta < 3-\tfrac{1}{50}$ we get that
	    \[
		\frac{|E(J)|}{|V(J)|}
		\le 1 + \frac{\tfrac{1}{2}\delta -1}{1 + \tfrac{1}{2}\delta\cdot (\tfrac{3(1-\lambda)}{2\lambda} -1)}
		< 1 + \frac{\lambda}{4},
	    \]
	    whenever $\lambda < \frac{11}{109}$.
	    Since we have assumed that $|V(J_3)|\ge 100$ and that the density of $J_3$ is larger than $1$, $|V(J)| \ge 100\cdot k \ge 100(\tfrac{3(1-\lambda)}{\lambda} - 1) > \frac{6}{\lambda}$, and we have that
	    \[
		|E(J)| < (1+\frac{\lambda}{3})|V(J)| -2,
	    \]
	    as desired.
	\end{proof}
	Finally, we note that the size of $\cH_{d'}^v$ is simply the size of $\cH_{3}^{\lambda v}$, from which we get our bound.
\end{proof}

We now prove that our graphs are strictly balanced; in fact we will need a quantitative version of this statement, which first requires the following supporting lemma.
\begin{lemma}
    \label{lemma: strict balance on core}
    Let $d' = 2 + \lambda$, with $0 < \lambda \le \frac{1}{9}$.
    Suppose that $H \in \cH_{d'}^v$ is a graph on $v$ vertices with average degree $d'$, produced by \cref{alg:small-p} by subdividing the $3$-regular graph $H_3$ on $\lambda v$ vertices into $k$- and $(k+1)$-paths.
    Then any subgraph $J \subset H$ containing at most  $(1-\theta)\lambda v$ vertices of $H_3$ (vertices of degree $3$ in $H$) has density at most $1 + \frac{1}{2}\lambda - \frac{1}{30}\theta\lambda + \frac{152}{150}\theta \lambda^2$.
\end{lemma}
\begin{proof}
    First, the following claim allows us to consider only $J$ corresponding to full edge-induced $J_3 \subset H$.
    \begin{claim}\label{claim:black magic}
	Suppose that $J$ is a subgraph of $H$ with density at least $1$, and that $J$ contains only a proper subset of the edges an vertices which correspond to the subdivided path given by the edge $e$ in $H_3$.
	Then the subgraph $\tilde J \subset J$ given by removing all vertices and edges from the path corresponding to $e$ has larger density than $J$.
    \end{claim}
    \begin{proof}
	The edge $e$ is subdivided into a path.
	If $J$ contains a strict sub-path, then it must contain a vertex of degree $1$.
	In a graph of density at least $1$, density can only increase when vertices of degree-$1$ are removed.
	Therefore, if we remove the degree-1 vertices one-by-one until we obtain $\tilde{J}$, we obtain a graph that can only be denser.
    \end{proof}

    By \cref{claim:black magic}, $J$ is at its densest if every subdivided edge from $H$ is either present or completely missing, thus we may restrict our attention to subgraphs $J$ which correspond exactly to edge-induced subgraphs $J_3$ of $H_3$, with all subdivided edges fully present.

    Suppose that $J_3$ has average degree $\delta$.
    To obtain $J$, every edge of $J_3$ is subdivided into a $k$- or $(k+1)$-path; let the average path length among subdivided edges present in $J$ be $\overline{k}$.
    In $H$, by design the average path length is $\frac{2(1 - \lambda)}{3\lambda}$.
    From this, we can lower bound $\overline{k}$ as follows:
    \begin{claim}\label{claim:kbar}
	$$ \overline{k} \ge \frac{2(1 - \lambda)}{3\lambda} - \theta -  \left(1-\frac{\delta}{3}\right)(1-\theta).
	$$
    \end{claim}
    \begin{proof}
	By assumption, $|V(J_3)| = (1 - \theta)\lambda v$.
	Suppose that $|E(J_3)| = (\frac{3}{2}(1-\theta)\lambda v - \ell)$, where $\ell$ accounts for edges missing from $J_3$ entirely and also for edges with only one endpoint in $V(J_3)$.
	The total number of edges in subdivided paths in $J$ is then $(\frac{3}{2}(1 - \theta)\lambda v - \ell)\overline{k}$.
	Therefore the total number of edges along subdivided paths outside $J$ is at most $(\frac{3}{2}\theta \lambda v + \ell)(\overline{k} + 1)$. 
	The total number of vertices participating in subdivided paths in $H$ is $(1 - \lambda) v$, and from this we get
	\[
	\left(\frac{3}{2}(1 - \theta)\lambda v - \ell\right)\cdot\overline{k} + \left(\frac{3}{2}\theta \lambda v + \ell\right)\cdot (\overline{k} + 1)
	~\ge (1-\lambda)v,
	\]
	and solving for $\overline{k}$ we get
	\[
	    \overline{k}
	    ~\ge~ \frac{2(1 - \lambda)}{3\lambda} - \theta - \frac{2\ell}{3\lambda v}.
	    \]
	    By definition of $\ell$ and $\delta$, we have $\ell = \frac{3-\delta}{2}(1-\theta)\lambda v$.
	    Plugging this in, we have our bound.
    \end{proof}
    Then applying our lower bound on $\overline{k}$ from \cref{claim:kbar}, we have that
    \begin{equation}\label{eq:subdens}
	\frac{|E(J)|}{|V(J)|}
	\le \frac{\frac{1}{2}\delta|V(J_3)|(\overline{k}+1)}{|V(J_3)| + \frac{1}{2}\delta|V(J_3)|\overline{k}}
	= 1 + \frac{\frac{1}{2}\delta - 1}{1 + \frac{1}{2}\delta \overline{k}}
	\leq 1 + \frac{\frac{1}{2}\delta - 1}{1 + \frac{1}{2}\delta \left(\frac{2(1-\lambda)}{3\lambda} - \theta - (1-\theta)(1-\frac{\delta}{3}) \right)},
    \end{equation}
    where in the final inequality we apply our lower bound on $\overline{k}$.
    Taking the derivative with respect to $\delta$, one may verify that so long as $\lambda < 1/7$ and $\delta \le 3$, this quantity grows with $\delta$.

    We obtain an upper bound on $\delta$ using the expansion of $H_3$, from \cref{thm:expands}.
    If $\theta < \frac{1}{2}$, we have that $J_3$ can contain at most $\frac{1}{2}(3 - \theta \cdot \frac{3}{25})\lambda v$ edges, since the set of size $\theta \lambda v$ not included in $J_3$ has expansion at least $\tfrac{3}{50}\theta \lambda v$.
    If $\theta > \frac{1}{2}$, we have that $J_3$ can contain at most $\frac{1}{2}(3 - \frac{3}{25})(1-\theta)\lambda v$ edges, for the same reason.
    Since $\frac{1}{2}(3-\theta\frac{3}{25}) \ge \frac{1}{2}(3-\frac{3}{25})(1-\theta)$ for non-negative $\theta$, we use the former, looser bound.

    Returning to \cref{eq:subdens} we have that
    \begin{align*}
	\frac{|E(J)|}{|V(J)|} & \leq 1 + \frac{\frac{1}{2}(3-\theta \frac{3}{25})-1}{1 + \frac{1}{2}(3-\theta \frac{3}{25})\left(\frac{2(1-\lambda)}{3\lambda} - \theta - (1-\theta)\frac{1}{25}\theta\right)} \\
	& = 1 + \frac{1}{2}\lambda\cdot\left( \frac{1 - \frac{3}{25}\theta}{1 - \frac{1}{25}\theta - \frac{38}{25}\lambda \theta + \frac{153}{1250}\lambda \theta^2 - \frac{3}{1250} \lambda \theta^3}\right) \\
	& \leq 1 + \frac{1}{2}\lambda\cdot \left(1 - \frac{3}{25}\theta\right)\cdot\left(1 + \frac{4}{3}\left(\frac{1}{25}\theta + \frac{38}{25}\lambda \theta - \frac{153}{1250}\lambda \theta^2 + \frac{3}{1250} \lambda \theta^3\right)\right) \\
	& = 1 + \frac{1}{2}\lambda\cdot \left(1 - \frac{1}{15}\theta + \frac{152}{75}\lambda \theta - \frac{4}{625} \theta^2 - \frac{254}{625}\lambda \theta^2 + \frac{356}{15625}\lambda \theta^3 - \frac{6}{15625}\lambda \theta^4\right)\\
	& \le 1 + \frac{1}{2}\lambda\cdot \left(1 - \frac{1}{15}\theta + \frac{152}{75}\lambda \theta \right).
    \end{align*}
    Where to obtain the third line we have used that for $\theta \in [0,1]$ and $\lambda \le \frac{1}{9}$ the denominator within the parentheses is at least $\tfrac{3}{4}$, and in the last step we have used that $\theta \in [0,1]$ and $\lambda \le \frac{1}{9}$.
    The conclusion follows.
\end{proof}

Finally, we prove \cref{lem:qbalsparse}.
\qbalsparse*
\begin{proof}
    Here, we consider arbitrary edge-induced subgraphs $J\subseteq H$.
    We let $J_3 \subseteq H_3$ be the edge-induced subgraph of $H_3$ which contains any edge for which the corresponding path in $J$ is at least partially present.
    That is, we allow the subdivided path corresponding to each edge $\ell \in E(J_3)$ to be only partially present in $J$.
    For an edge $\ell \in E(J_3)$, we'll let and $V_J(\ell), E_J(\ell)$ be the vertex set and edge set of the corresponding path in $J$, and define $V_H(\ell), E_H(\ell)$ correspondingly for $H$; also define $k(\ell)$ to be the path length of the path corresponding to $\ell$ in $H$.

    For a subgraph $J \subset H$ where $V(J) = (1 - \theta)v$, we consider two cases: $|V(J_3)| \leq (1 - \frac{1}{2}\theta)\lambda v$ and $|V(J_3)| \geq (1 - \frac{1}{2}\theta)\lambda v$.
    For the first case, by \cref{lemma: strict balance on core} we know the density is at most $1 + \frac{1}{2}\lambda - \frac{1}{60}\theta \lambda + \frac{152}{300}\theta \lambda^2 \le 1 + \frac{1}{2}\lambda - \frac{1}{100}\theta\lambda$ when $\lambda \le \frac{1}{76}$.

    For the second case, the number of ``path vertices'' (vertices that result from subdividing edges of $H_3$) in $J$ is at most
    \begin{equation}
	(\# \text{ path vtcs } \in J) \le
	(1 - \theta) v - (1 - \frac{1}{2}\theta)\lambda v = (1 - \lambda)v - \theta v + \frac{1}{2}\theta \lambda v\label{eq:pbd}
    \end{equation}
Note that for each edge $\ell$ in $J_3$ we have $|E(\ell)| \leq (1 + \frac{1}{k(\ell)}) \cdot |V(\ell)|$.
    We can bound the average $\frac{1}{k(\ell)}$ in $J$ using the following claim:
    \begin{claim}\label{claim:sub-leng}
	Define $\hat{k}$ to be the average length of a subdivided path in $J$, as weighted by the number of vertices present in that path:
	\[
	    \frac{1}{\hat k} := \frac{1}{\sum_{\ell \in E(J_3)} |V_J(\ell)|}\sum_{\ell \in E(J_3)} \frac{1}{k(\ell)} |V_J(\ell)|.
	    \]
	If $J$ contains at least a $(1-\alpha)$-fraction of the path vertices of $H$, then
	\[
	    \frac{3\lambda}{2(1 - \lambda)}(1 - \frac{27}{16}\lambda \alpha) \leq \frac{1}{\hat{k}} \leq \frac{3\lambda}{2(1 - \lambda)}(1 + \frac{27}{16}\lambda \alpha).
	    \]
    \end{claim}
    \begin{proof}
	Recall that $V_H(\ell)$ is the number of path vertices on path $\ell$ in $H$. Consider the sum $\sum_{\ell \in E(H_3)} (1 + \frac{1}{k(\ell)}) |V_H(\ell)|$; we have that
	\[
	    \sum_{\ell \in E(H_3)} \left(1 + \frac{1}{k(\ell)}\right)|V_H(\ell)| = |E(H)| = \left(1 + \frac{\lambda}{2}\right)v,\quad \text{ and }\quad
	    \sum_{\ell \in E(H_3)} |V_H(\ell)| = (1 - \lambda)v.
	\]
	Therefore
	\[
	    \frac{\sum_{\ell \in E(H_3)} \frac{1}{k(\ell)}|V_H(\ell)|}{\sum_{\ell \in E(H_3)} |V_H(\ell)|} = \frac{3\lambda}{2(1 - \lambda)}.
	\]
	We can split the sum over $\ell \in E(H_3)$ into to the sum inside and outside of $J$.
	Because $\frac{1}{\hat k}$ is a convex combination of $\frac{1}{k+1}$ and $\frac{1}{k}$, the average of $\frac{1}{k(\ell)}$ outside of $J$ can be
	lower bounded by $\frac{1}{k + 1}$.
	Thus we have
	$$ \alpha \frac{1}{k + 1} + (1 - \alpha)\frac{1}{\hat{k}} \leq  \frac{\sum_{\ell \in E(H_3)} \frac{1}{k(\ell)}|V_J(\ell)| + \sum_{\ell \in E(H_3)} \frac{1}{k(\ell)}|V_{H\backslash J}(\ell)| }{\sum_{\ell \in E(H_3)} |V_H(\ell)|} = \frac{3\lambda}{2(1 - \lambda)},
	$$
        Which gives:
        $$ 
	\frac{1}{\hat{k}} \leq\frac{3\lambda}{2(1 - \lambda)}(1 + \frac{\alpha}{k + 1}).
	$$
	Now, since by construction $k = \lfloor \frac{2(1-\lambda)}{3\lambda} \rfloor$ and $\lambda < \frac{1}{9}$, using the upper bound $\frac{1}{k}\le \frac{27}{16}\lambda$ gives us the claim.
    \end{proof}

    Now in the extreme case when all of the vertices of $H_3$ are included in $J$, $J$ has at least $(1 - \theta - \lambda)v$ path vertices.
    Thus, the proportion of path vertices of $H$ included in $J$, is at least
    \[
	\frac{\# \text{ path vtcs } \in V(J)}{\# \text{ path vtcs } \in V(H)} \le \frac{1 - \lambda - \theta}{1-\lambda}  = 1 - \frac{\theta}{1-\lambda}	\ge 1 - \frac{9}{8}\theta.
	\]
	assuming $\lambda < \frac{1}{9}$.
	Thus we can apply \cref{claim:sub-leng}, with $\alpha := \frac{9}{8}\theta$, and our bound on the number of path vertices in $J$ \cref{eq:pbd}, and we have
    \begin{align*}
	|E(J)|
	= \sum_{\ell \in E(J_3)} |E_J(\ell)|
	 \le \sum_{\ell \in E(J_3)} \left(1 + \frac{1}{k(\ell)}\right) |V_J(\ell)|
	&= \left(1 + \frac{1}{\hat k}\right)\sum_{\ell \in E(J_3)} |V_J(\ell)|\\
	&\le \left(1 + \frac{3\lambda}{2(1-\lambda)}(1+\frac{243}{128}\theta\lambda)\right)\cdot\left(1 - \lambda - \theta + \frac{1}{2} \theta\lambda\right)v\\
	&\le \left(1 - \lambda - \theta + \frac{1}{2}\theta \lambda + \frac{3}{2}\lambda(1+2\lambda \theta) - \theta(1-\frac{1}{2}\lambda)\frac{3\lambda}{(1-\lambda)}\right)v\\
	&\le \left(1 + \frac{1}{2}\lambda - \theta + \frac{1}{2}\lambda \theta + 3\lambda^2 \theta - 3\theta\lambda (1-\frac{1}{2}\lambda)(1+\lambda)\right)v\\
	&\le \left(1 + \frac{1}{2}\lambda - \theta - \frac{5}{2}\lambda \theta + 3\lambda^2 \theta\right)v,
    \end{align*}
    Where to obtain the third line we have used that $\frac{(1+\frac{243}{128}\theta \lambda}{(1-\lambda)} < 2$ for $\lambda < \frac{1}{9}$, and to obtain the final line we have used that $(1+\lambda)(1-\frac{1}{2}\lambda) \ge 1$ for $\lambda \le 1$.
    Therefore the density is upper bounded by
    \begin{align*}
	\frac{E(J)}{V(J)} & \leq \frac{\left(1 + \frac{1}{2}\lambda - \theta - \frac{5}{2}\theta \lambda + 3\lambda^2\theta\right)}{(1 - \theta)v}
	= \frac{(1-\theta) + \frac{1}{2}\lambda(1-\theta) - 2\lambda \theta + 3\lambda^2\theta}{1-\theta}
	= 1 + \frac{1}{2}\lambda - \lambda\theta\frac{2- 3\lambda}{(1-\theta)},
    \end{align*}
    and this quantity is at most $1 + \frac{1}{2}\lambda - \lambda \theta$ when $\lambda \le \frac{1}{3}$.
    Combine the above two cases we get the lemma.
\end{proof}

\section*{Discussion and open problems}

We have shown the first (nearly) efficient algorithms for the graph matching problem on correlated \ER graphs in a wide range of parameters.
However, our results can still be improved in several ways.
First of all, for the actual recovery task, we obtain only \textit{quasipolynomial} as opposed to polynomial time algorithms.
Second, even for distinguishing, our algorithms still don't work in all regimes of parameters, and there are still some annoying ``gaps'' in our coverage.
Resolving both these problems would be extremely interesting.

One approach to give a more general and robust method can be to embed our algorithms in a convex program such as the \textit{sum-of-squares semidefinite program}.
One advantage of this program is that while its \textit{analysis} might need to use the existence of a test set, the actual algorithm will be independent of it, and would be identical in all ranges of parameters, and independent of the correlated \ER model.
This suggests that it could generalize to more settings, and perhaps also help bridge the gap between distinguishing and recovery.

Another question is whether our algorithms can inform in any way practical heuristics for the graph matching problem.
One of the most common heuristics is the propagation graph matching (PGM) algorithm that  gradually grows knowledge of the permutation from a ``seed set'' by looking at vertices that have at least $r$ neighbors into this set.
From our perspective, this can be thought of as an algorithm that looks for a particular subgraph $H$ which is the ``star graph'' with one internal vertex and $r$ leaves.
Our results suggest that it might be possible to get some mileage from looking at more complicated patterns and more than one pattern.

\section*{Acknowledgements}

Boaz Barak is grateful to Negar Kiyavash for telling him about the graph matching problem, and the results of~\cite{cullina2016improved,cullina2017exact} during a very enjoyable and useful visit to EPFL in the summer of 2017.
Tselil Schramm thanks Sam Hopkins for valuable and vigorous discussions which helped guide the exposition.
We all thank anonymous STOC 2019 referee for pointing out some issues in the prior version of this paper.

\bibliography{mybib}

\newcommand{\etalchar}[1]{$^{#1}$}
\providecommand{\bysame}{\leavevmode\hbox to3em{\hrulefill}\thinspace}
\providecommand{\MR}{\relax\ifhmode\unskip\space\fi MR }
\providecommand{\MRhref}[2]{%
  \href{http://www.ams.org/mathscinet-getitem?mr=#1}{#2}
}
\providecommand{\href}[2]{#2}
\begin{thebibliography}{OWWZ14}

\bibitem[BBM05]{berg2005shape}
Alexander~C Berg, Tamara~L Berg, and Jitendra Malik, \emph{Shape matching and
  object recognition using low distortion correspondences}, Computer Vision and
  Pattern Recognition, 2005. CVPR 2005. IEEE Computer Society Conference on,
  vol.~1, IEEE, 2005, pp.~26--33.

\bibitem[B{\c{C}}PP99]{Burkard1999}
Rainer~E. Burkard, Eranda {\c{C}}ela, Panos~M. Pardalos, and Leonidas~S.
  Pitsoulis, \emph{The quadratic assignment problem}, pp.~1713--1809, Springer
  US, Boston, MA, 1999.

\bibitem[Bol]{Bollobas-automorph}
B\'{e}la Bollob\'{a}s, \emph{The asymptotic number of unlabelled regular
  graphs}, Journal of the London Mathematical Society \textbf{s2-26}, no.~2,
  201--206.

\bibitem[Bol81]{bollobas1981random}
B{\'e}la Bollob{\'a}s, \emph{Random graphs}, Cambridge University Press, 1981,
  pp.~80--102.

\bibitem[Bol88]{BOLLOBAS1988241}
B\'{e}la Bollob\'{a}s, \emph{The isoperimetric number of random regular
  graphs}, European Journal of Combinatorics \textbf{9} (1988), no.~3, 241 --
  244.

\bibitem[CFSV04]{conte2004thirty}
Donatello Conte, Pasquale Foggia, Carlo Sansone, and Mario Vento, \emph{Thirty
  years of graph matching in pattern recognition}, International journal of
  pattern recognition and artificial intelligence \textbf{18} (2004), no.~03,
  265--298.

\bibitem[CK16]{cullina2016improved}
Daniel Cullina and Negar Kiyavash, \emph{Improved achievability and converse
  bounds for erdos-renyi graph matching}, ACM SIGMETRICS Performance Evaluation
  Review, vol.~44, ACM, 2016, pp.~63--72.

\bibitem[CK17]{cullina2017exact}
\bysame, \emph{Exact alignment recovery for correlated erdos renyi graphs},
  arXiv preprint arXiv:1711.06783 (2017).

\bibitem[CL12]{cho2012progressive}
Minsu Cho and Kyoung~Mu Lee, \emph{Progressive graph matching: Making a move of
  graphs via probabilistic voting}, Computer Vision and Pattern Recognition
  (CVPR), 2012 IEEE Conference on, IEEE, 2012, pp.~398--405.

\bibitem[CSS07]{cour2007balanced}
Timothee Cour, Praveen Srinivasan, and Jianbo Shi, \emph{Balanced graph
  matching}, Advances in Neural Information Processing Systems, 2007,
  pp.~313--320.

\bibitem[DMCW84]{McKayWormwald}
B~D.~McKay and N~C.~Wormald, \emph{Automorphisms of random graphs with
  specified degrees}, 325--338.

\bibitem[HKP{\etalchar{+}}17]{hopkins2017power}
Samuel~B Hopkins, Pravesh~K Kothari, Aaron Potechin, Prasad Raghavendra, Tselil
  Schramm, and David Steurer, \emph{The power of sum-of-squares for detecting
  hidden structures}, FOCS (2017).

\bibitem[JLG{\etalchar{+}}15]{ji2015your}
Shouling Ji, Weiqing Li, Neil~Zhenqiang Gong, Prateek Mittal, and Raheem~A
  Beyah, \emph{On your social network de-anonymizablity: Quantification and
  large scale evaluation with seed knowledge.}, NDSS, 2015.

\bibitem[JLR11]{Janson}
S.~Janson, T.~Luczak, and A.~Rucinski, \emph{Random graphs}, Wiley Series in
  Discrete Mathematics and Optimization, Wiley, 2011.

\bibitem[J{\L}T{\etalchar{+}}12]{janson2012bootstrap}
Svante Janson, Tomasz {\L}uczak, Tatyana Turova, Thomas Vallier, et~al.,
  \emph{Bootstrap percolation on the random graph $ g\_ $\{$n, p$\}$ $}, The
  Annals of Applied Probability \textbf{22} (2012), no.~5, 1989--2047.

\bibitem[KHG15]{kazemi2015growing}
Ehsan Kazemi, S~Hamed Hassani, and Matthias Grossglauser, \emph{Growing a graph
  matching from a handful of seeds}, Proceedings of the VLDB Endowment
  \textbf{8} (2015), no.~10, 1010--1021.

\bibitem[KL14]{korula2014efficient}
Nitish Korula and Silvio Lattanzi, \emph{An efficient reconciliation algorithm
  for social networks}, Proceedings of the VLDB Endowment \textbf{7} (2014),
  no.~5, 377--388.

\bibitem[KSV02]{kim2002asymmetry}
Jeong~Han Kim, Benny Sudakov, and Van~H Vu, \emph{On the asymmetry of random
  regular graphs and random graphs}, Random Structures \& Algorithms
  \textbf{21} (2002), no.~3-4, 216--224.

\bibitem[LAV{\etalchar{+}}14]{lyzinski2014other}
Vince Lyzinski, Sancar Adali, Joshua~T Vogelstein, Youngser Park, and Carey~E
  Priebe, \emph{Seeded graph matching via joint optimization of fidelity and
  commensurability}, arXiv preprint arXiv:1401.3813 (2014).

\bibitem[Law63]{lawler1963quadratic}
Eugene~L Lawler, \emph{The quadratic assignment problem}, Management science
  \textbf{9} (1963), no.~4, 586--599.

\bibitem[LFF{\etalchar{+}}16]{Lyzinskiconvexrisk}
V.~Lyzinski, D.~E. Fishkind, M.~Fiori, J.~T. Vogelstein, C.~E. Priebe, and
  G.~Sapiro, \emph{Graph matching: Relax at your own risk}, IEEE Transactions
  on Pattern Analysis and Machine Intelligence \textbf{38} (2016), no.~1,
  60--73.

\bibitem[LFP14]{lyzinski2014seeded}
Vince Lyzinski, Donniell~E Fishkind, and Carey~E Priebe, \emph{Seeded graph
  matching for correlated erd{\"o}s-r{\'e}nyi graphs.}, Journal of Machine
  Learning Research \textbf{15} (2014), no.~1, 3513--3540.

\bibitem[LR13]{livi2013graph}
Lorenzo Livi and Antonello Rizzi, \emph{The graph matching problem}, Pattern
  Analysis and Applications \textbf{16} (2013), no.~3, 253--283.

\bibitem[Luc89]{luczak1989sparse}
Tomasz Luczak, \emph{Sparse random graphs with a given degree sequence},
  Proceedings of the Symposium on Random Graphs, Poznan, 1989, pp.~165--182.

\bibitem[MW84]{McKay1984}
B.~D. McKay and N.~C. Wormald, \emph{Automorphisms of random graphs with
  specified vertices}, Combinatorica \textbf{4} (1984), no.~4, 325--338.

\bibitem[MW90]{MCKAY1990565}
Brendan~D. McKay and Nicholas~C. Wormald, \emph{Asymptotic enumeration by
  degree sequence of graphs of high degree}, European Journal of Combinatorics
  \textbf{11} (1990), no.~6, 565 -- 580.

\bibitem[MX18]{mossel2018seeded}
Elchanan Mossel and Jiaming Xu, \emph{Seeded graph matching via large
  neighborhood statistics}, arXiv preprint arXiv:1807.10262 (2018).

\bibitem[NS09]{narayanan2009anonymizing}
Arvind Narayanan and Vitaly Shmatikov, \emph{De-anonymizing social networks},
  Security and Privacy, 2009 30th IEEE Symposium on, IEEE, 2009, pp.~173--187.

\bibitem[OWWZ14]{o2014hardness}
Ryan O'Donnell, John Wright, Chenggang Wu, and Yuan Zhou, \emph{Hardness of
  robust graph isomorphism, lasserre gaps, and asymmetry of random graphs},
  Proceedings of the twenty-fifth annual ACM-SIAM symposium on Discrete
  algorithms, SIAM, 2014, pp.~1659--1677.

\bibitem[PG11]{pedarsani2011privacy}
Pedram Pedarsani and Matthias Grossglauser, \emph{On the privacy of anonymized
  networks}, Proceedings of the 17th ACM SIGKDD international conference on
  Knowledge discovery and data mining, ACM, 2011, pp.~1235--1243.

\bibitem[SXB08]{singh2008global}
Rohit Singh, Jinbo Xu, and Bonnie Berger, \emph{Global alignment of multiple
  protein interaction networks with application to functional orthology
  detection}, Proceedings of the National Academy of Sciences \textbf{105}
  (2008), no.~35, 12763--12768.

\bibitem[VCP{\etalchar{+}}11]{vogelstein2011large}
Joshua~T Vogelstein, John~M Conroy, Louis~J Podrazik, Steven~G Kratzer, Eric~T
  Harley, Donniell~E Fishkind, R~Jacob Vogelstein, and Carey~E Priebe,
  \emph{Large (brain) graph matching via fast approximate quadratic
  programming}, arXiv preprint arXiv:1112.5507 (2011).

\bibitem[Wor]{wormald1999models}
Nicholas~C Wormald, \emph{Models of random regular graphs}.

\bibitem[Wor80]{wormald1980some}
\bysame, \emph{Some problems in the enumeration of labelled graphs}, Bulletin
  of the Australian Mathematical Society \textbf{21} (1980), no.~1, 159--160.

\bibitem[YG13]{yartseva2013performance}
Lyudmila Yartseva and Matthias Grossglauser, \emph{On the performance of
  percolation graph matching}, Proceedings of the first ACM conference on
  Online social networks, ACM, 2013, pp.~119--130.

\end{thebibliography}
\bibliographystyle{amsalpha}

\end{document}